\documentclass[12pt]{article}
\usepackage{setspace}
\usepackage{tikz}
\usetikzlibrary{positioning,arrows.meta,fit}
\usepackage[utf8]{inputenc}
\usepackage{bbm}
\usepackage{array}
\usepackage{booktabs}
\usepackage{ragged2e}
\usepackage{multirow}
\usepackage{csquotes}
\usepackage{verbatim}
\usepackage[a4paper, left=0.9in,right=0.9in,top=0.5in,bottom=1in]{geometry}
\usepackage[T1]{fontenc}
\usepackage{lmodern}
\usepackage{mathtools}
\usepackage[ruled, vlined, linesnumbered]{algorithm2e}
\usepackage[english]{babel}
\usepackage{amsmath,amsthm,amssymb}
\usepackage{amsmath}
\usepackage{amsfonts}
\usepackage{hyperref}
\usepackage{graphicx}
\usepackage[shortlabels]{enumitem}
\usepackage{amsmath, amsthm, amssymb, amsfonts}
\usepackage{bm}
\usepackage{comment}
\usepackage[
    backend=biber,
    style=numeric,
    sorting=nyt,
    giveninits=true,
    maxbibnames=99,
    doi=false,
    url=false,
    isbn=false,
    eprint=false
]{biblatex}
\usepackage{color}
\usepackage{float}
\usepackage{xcolor}
\usepackage{cancel}
\graphicspath{{/texfiles/}}

\hypersetup{
    colorlinks=true,
    linkcolor=blue,
    filecolor=magenta,      
    urlcolor=cyan,
    citecolor = blue,
}

\hypersetup{
    colorlinks=true,
    linkcolor=blue,
    filecolor=magenta,      
    urlcolor=cyan,
}
\newtheorem{assumption}{Assumption}
\newtheorem{theorem}{Theorem}

\newtheorem{theoremapx}{Theorem}

\newtheorem{proposition}{Proposition}
\newtheorem{corollary}{Corollary}
\newtheorem{defn}{Definition}
\newtheorem{lemma}{Lemma}

\newtheorem{claim}{Claim}

\newcommand{\E}{\mathbb{E}}
\newcommand{\Es}{\mathbb{E}^*}
\newcommand{\pr}{\mathbb{P}}
\newcommand{\ps}{\mathbb{P}^*}
\renewcommand{\P}{\mathcal{P}}
\newcommand{\V}{\mathbb{V}}
\newcommand{\Var}{\mathbb{V}}

\newcommand{\R}{\mathbb{R}}

\newcommand{\U}{\mathring{U}}
\newcommand{\Uh}{\mathring{U}^*_{n,k}}

\DeclareMathOperator{\Cov}{Cov}

\newcommand{\X}{\mathcal{X}}

\newcommand\numberthis{\addtocounter{equation}{1}\tag{\theequation}}

\theoremstyle{definition}

\title{Enhancing Statistical-Computational Balance in Uncertainty Quantification of Bagging: A Cheap Bootstrap Approach}

\author{
Arindam RoyChowdhury\thanks{Department of Industrial Engineering and Operations Research, Columbia University, New York, NY, USA. Email: \texttt{arindam.roychowdhury@columbia.edu}.}
\and
Henry Lam\thanks{Department of Industrial Engineering and Operations Research, Columbia University, New York, NY, USA. Email: \texttt{henry.lam@columbia.edu}.}
}

\date{ }

\begin{document}

\maketitle

\maketitle
\hrule

\begin{abstract}
Bagging, which refers to the repeated subsampling of data and averaging the resulting outputs, is a common approach to reduce sampling variability of unstable learning and data-driven optimization procedures. However, uncertainty quantification for these bagged estimators, namely to generate confidence intervals for their target quantities, faces challenges in the simultaneous attainment of coverage validity while being computationally affordable. The latter refers to the practical consideration of overall resampling effort which, when using conventional bootstraps for instance, would lead to nested resampling procedures and thus can be multiplicatively large. In this paper, we develop an interval construction approach that is both asymptotically coverage-valid and computationally light in that it does not use much more than the resampling effort in the bagging scheme itself. Our approach builds on the recent so-called cheap bootstrap idea, a method designed to substantially reduce the number of resamples based on an integrative treatment between data noises and Monte Carlo resampling noises, and an in-depth derivation of the algorithmic configurations  under which the Monte Carlo error from incomplete bagging is negligible. Our theory shows that our required computation budget is significantly smaller than existing methods in relation to the bag size and number of Monte Carlo subsamples. Numerical experiments in random forests, as well as an optimality bound estimation in stochastic optimization where bagging is useful, show that our proposed intervals achieve reliable coverage with low computation cost.

% , and it is central to methods such as random forests and bagging-based bounds in stochastic optimization. computationally tractable and  require placing an outer resampling loop around the bagging procedure, leading to a costly nested resampling scheme. . The latter refers to the overall resampling effort that, when using conventional bootstraps for instance, would lead to nested resampling procedures and thus can be multiplicatively large. in the sense of not needing a multiplicatively large number of resampling effort due to nested sampling procedures  for the expected output of the base learner, using the recently proposed cheap bootstrap: a
\end{abstract}

\hrule

\section{Introduction}
Bagging, short for bootstrap aggregating, is a general ensemble methodology for reducing the variability of statistical learning and data-driven decision-making procedures \cite{breiman1996bagging}. The method repeatedly resamples or subsamples the data, and applies a so-called base learner to each resample followed by an aggregation of the resulting outputs. Here, a base learner means the original algorithm or statistical procedure that one would have applied to a single dataset before aggregation. For example, in a regression problem, the base learner may be a decision tree regressor that maps the training data to a fitted prediction function; in classification, it may be a tree classifier that assigns class labels. The key motivation is that many useful base learners are sensitive to small perturbations of the training data \cite{buhlmann2002analyzing}. By averaging over many such data perturbations, bagging reduces this sampling variability while preserving the main structure learned by the base procedure. This idea underlies widely used machine learning methods such as random forests \cite{Breiman2001_Random_Forest} as well as the optimality bound estimator for stochastic optimization based on the optimism of empirical solution that is in turn expressible as expectation of resampled optimization problems \cite{lam2018bounding}.

Despite the widespread use of bagging, its uncertainty quantification remains comparatively less settled. Here, uncertainty quantification means the construction of confidence interval for the bagged target, namely the expected value of the bagged estimator. A main challenge in this task comes from the intricate balance between computation need and statistical validity: To attain coverage validity, existing approaches often require a computation cost, in terms of the overall resampling effort, that is much larger than the resampling used to create the bagged estimator itself. To briefly explain this challenge, first note that in practice bagging is implemented using only finitely many subsamples, and thus a coverage-valid interval must account for the additional randomness introduced by this finite-subsampling step. That is, this interval should encompass both data and Monte Carlo noises in the coverage guarantee. Unless we expend significant subsampling effort, these noises are intricately entangled and cause methodological complications. With this, using standard approaches such as the bootstrap to construct intervals would lead to ``nested resampling'', namely each resample estimate in the bootstrap scheme would involve an inner layer of resampling arising from the creation of the bagged estimator itself, thus making the total computation effort multiplicatively large and very expensive.

Due to the above challenges in attaining statistical-computational balance, existing inference approaches for bagging either do not have full coverage guarantees, or require highly expensive resampling effort. Using a normality interval requires a limiting variance that depends on unknown covariance components \cite{mentch_hooker_2016}, whose estimation in turn requires additional subsampling, and the total required computation to construct coverage-valid confidence interval is not explicitly known. The infinitesimal-jackknife (IJ) variance estimation, popular in the random forest literature \cite{IJ_first}, is analyzed in theory assuming either infinite ensembles \cite{athey_RF}, or else again no interval coverage guarantee is known \cite{IJ_CLT}. \cite{lam2018bounding} develop IJ-based confidence bounds for optimality-gap estimation, but its analysis is specialized to the SAA structure in stochastic optimization rather than providing a general inferential framework for bagged estimators. Moreover, the required computational burden is not explicitly linked to the subsample size.
% , leaving room for further improvement.
Finally, classical bootstrap methods \cite{efron1979bootstrap,Efron1994_bootstrap_intro,diciccio_efron_1996,hall1992bootstrap} can in principle be employed. However, existing off-the-shelf bootstrap theory does not directly establish validity for the implemented bagged estimator, particularly when the finite subsampling introduces additional Monte Carlo error. Moreover, as mentioned above, computationally each outer bootstrap replication requires rebuilding a bagged estimator from many new subsamples, leading to a nested procedure whose cost can be much larger than the original bagging computation. Table~\ref{tab:lit_review_summary}, which we will discuss in further detail in the sequel, highlights some of the challenges faced by these methods.

The main goal of this paper is to develop a new approach to construct intervals for bagging that pushes forward the statistical-computational frontier -- We attain provably asymptotic coverage validity, while at the same time requiring a substantially lower computational budget than existing approaches. Roughly speaking, our approach follows the classical bootstrap route, in that we resample the data which is then used to create resampled bag estimates. The discussion above implies this typically leads to nested resampling effort and hence computationally expensive. Our twist comes from two main ingredients: One is to draw upon the idea of so-called cheap bootstrap~\cite{lam2022cheap,lam2023bootstrap,ohlendorff2025cheap}. This is a recent bootstrap implementation approach that utilizes an integrated perspective in handling data noise and resampling noise, instead of separate treatment of these two noises in conventional bootstraps, which in turn allows the construction of pivotal quantities using very few resamples. The cheap bootstrap has been recently applied to the inference of stochastic gradient descent \cite{lam2023resampling}, ensemble neural networks \cite{huang2023efficient} and simulation modeling \cite{lam2022cheap}, and it is particularly well-suited to reduce the inference computation cost for bagging that otherwise requires nested resampling effort. Our second main ingredient is an in-depth analysis on the number of bags needed to ``wash away'' the Monte Carlo noises carried in the confidence interval for bagged estimators. As mentioned above, the latter needs to account for both data noises and Monte Carlo noises in the bagging scheme, and these noises interact intricately. By washing away the Monte Carlo noises, we can focus on only the data noises in which the cheap bootstrap is applicable, so that in total we achieve only a small amplifying factor of the resampling effort used in the original bagged estimator.

The conceptual novelty of our approach is two-fold. First, our approach avoids the variance-estimation route that underlies much of the existing literature. Standard approaches, such as the ones discussed above, typically proceed by first establishing a normal approximation for the ensemble estimator and then using a consistent estimator of the limiting variance. For example, \cite{mentch_hooker_2016} estimates variance through an additional Monte Carlo procedure, while IJ in \cite{IJ_first} works with an idealized variance that can be approximated by reusing the subsamples used to construct the bagged estimator.  Thus, the final validity of these intervals depends not only on the ensemble central limit theorem (CLT), but also on the accuracy of a separate variance-estimation step, whose overall computational requirements are not always made explicit. Moreover, in some cases such as \cite{IJ_first}, establishing simultaneously the ensemble CLT and the validity of the accompanying variance-estimation step, which is needed to achieve coverage validity, can be lacking. Our approach bypasses this consistent variance estimation. Instead of estimating the asymptotic variance, we compare the original bagged estimator against a small number of bootstrap recomputations of the same procedure, and use the original estimator and the few resample estimators to form a pivotal statistic that can be converted into a coverage-valid interval. In particular, this construction uses a $t$-type interval instead of Gaussian, where, intriguingly, the degrees of freedom is the number of resamples rather than the data size as in conventional thinking.

Our second novelty is the explicit treatment of the finite number of bags used in practice. The implemented bagged estimator aggregates a finite number of base learners, each trained on a subsample of the full data, which is an approximation for the fully bagged, or the complete U-statistic. Valid inference must account for the additional Monte Carlo variability introduced by this approximation. Our analysis gives an explicit relation among the number of bags, the subsample size, and the original data size: roughly, the Monte Carlo error from incomplete bagging becomes negligible relative to sampling error when the number of bags grows faster than the ratio of the original data size to the subsample size. With these, our method replaces the usual ``CLT with variance estimation'' pipeline with a direct pivotal construction for the practically relevant finite-ensemble bagging predictor, thus achieving a simultaneous guarantee on both statistical coverage validity and computational cheapness that is superior to existing approaches. 

The rest of the paper is organized as follows. Section~\ref{sec:problem_setup} introduces the general bagging framework and the related notations used throughout the paper. Section~\ref{sec:lit_review} reviews existing inferential approaches for bagged estimators and discusses the associated statistical and computational challenges. Section~\ref{sec:the_CB} presents the cheap bootstrap bagging interval and the corresponding algorithm. Section~\ref{sec:main results} establishes the coverage guarantees, including conditions under which the finite-bagging approximation error is asymptotically negligible, and extends the results to bagging with replacement. Section~\ref{sec:Applications} discusses applications, including random forests and stochastic optimization. Section~\ref{sec:Numerical_Examples} reports numerical experiments, with concluding remarks in Section~\ref{sec:Summary}.

\section{Problem Setup}\label{sec:problem_setup}

To formalize the setup, let \(\Xi\) denote the sample space of observations, such that each data point \(\xi_i \in \Xi\). We represent the base learner by a \emph{symmetric kernel}
\[
    H_k:\Xi^k\to\mathbb R,
\]
which maps a dataset of \(k\) observations to a real-valued output. If \(\xi_1,\ldots,\xi_k\) are independent observations drawn from the underlying distribution \(\P\), we also define $W_k$ to be the expectation of the base learner output:
\begin{equation}\label{eq:Wk_def}
     W_k
    :=
    \mathbb E_{\P}\!\left[H_k(\xi_1,\ldots,\xi_k)\right].
\end{equation}

This abstraction covers a wide range of procedures. For example, in a supervised learning problem, each observation may be \(\xi_i=(X_i,Y_i)\), where \(X_i\) is a covariate vector and \(Y_i\) is a response. Fix a target covariate value \(x_0\) at which prediction is desired. Given \(k\) training observations \(\xi_1,\ldots,\xi_k\), a decision tree algorithm first fits a tree-based prediction rule \(\hat f_k(\cdot;\xi_1,\ldots,\xi_k)\) by recursively partitioning the covariate space and assigning a fitted response value to each terminal leaf. The kernel output is then the prediction of this fitted tree at the fixed target point:
\[
    H_k(\xi_1,\ldots,\xi_k)
    =
    \hat f_k(x_0;\xi_1,\ldots,\xi_k).
\]  

The same framework also encompasses stochastic optimization, where the goal is to choose a decision under uncertainty. Consider the stochastic program
\[
    H^*
    =
    \min_{x\in\mathcal X}
    \mathbb E_{\P}[h(x,\xi)],
\]
where \(x\in\mathcal X\) is the decision variable, $\xi$ is a random object from an unknown distribution \(\P\), and \(h(x,\xi)\) is the cost incurred by decision \(x\) under realization \(\xi\).  A standard approach in estimating $H^*$ is the sample average approximation (SAA), which replaces the unknown expectation by an empirical average \cite{Shapiro2009,Kleywegt2002,Higle1991}. Given \(k\) observations \(\xi_1,\ldots,\xi_k\), the SAA optimal value is
\[
    H_k(\xi_1,\ldots,\xi_k)
    =
    \min_{x\in\mathcal X}
    \frac{1}{k}\sum_{i=1}^k h(x,\xi_i).
\]
Now, a direct application of Jensen's inequality shows that $W_k = \E_{\P}[H_k] \leq H^*$, making \(W_k\) an optimistic lower-bound target for the true optimal value \(H^*\). This observation holds under very general conditions and has been widely used to construct statistical lower bounds for stochastic programs \cite{mak1999monte,shapiro2003monte,glasserman2004monte}. Moreover, under additional regularity conditions, \(W_k\) approaches \(H^*\) as \(k\) grows, so choosing a sufficiently large \(k\) can make this lower-bound target close to the true optimal value.

From this perspective, estimating \(W_k\) requires estimating the average SAA optimal value over repeated samples of size \(k\). Given only one dataset, a natural way to approximate this expectation is to repeatedly draw subsamples of size \(k\) from the observed data, solve the corresponding SAA problem on each subsample, and average the resulting optimal values. Thus, bagging arises quite naturally in the context of stochastic optimization, where the base learner is the SAA optimal value \(H_k\), and the bagged estimator averages this output over many resampled datasets \cite{lam2018bounding}.

\subsection{Bagging Procedure}

We next formalize the bagging construction using the kernel notation introduced above. Given the observed data \(\xi_1,\ldots,\xi_n\), for any $k\leq n$, the ideal version of bagging averages the base output \(H_k\) over every subset of \(k\) observations. For subsampling without replacement, this fully bagged estimator is the \(U\)-statistic:
\begin{equation}\label{eq:U_nk_def}
    U_{n,k}
    :=
    \binom{n}{k}^{-1}
    \sum_{1\le i_1<i_2<\cdots<i_k\le n}
    H_k(\xi_{i_1},\ldots,\xi_{i_k}).
\end{equation}
 Since \(H_k\) is symmetric and the observations are i.i.d., each summand in \eqref{eq:U_nk_def} has expectation \(W_k\). Therefore, for any \(k\), \(\mathbb E_{\P}[U_{n,k}]=W_k\). Thus, \(U_{n,k}\) is also an unbiased estimator for $W_k$.

The complete estimator in \eqref{eq:U_nk_def}, however, is usually not computable in realistic applications. It requires \(\binom{n}{k}\) evaluations of the base procedure \(H_k\), which can be prohibitive even for moderate \(n\) and \(k\). Thus, practical bagging uses only a finite number \(D \ll \binom{n}{k}\) of randomly selected subsamples. Let
\(
    I^{(d)}
    =
    \bigl(i^{(d)}_1,\ldots,i^{(d)}_k\bigr),
\) for $d = 1,...,D$;
denote \(D\) independently drawn \(k\)-tuples of distinct indices from \(\{1,\ldots,n\}\), corresponding to subsamples drawn without replacement. The implemented, or incomplete, bagged estimator is
\begin{equation}\label{eq:U_incomplete_def}
    U_{n,k,D}
    :=
    \frac{1}{D}\sum_{d=1}^D
    H_k\!\bigl(\xi_{i^{(d)}_1},\ldots,\xi_{i^{(d)}_k}\bigr).
\end{equation}
Thus, \(U_{n,k,D}\) is a Monte Carlo approximation to the complete bagged estimator \(U_{n,k}\), or equivalently an incomplete \(U\)-statistic. The parameter \(D\) controls the computational budget: larger \(D\) gives a closer approximation to complete bagging, but requires more evaluations of the base procedure. 

\subsection{Confidence Intervals and Validity Notion}
% \paragraph{Goal: }
Our goal is to construct a confidence interval for the population bagged target \(W_k\), based on the finite-\(D\) bagged estimator \(U_{n,k,D}\). This procedure should account for the randomness introduced by using a finite number of randomly selected subsamples, in addition to the noises coming from the original data.
% while avoiding inference schemes whose computational cost substantially exceeds that of the original bagging computation.

More precisely, for a prescribed confidence level \(1-\alpha\), we seek an interval \(I_{n,k,D}\), constructed from the observed data and the finite-\(D\) bagging procedure, such that
\[
\lim_{n\to\infty}\mathbb{P}\bigl( W_k \in I_{n,k,D} \bigr) = 1-\alpha .
\]
Here, the probability is taken over all sources of randomness involved in constructing the interval, including sampling variability in the observed data, Monte Carlo randomness from the random choice of bagging subsamples, and any additional randomness introduced by the inferential procedure itself.

\paragraph{Notation.}
Throughout the remainder of the paper, we adopt the following notational conventions.  
We observe an i.i.d.\ sample $\xi_1,\ldots,\xi_n$ from an unknown distribution $\P$, which constitutes all available data.  
The estimators $U_{n,k}(\equiv U_{n,k,\infty})$ and $U_{n,k,D}$ depends on this sample. Similarly, $\P_n$ denotes the empirical distribution based on $\xi_1,\ldots,\xi_n$.  Bootstrap resamples are denoted by $\xi_1^*,\ldots,\xi_n^*$, drawn i.i.d.\ from $\P_n$. 

Unless stated otherwise, $\E[\cdot]$ and $\pr(\cdot)$ denote expectation and probability under the true data-generating distribution $\P$, i.e., $\E \equiv \E_{\P}$ and $\pr \equiv \pr_{\P}$. Similarly, $\E^*[\cdot]$ and $\pr^*(\cdot)$ denote expectation and probability under the bootstrap law induced by the empirical distribution, i.e., $\E^* \equiv \E_{\P_n}$ and $\pr^* \equiv \pr_{\P_n}$. In particular, for any statistic $T(\xi_1^*,\ldots,\xi_k^*)$, we use the shorthand
\[
    \E^*\!\left[T\right] 
    \;:=\; 
    \E\!\left[T(\xi_1^*,\ldots,\xi_k^*) \,\middle|\, \P_n\right],
\]
and analogously for $\V^*$. For any other probability measure $Q$, we will explicitly write $\E_Q[\cdot]$ and $\pr_Q(\cdot)$.

\section{Existing Approaches and Challenges}\label{sec:lit_review}

We discuss existing approaches for inference on bagged estimators and the main challenges that arise in applying them. Table~\ref{tab:lit_review_summary} summarizes the comparison across these approaches, including their requirements on the bagging size \(D\), subsample size \(k\), total computational demand, and theoretical guarantee for asymptotic coverage. The remainder of this section elaborates on the existing approaches listed in the table, namely direct CLT-based intervals, infinitesimal jackknife approaches, and classical bootstrap intervals. The cheap bootstrap, which is the focus of this paper, is included in the last row for comparison and is developed next in Section \ref{sec:the_CB}.

\begingroup
\renewcommand{\arraystretch}{1.4}

\begin{table}[H]
\begin{tabular}{|c|cc|c|c|}
\hline
\multirow{2}{*}{\begin{tabular}[c]{@{}c@{}}\\[-1pt]\textbf{Approach}\end{tabular}} & \multicolumn{2}{c|}{\textbf{Sufficient Requirements}}                                                                                                                                                                                    & \multirow{2}{*}{\textbf{\begin{tabular}[c]{@{}c@{}}\\[-10pt]Total\\[-10pt]  Computational\\[-10pt]  Demand\end{tabular}}} & \multirow{2}{*}{\textbf{\begin{tabular}[c]{@{}c@{}}
Theoretical \\[-10pt] Guarantee\\[-10pt]  on\\[-10pt]  Asymptotic \\[-10pt]  Coverage
\end{tabular}}} \\ \cline{2-3}
                                   & \multicolumn{1}{c|}{\textbf{\begin{tabular}[c]{@{}c@{}}Bagging Size\\[-5pt]  (D)\end{tabular}}} & \textbf{\begin{tabular}[c]{@{}c@{}}Subsample \\[-10pt]  Size\\[-5pt]  (k)\end{tabular}}                               &                                                                                                  &                                                                                                                          \\ \hline
Direct CLT                         & \multicolumn{1}{c|}{\begin{tabular}[c]{@{}c@{}}No explicit \\[-10pt] condition\\[-10pt](Except \\[-10pt]$D_n\to\infty$)\end{tabular} }                                                                                  & $o(\sqrt{n})$                                                                                          & \begin{tabular}[c]{@{}c@{}}Unknown\\[-10pt] (requires separate \\[-10pt] variance estimation)\end{tabular}     & \begin{tabular}[c]{@{}c@{}}Yes\\[-5pt](assuming \\[-10pt] consistent \\[-10pt] variance \\[-10pt] estimation)\end{tabular}                                                                                                                       \\ \hline
\multirow{2}{*}{IJ for RF}         & \multicolumn{1}{c|}{$\infty$}                                                                                   & \multirow{2}{*}{\begin{tabular}[c]{@{}c@{}}$\Theta(n^{\beta})$\\ ($\beta_{min}<\beta<1$)\end{tabular}} & $\Theta(n^k)$                                                                                    & Yes                                                                                                                      \\ \cline{2-2} \cline{4-5} 
                                   & \multicolumn{1}{c|}{$\Theta(n)$}                                                                                     &                                                                                                        & $\Theta(n)$                                                                                      & Unknown                                                                                                                  \\ \hline
IJ for SAA                         & \multicolumn{1}{c|}{$\Omega(n)$}                                                                                     & $o(\sqrt{n})$                                                                                          & $\Omega(n)$                                                                                      & Yes                                                                                                                      \\ \hline
Classical Bootstrap                & \multicolumn{1}{c|}{$\infty$}                                                                                 & Unknown                                                                                                & $\Omega(n^{2k})$                                                                                 & \begin{tabular}[c]{@{}c@{}}Unknown\end{tabular}                                                         \\ \hline
Cheap Bootstrap                    & \multicolumn{1}{c|}{$\Omega(n/k)$}                                                                                   & $o(n^{1/4})$                                                                                           & $\Omega(n/k)$                                                                                    & Yes                                                                                                                      \\ \hline
\end{tabular}
\caption{Comparison of inferential approaches for bagged estimators for a given sample size $n$}
\label{tab:lit_review_summary}
\end{table}
\endgroup

 % related resampling-based procedures, discussing how addressing the two requirements simultaneously remains a nontrivial issue
\subsection{Normality Interval with Estimated Variance}

A natural way to construct an interval for \(W_k\), corresponding to the direct CLT approach in Table~\ref{tab:lit_review_summary}, is to establish a central limit theorem for \(U_{n,k,D}\) and then estimate the scale appearing in the limit. In particular, suppose that we could establish
\[
\frac{U_{n,k,D}-W_k}{\sigma}\xrightarrow{d}N(0,1),
\]
for an appropriate standard deviation \(\sigma=\sigma_{n,k,D}\). Additionally, if one can construct an estimator \(\hat{\sigma}\) such that \(\hat{\sigma}/\sigma \xrightarrow{p} 1\), then the usual Gaussian interval
\[
U_{n,k,D} \pm z_{1-\alpha/2}\hat{\sigma}
\]
is asymptotically valid.

\cite{mentch_hooker_2016} follows this general route for subsampled ensemble predictors. They show that an ensemble average over random subsamples can be represented as an incomplete U-statistic, and use this representation to obtain a normal approximation. This gives the asymptotic coverage guarantee indicated in the first row of Table~\ref{tab:lit_review_summary}. The scale \(\sigma\) in their limit depends on covariance components that quantify the dependence between base learners trained on overlapping subsamples, together with the variation of individual base-learner predictions. Since these quantities are unknown, they propose an estimator \(\hat{\sigma}\) based on additional repeated subsampling conditional on fixed observations. Thus, the construction of \(\hat{\sigma}\) introduces additional Monte Carlo budgets beyond the original number \(D\) of bagging replications. Although the consistency of \(\hat{\sigma}\) is argued, the precise computational effort required for this variance estimator to be accurate enough for the final confidence interval is not established, which is reflected by the ``Unknown'' computational-demand entry in Table~\ref{tab:lit_review_summary}.

\subsection{Influence-Function Based Variance Estimation} Another similar route is to estimate the asymptotic variance through influence-function based approaches, such as the infinitesimal jackknife (IJ). This approach is especially prominent in the
random forest literature. Consider the ideal, or
perfectly bagged estimator $U_{n,k,\infty}$, obtained by taking $D\to \infty$ in (\ref{eq:U_incomplete_def}). The infinitesimal jackknife $V_{\mathrm{IJ}}^{\infty}$ was introduced by \cite{IJ_first} as a way to estimate the variance of this ideal estimator $U_{n,k,\infty}$:
\[
    V_{\mathrm{IJ}}^{\infty} \approx \text{Var}(U_{n,k,\infty}).
\]

The difficulty is that \(V_{\mathrm{IJ}}^{\infty}\) is itself an idealized
quantity. Computing it exactly requires access to full realization of the resampling
distribution, which is not available when bagging is implemented with only
finitely many replications. In practice, one computes a finite-\(D\) bagged
estimator and uses the realized bagging replications to form a computable IJ
variance estimate \(\widehat V_{\mathrm{IJ}}^D\). Thus, the practical IJ-based route has two layers:
first, \(V_{\mathrm{IJ}}^{\infty}\) is used as an ideal approximation to
\(\operatorname{Var}(U_{n,k,\infty})\); second, \(\widehat V_{\mathrm{IJ}}^D\)
is used as a Monte Carlo approximation to \(V_{\mathrm{IJ}}^{\infty}\). 

A key contribution of~\cite{IJ_CLT} is to analyze this
second layer of Monte Carlo approximation. They show that, for the standard
jackknife and infinitesimal jackknife variance estimators, controlling the
simulation error in approximating \(V_{\mathrm{IJ}}^{\infty}\) may require as
many as
\(
    D=\Theta(n^{3/2})
\)
bagging replications. To reduce this computational burden, they further propose
bias-corrected jackknife and infinitesimal jackknife estimators, for which the
required number of bagging replications decreases to
\(
    D=\Theta(n).
\)
These two regimes are reflected in the ``IJ for RF'' rows of
Table~\ref{tab:lit_review_summary}, where the computational demand depends on
whether the uncorrected or bias-corrected IJ estimator is used.

\cite{athey_RF} provides a complementary result by establishing
asymptotic normality for random forests. Under suitable conditions, 
their theory yields a Gaussian limit of the form
\[
    \frac{ U_{n,k,\infty}-W_k}
    {\sqrt{\operatorname{Var}(U_{n,k,\infty})}}
    \xrightarrow[]{d} N(0,1),
\]
with the IJ estimator satisfying
\(
    { V_{\mathrm{IJ}}^{\infty}}/
    {\operatorname{Var}( U_{n,k,\infty})}
    \xrightarrow[]{p} 1
\), when $k = \Theta(n^{\beta})$ for some $\beta_{min}<\beta<1$.

Taken together, \cite{IJ_CLT} and \cite{athey_RF}  provide strong ingredients for inference with
bagged predictors: a central limit theorem for an ideal bagged object, and an IJ
estimator for its asymptotic variance. Thus, combining the above results suggest the possibility of having a valid confidence interval of the form:
\[
     U_{n,k,D}
    \pm
    z_{1-\alpha/2}\sqrt{\widehat V_{\mathrm{IJ}}^D}.
\]

However, this literature does not directly provide a unified coverage theorem
for the actually implemented statistic \(U_{n,k,D}\). The available normality results are mainly for the idealized infinite-ensemble object, while the finite-\(D\) analyses focus on the Monte Carlo error in estimating the IJ variance. This distinction is summarized in the two segments of the ``IJ for RF'' row of Table~\ref{tab:lit_review_summary}: depending on the required bagging budget, the computational demand can be characterized, but the end-to-end coverage guarantee for the finite-\(D\) estimator is not available in the low-cost regime.

Thus, although the IJ literature provides important tools for inference, it does
not directly yield the type of end-to-end confidence-interval guarantee we seek
for \(U_{n,k,D}\) under a prescribed computational budget. Such a guarantee
would require jointly controlling the approximation error between the finite and
infinite ensembles and the error in the variance estimator, under a common set of
assumptions and with an explicit growth condition on \(D\). Moreover, existing
results do not give a transparent account of how the required bagging budget
\(D\) should scale with the subsample size \(k\) in order to support valid
coverage.

% Our approach addresses these gaps by applying bootstrap directly to the bagged estimator, rather than using an IJ-based variance estimator. This provides a direct route to coverage-valid inference for \(U_{n,k,D}\). In addition, by separating the incomplete-bagging approximation from the outer resampling step, the analysis treats \(D\) as an explicit computational choice and gives conditions under which the finite-\(D\) error is asymptotically negligible, establishing the desired confidence-interval guarantee and clarifying how the required bagging computation depends on \(k\).

A related line of work is bagging-based inference for stochastic optimization.
\cite{lam2018bounding} studies bagged sample-average approximation
(SAA) optimal values to construct lower confidence bounds for the stochastic
program optimum \(H^* := \min_{x\in \X}\E[h(x,\xi)]\). In their setting,
\[
    W_k
    =
    \mathbb{E}[\hat H_k]
    =
    \mathbb{E}\left[
        \min_{x\in\mathcal X}
        \frac{1}{k}\sum_{i=1}^k h(x,\xi_i)
    \right]
\]
is an optimistic target satisfying \(W_k\leq H^*\), so valid lower confidence
bounds for \(W_k\) also provide valid lower confidence bounds for \(H^*\).
They prove central limit theorems for the corresponding bagged SAA estimators
and combine these results with infinitesimal-jackknife variance consistency to
obtain asymptotically valid coverage for \(W_k\).

This theory, however, is developed for the SAA structure rather than for a
general bagged kernel or base learner. Thus, while it provides strong guarantees
in stochastic optimization, a general coverage theory for computationally
implemented bagged estimators remains less settled. In addition, as indicated in the ``IJ for SAA'' row of Table~\ref{tab:lit_review_summary}, the variance-estimation route requires a large Monte Carlo bagging budget with \(D/n\to\infty\), and this requirement is not explicitly linked to the subsample size \(k\).

\subsection{Classical Bootstrap Inference} Instead of relying on
a variance estimator, we return to the classical bootstrap principle.
 Let
\(\xi_1,\ldots,\xi_n\) be i.i.d. observations from an unknown distribution \(P\),
and let \(P_n\) denote the empirical distribution. For a population functional
\(\theta:=t(P)\), the plug-in estimator is defined as \(\hat\theta_n:=t(P_n)\). The
nonparametric bootstrap approximates the sampling distribution of
\(\hat\theta_n\) by treating \(P_n\) as a surrogate for \(P\): bootstrap samples
are drawn from \(P_n\), the statistic \(t(\cdot)\) is recomputed on each
resampled dataset, and the resulting conditional distribution is used as an
estimate of the sampling distribution of \(\hat\theta_n\)
\cite{efron1979bootstrap,Efron1994_bootstrap_intro}.

This bootstrap approximation can be converted into confidence intervals in several standard ways, as discussed in for example \cite{diciccio_efron_1996}. The simplest construction is the ``normal'', or ``bootstrap-standard'',
interval: the bootstrap samples are used only to estimate the standard error of
\(\hat\theta_n\), and the interval is then formed by centering at
\(\hat\theta_n\) and using a normal critical value. The ``percentile interval'' uses the empirical quantiles of
\(\hat\theta_n^*\) directly as the interval endpoints. More refined constructions, such as ``BCa'' intervals, adjust the
percentile method for bias and skewness through bias-correction and acceleration
terms, while ``ABC'' and ``calibrated intervals'' use analytic or higher-order
corrections to improve coverage accuracy
\cite{hall1992bootstrap}.

Bootstrap confidence intervals are attractive because they avoid the need to derive and estimate a closed-form asymptotic variance. However, their validity is not guaranteed by resampling alone. Standard bootstrap arguments typically require suitable conditions such as asymptotic normality, consistency of the conditional bootstrap distribution, nondegeneracy of the limiting variance, or some other appropriate moment or regularity conditions. 

In the present setting, the statistic of interest is the incomplete bagged estimator $  U_{n,k,D}$ 
which estimates the population bagged target $W_k$.
Thus, an off-the-shelf bootstrap interval for \(U_{n,k,D}\) does not automatically establish coverage for \(W_k\). This is especially nontrivial when \(H_k\) depends on a growing subsample size \(k=k_n\), or when \(H_k\) is nonsmooth because it is defined through an optimization or learning algorithm. Again, the statistic \(U_{n,k,D}\) contains two sources of randomness:
    sampling noise from \(\xi_1,\ldots,\xi_n\), and Monte Carlo noise from using
    only \(D\) bagging replications. Classical bootstrap CI theory usually
    treats the statistic as exactly recomputed and does not directly control this finite-\(D\) error. Thus, while the error is potentially analyzable, as mentioned in Table~\ref{tab:lit_review_summary}, it is currently unknown whether the classical bootstrap literature is directly applicable to bagged estimators.

More importantly, there is a substantial computational limitation. A conventional bootstrap places an outer resampling loop on top of the statistic, which in our case is a bagged estimator. Thus, each bootstrap replication requires recomputing the bagged estimator using \(D\) subsamples. With \(B\) outer bootstrap replications, this leads to roughly \(B\times  D\) base-learner evaluations, where traditional bootstrap literature typically requires $B\to \infty$, making the procedures intractable. Moreover, because \(U_{n,k,D}\) includes Monte Carlo randomness from using only finitely many bagging subsamples, a standard bootstrap implementation does not specify how large \(D\) must be for this additional randomness to be negligible for coverage. In Table~\ref{tab:lit_review_summary}, the
entries \(D=\Omega(n^k)\) and total demand \(\Omega(n^{2k})\) represents the sufficient orders that recover the complete bagging and outer bootstrap behavior.

\section{Our Approach: The Cheap Bootstrap Bagging Interval}\label{sec:the_CB}

The preceding discussion highlights a gap in the literature: there is a need for an inference procedure for bagged estimators that is both statistically valid for the target \(W_k\) and computationally tractable under the finite resampling budgets used in practice.  For this, we build on the cheap bootstrap introduced by~\cite{lam2022cheap}, which was developed precisely for settings where repeatedly recomputing an estimator many times is computationally expensive. The method is still bootstrap-based, but it uses bootstrap recomputations in a different way from the conventional bootstrap.

\subsection{Main Procedure}

We first describe the general cheap bootstrap procedure introduced in \cite{Lam_cheap_bootstrap}, and then specialize it to our general bagged estimator. Let \(\psi(\P)\) denote a target functional of the data-generating distribution \(\P\), and let \(\hat\psi_n=\hat\psi_n(\xi_1,\ldots,\xi_n)\) be an estimator computed from i.i.d.\ observations \(\xi_1,\ldots,\xi_n\sim\P\). As in the ordinary nonparametric bootstrap, the empirical distribution \(\P_n\) is used as a proxy for \(\P\): for \(b=1,\ldots,B\), draw an independent bootstrap sample \(\xi^{*b}_1,\ldots,\xi^{*b}_n\stackrel{\mathrm{i.i.d.}}{\sim}\P_n\), and recompute the estimator on this resampled dataset, denoted by \(\hat\psi_n^{*b}:=\hat\psi_n(\xi^{*b}_1,\ldots,\xi^{*b}_n)\).

The distinction from the conventional bootstrap is that the number of bootstrap resamples \(B\) is not treated as a large Monte Carlo parameter. Instead, the cheap bootstrap allows \(B\) to remain fixed, possibly as small as $B = 1$, and uses the resulting bootstrap recomputations only to estimate the scale of the original estimator. Specifically, the scale estimate is formed by centering the bootstrap estimators at \(\hat\psi_n\):
\begin{equation}\label{eq:S_nb_def}
    S_{n,B}^2
    :=
    \frac{1}{B}\sum_{b=1}^B
    \left(\hat\psi_n^{*b}-\hat\psi_n\right)^2,
\end{equation}
One then outputs the interval:
\begin{equation}\label{eq:cheap_boot_CI_general}
    I_{n,B}(\psi)
    :=
    \Bigl[
    \hat\psi_n \pm t_{B,1-\alpha/2}S_{n,B}
    \Bigr],
\end{equation}
where \(t_{B,1-\alpha/2}\) is the \(1-\alpha/2\) quantile of the Student-\(t\) distribution with \(B\) degrees of freedom.

With the above, we now state the proposed interval for a bagged target. Specifically, the confidence interval construction \eqref{eq:cheap_boot_CI_general} can be generalized to the bagging regime by taking \(\psi(\P)=W_k\) and \(\hat\psi_n=U_{n,k,D}\) (from \eqref{eq:Wk_def} and \eqref{eq:U_incomplete_def}). For each bootstrap sample \(\xi^{*b}_{1:n}\), we recompute the same bagging procedure on the resampled data:
\begin{equation}\label{eq:bagged_boot_def}
    \bigl(U_{n,k,D}\bigr)^{*b}
    :=
    U_{n,k,D}\!\left(\xi^{*b}_{1:n}\right),
    \qquad b=1,\ldots,B.
\end{equation}
That is, within each bootstrap replication, the estimator is again formed by averaging \(D\) evaluations of \(H_k\) over subsamples of size \(k\), using the same bagging mechanism as in the original computation. The inner bagged subsamples are generated independently across bootstrap replications.

The cheap bootstrap scale estimate for the incomplete bagged estimator is therefore
\begin{equation}\label{eq:S_nkDB_def}
    S^2_{n,k,D,B}
    :=
    \frac{1}{B}\sum_{b=1}^B
    \left[
    \bigl(U_{n,k,D}\bigr)^{*b}
    -
    U_{n,k,D}
    \right]^2.
\end{equation}
Thus, the proposed confidence interval for \(W_k\) is given by:
\begin{equation}\label{eq:cheap_boot_CI_Wk}
    I_{n,k,D,B}
    :=
    \Bigl[
    U_{n,k,D}
    \pm
    t_{B,1-\alpha/2}S_{n,k,D,B}
    \Bigr].
\end{equation}
Algorithm~\ref{alg:cheap_bootstrap_bagging} summarizes this construction. 

\begin{algorithm}[t]
\caption{Cheap Bootstrap Bagging Interval for $W_k$}\label{alg:cheap_bootstrap_bagging}
\KwIn{Data $\xi_{1:n}=(\xi_1,\ldots,\xi_n)$; subsample size $k$; number of bagging resamples $D$; number of bootstrap replicates $B$; confidence level $1-\alpha$.}
\KwOut{Cheap bootstrap bagging interval $I_{n,k,D,B}$ for $W_k$.}

\BlankLine
\textbf{Step 0 (Original bagged estimate).}
Compute the bagged estimator $U_{n,k,D}$ from \eqref{eq:U_incomplete_def} using $D$ subsamples of size $k$ drawn without replacement from $\{1,\ldots,n\}$ (reuse across subsamples allowed).\\
\BlankLine

\textbf{Step 1 (Outer bootstrap loop).}
\For{$b=1,\ldots,B$}{
Draw a bootstrap resample $\xi^{*b}_{1:n}$ of size $n$ by sampling with replacement from $\xi_{1:n}$ (equivalently, i.i.d.\ from $\P_n$).\\
Compute the bagged estimate on the resampled data:
\[
\bigl(U_{n,k,D}\bigr)^{*b} := U_{n,k,D}\!\left(\xi^{*b}_{1:n}\right),
\]
using the same bagging mechanism as in Step 0 (i.e., $D$ subsamples of size $k$).
}

\BlankLine
\textbf{Step 2 (Bootstrap scale estimate).}
Compute
\[
S^2_{n,k,D,B}
:=
\frac{1}{B}\sum_{b=1}^B
\left(
\bigl(U_{n,k,D}\bigr)^{*b} - U_{n,k,D}
\right)^2,
\qquad
S_{n,k,D,B}:=\sqrt{S^2_{n,k,D,B}}.
\]

\BlankLine
\textbf{Step 3 (Cheap bootstrap bagging interval).}
Let $t_{B,1-\alpha/2}$ be the $1-\alpha/2$ quantile of the Student-$t$ distribution with $B$ degrees of freedom and return
\[
I_{n,k,D,B}
:=
\Bigl[
\,U_{n,k,D} \pm t_{B,1-\alpha/2}\, S_{n,k,D,B}\,
\Bigr].
\]
\end{algorithm}

\subsection{Explanation and Conceptual Distinctions with Classical Bootstraps}

With the above procedural descriptions, we now explain the underpinning principle of the cheap bootstrap bagging interval and its novelty relative to the conventional bootstrap. To begin with, in the latter approach, validity is typically understood through a two-stage argument. The first stage is the statistical approximation: the conditional distribution of the bootstrap estimator, given the observed data, is used as a surrogate for the sampling distribution of the original estimator. The second stage is the numerical approximation: because this conditional distribution is not available in closed form, its relevant summaries, such as a standard error or quantile, are estimated by Monte Carlo simulation over many bootstrap resamples. Thus, the statistical approximation error and the Monte Carlo error are treated separately. As a consequence, conventional bootstrap intervals generally require a sufficiently large number of bootstrap replications so that the Monte Carlo approximation of the resampling distribution is accurate. When only a small number of bootstrap replications is used, this second approximation can be too noisy, and the usual bootstrap validity argument no longer applies.

The cheap bootstrap takes a different, single-stage perspective. Rather than using many bootstrap resamples to approximate the full conditional resampling distribution, it works with the joint distribution of the original estimator and a fixed number of bootstrap recomputations. Under suitable regularity conditions, these quantities are asymptotically Gaussian and asymptotically independent after normalization. This joint limit makes it possible to construct a Studentized pivotal statistic whose limiting distribution is Student \(t\), with degrees of freedom determined by the number \(B\) of bootstrap recomputations. In this formulation, the statistical fluctuation of the original estimator and the Monte Carlo fluctuation from using only finitely many bootstrap recomputations are handled together, rather than in two separate approximation steps. Consequently, \(B\) need not grow with the sample size; it can remain fixed, and can in principle be as small as \(B=1\). A smaller \(B\) leads to a wider interval through the corresponding \(t\)-critical value, but does not harm the coverage guarantee.

This distinction is important for bagged estimators. Applying a conventional bootstrap to \(U_{n,k,D}\) would require recomputing the incomplete bagged estimator inside each bootstrap resample. Since each such recomputation already involves \(D\) evaluations of the base procedure \(H_k\), the resulting nested resampling scheme can be very costly. The cheap bootstrap approach instead constructs the interval from the original incomplete bagged estimate together with only a small fixed number of bootstrap recomputations of the same estimator, significantly reducing the computational burden.

The following sections establish conditions under which \eqref{eq:cheap_boot_CI_Wk} attains asymptotically nominal coverage for \(W_k\) with fixed \(B\), and characterize growth conditions on \(D\) that make the Monte Carlo error from incomplete bagging asymptotically negligible.

\section{Main Theory}\label{sec:main results}
This section establishes the assumptions and main theoretical results underlying the asymptotic validity of the cheap bootstrap bagging interval. In Section~\ref{sec:main_result_U_statistic}, we first introduce sufficient regularity conditions under which the procedure delivers asymptotically correct \((1-\alpha)\) coverage for the target estimand \(W_k\). Under these conditions, we have the main coverage guarantee for the interval \(I_{n,k,D,B}\) constructed in Algorithm~\ref{alg:cheap_bootstrap_bagging}. We then discuss in Section \ref{sec: main_result_V_statistic} that the same inference principle extends to incomplete V-statistic bagged estimators, corresponding to bagging with replacement. The proof of the main result uses the complete U-statistic as an idealized benchmark. Accordingly, we also collect the supporting results for the complete U-statistic in Section \ref{sec: main_results_complete_U}, including the asymptotic normality and bootstrap validity ingredients needed to transfer the complete-statistic guarantee to the incomplete estimator once the finite-\(D\) Monte Carlo error is controlled. In Section~\ref{sec:Applications}, we verify that the stated regularity conditions are satisfied for several commonly used classes of bagged estimators.

\subsection{Inference for Bagging without Replacement}\label{sec:main_result_U_statistic}

We begin with the following regularity conditions.
\begin{assumption}\label{asm: h_bounded}
   There exists a constant $K_0<\infty$ such that
   \begin{equation}\label{eqn: h_bounded_P}
       |H_k(\xi_1,\ldots,\xi_k)| \le K_0 \qquad \P\text{-a.s.}
   \end{equation}
   Moreover, under the bootstrap law induced by $\P_n$, we have
   \begin{equation}\label{eqn: h_bounded_Pn}
       \pr^*\!\bigl(|H_k(\xi_1^*,\ldots,\xi_k^*)|\le K_0\bigr)=1
       \qquad \P\text{-a.s.}
   \end{equation}
\end{assumption}
The second requirement (\ref{eqn: h_bounded_Pn}) in the above assumption 
does \emph{not} follow automatically from the unconditional bound
(\ref{eqn: h_bounded_P}). It is possible for an adversarially defined kernel $H_k$ to be bounded $\P$-a.s.\ while becoming unbounded on configurations that have $\P$-probability zero but may receive positive probability under $\P_n$ (e.g., a kernel that blows up whenever two of its arguments coincide). Assumption~\ref{asm: h_bounded} rules out such pathologies by requiring that, for $\P$-a.e.\ realized dataset, a bootstrap draw from $\P_n$ falls in the bounded region with conditional probability one.

The next assumption imposes a bounded-differences condition on the kernel.

\begin{assumption}[Bounded sensitivity of the kernel]\label{asm: sup bound on H diff}
There exists a constant $K_1>0$ such that, for every $k$ and every $1 \le c \le k$, the kernel $H_k$ satisfies the following stability property under replacement of its first $c$ arguments. If $\xi_i',\xi_i \sim \P$, then
\[
\pr\!\left(
\bigl|
H_k(\xi_1',\ldots,\xi_c',\xi_{c+1},\ldots,\xi_k)
-
H_k(\xi_1,\ldots,\xi_c,\xi_{c+1},\ldots,\xi_k)
\bigr|
\;\le\;
K_1\,\frac{c}{k}
\right)=1.
\]
Likewise, under the bootstrap law induced by $\P_n$, if $\xi_i^*,\xi_i^{*'} \sim \P_n$, then
\[
\pr^*\!\left(
\bigl|
H_k(\xi_1^{*'},\ldots,\xi_c^{*'},\xi_{c+1}^*,\ldots,\xi_k^*)
-
H_k(\xi_1^*,\ldots,\xi_c^*,\xi_{c+1}^*,\ldots,\xi_k^*)
\bigr|
\;\le\;
K_1\,\frac{c}{k}
\right)=1,
\qquad \P\text{-a.s.}
\]
\end{assumption}

This condition states that replacing up to $c$ arguments of the kernel perturbs its value by at most a constant multiple of $c/k$.  
It is a uniform bounded-difference assumption ensuring that $H_k$ is not overly sensitive to changes in a small subset of its inputs.  
The bootstrap version mirrors the requirement in Assumption~\ref{asm: h_bounded} and is needed to guarantee analogous stability when working under $\P_n$.

The following assumption restricts the growth of the kernel order relative to the sample size.

\begin{assumption}\label{asm: rate_of_k}
There exists a constant $\gamma>0$ such that
$k^{4+\gamma}/n \to 0$ as $n\to\infty$. For notational convenience, this implies the existence of some $0<\epsilon<1$ with $4\epsilon \le \gamma$, for which
\[
    \frac{k^{4(1+\epsilon)}}{n} \;\longrightarrow\; 0 
    \qquad \text{as } n\to\infty.
\]
\end{assumption}

From hereon, we allow the bagging budget to depend on the sample size and write it as \(D_n\). This is useful because the asymptotic validity of the incomplete bagged estimator depends on whether the Monte Carlo error from using finitely many subsamples vanishes sufficiently fast as \(n\) grows. Thus, although \(D\) denotes a fixed computational choice in the finite-sample implementation, in the theoretical analysis we consider a sequence \(D=D_n\). Suppose we construct the \emph{incomplete} U-statistic by evaluating \(H_k\) on only \(D_n\) randomly selected subsamples of size \(k\) and averaging the resulting values.  We make the following assumptions on the growth rate of $D_n$:

\begin{assumption}\label{asm: D_n growth}
We assume that
\[
    \frac{n}{D_n\cdot k^{1-2\epsilon}}  \to 0 \quad\text{ as } n \to \infty,
\]
where $\epsilon$ is as defined in Assumption \ref{asm: rate_of_k}.
\end{assumption}

The role of Assumption~\ref{asm: D_n growth} is to control the additional Monte Carlo error introduced by incomplete bagging. Our analysis will first establish the coverage guarantee for the complete U-statistic \(U_{n,k}\), which represents the idealized fully bagged estimator. We then show that the incomplete estimator \(U_{n,k,D_n}\) is sufficiently close to \(U_{n,k}\) under Assumption~\ref{asm: D_n growth}. Consequently, the incomplete bagged estimator inherits the first-order asymptotic behavior of the complete U-statistic, and the finite-\(D_n\) approximation does not affect the resulting coverage guarantee. 

To state the coverage result, we introduce a first-order component that governs the asymptotic fluctuations of the U-statistic. For the kernel \(H_k\), define:
\begin{equation}\label{eq: Hajek_H_k}
    g_k(\xi)
    :=
    \E\!\left[H_k(\xi,\xi_2,\ldots,\xi_k)\right]-W_k .
\end{equation}

This function is the first-order Hájek projection component of the kernel. It captures the effect of perturbing a single observation and determines the leading variance term in the asymptotic distribution of \(U_{n,k}\). 

From the output of Algorithm \ref{alg:cheap_bootstrap_bagging}, we have $$I_{n,k,D_n,B} = \Bigl[
\,U_{n,k,D_n} \pm t_{B,1-\alpha/2}\, S_{n,k,D_n,B}\,
\Bigr]$$
Additionally, we consider a relaxed version of the above:
\begin{equation}\label{eqn: I and I_tilde}
    \tilde I_{n,k,D_n,B} = \Bigl[
\,U_{n,k,D_n} \pm t_{B,1-\alpha/2}\, S_{n,k,D_n,B}+ o_p\!\left(n^{-1/2}\right)\,
\Bigr]
\end{equation}
where $S_{n,k,D_n,B}$ and $U_{n,k,D_n}$ are as constructed from Algorithm \ref{alg:cheap_bootstrap_bagging}. Note that, and perhaps quite obviously, in practice we would be able to implement $I_{n,k,D_n,B}$ but not $\tilde I_{n,k,D_n,B}$ since the latter has an implicit remainder term. However, $\tilde I_{n,k,D_n,B}$ is a useful artifact to obtain coverage guarantees as we will see momentarily. We now state the main coverage result for the incomplete bagged estimator. The result shows that, under the stated regularity and growth conditions, the cheap bootstrap bagging interval constructed from \(U_{n,k,D_n}\) attains the desired asymptotic coverage for the bagged target \(W_k\).

 % to accommodate for the variance degeneracy case
\begin{theorem}\label{thm: growing_k_main_CI_D_n}
Under the Assumptions \ref{asm: h_bounded}-\ref{asm: D_n growth}, for any $B\geq 1$, we have
\[
    \liminf_{n\to\infty}
    \pr\!\left(W_k \in \tilde I_{n,k,D_n,B}\right)
    \;\ge\; 1-\alpha.
\]
Moreover, if the non-degeneracy condition
$\liminf_{k\to\infty} k^2 V(g_k(\xi))>0$ holds, then
\[
    \lim_{n\to\infty}
    \pr\!\left(W_k \in I_{n,k,D_n,B}\right)
    \;=\; 1-\alpha.
\]
\end{theorem}

Theorem~\ref{thm: growing_k_main_CI_D_n} is the main coverage result for the incomplete bagged estimator. It shows that the cheap bootstrap bagging interval achieves asymptotically valid coverage for the population target \(W_k\), even though the estimator is computed using only \(D_n\) randomly selected subsamples rather than the complete collection of all subsamples. The first statement gives a conservative coverage guarantee without imposing a nondegeneracy condition on the leading variance term. This statement should be interpreted as that, without the nondegeneracy condition, the implementable interval $I_{n,k,D_n,B}$ covers the target $W_k$ up to an $o_p(n^{-1/2})$ error. In particular, this error is either negligible compared to the width of $I_{n,k,D_n,B}$ (if the width is of order $n^{-1/2}$), or else the interval width itself is negligible compared to the usual order $n^{-1/2}$. When the first-order variance is nondegenerate, in the sense that
\(
    \liminf_{k\to\infty} k^2 \V(g_k(\xi))>0,
\)
the implementable interval \(I_{n,k,D_n,B}\) attains exact asymptotic coverage \(1-\alpha\).

\subsection{Inference for Bagging with Replacement}\label{sec: main_result_V_statistic}

Theorem \ref{thm: growing_k_main_CI_D_n} was developed for bagged estimators expressed as U-statistics, where the bagged subsamples are drawn without replacement.
In practice, however, it is often more convenient to work with sampling \emph{with}
replacement, either for algorithmic simplicity or because the bagging procedure itself
is implemented via bootstrap resampling.
This naturally leads to V-statistics, which we introduce and analyze in this section. Recall that the complete U-statistic is associated with the kernel $H_k$ from \eqref{eq:U_nk_def} where each $k$-tuple consists of distinct indices. In contrast, the corresponding V-statistic averages $H_k$ over all ordered $k$-tuples,
allowing repeated indices:
\begin{equation}\label{eq:V_nk_def}
    V_{n,k}
:=
\frac{1}{n^k}
\sum_{i_1=1}^n \cdots \sum_{i_k=1}^n
H_k(\xi_{i_1},\ldots,\xi_{i_k})
\end{equation}
As in the U-statistic case, computing $V_{n,k}$ exactly is infeasible when $n$ and $k$
are large.
Accordingly, we consider an \emph{incomplete} V-statistic obtained by Monte Carlo
averaging over $D_n$ i.i.d.\ $k$-tuples drawn with replacement from $\{1,\ldots,n\}$:
\[
V_{n,k,D_n}
:=
\frac{1}{D_n}
\sum_{d=1}^{D_n}
H_k\!\bigl(\xi_{I^{(d)}_1},\ldots,\xi_{I^{(d)}_k}\bigr),
\qquad
(I^{(d)}_1,\ldots,I^{(d)}_k)
\overset{\text{i.i.d.}}{\sim} \text{Unif}\{1,\ldots,n\}^k .
\]

A key observation underlying our analysis is that, under the regularity conditions
imposed earlier, U- and V-statistics are asymptotically equivalent at the $\sqrt{n}$
scale. 
Formally, under boundedness and stability assumptions on $H_k$ and the growth condition
on $k$, the difference $V_{n,k}-U_{n,k}$ is $o_p(n^{-1/2})$, and the same holds for their
incomplete counterparts provided $D_n$ grows sufficiently fast. As a consequence, the asymptotic distribution of $V_{n,k,D_n}$ coincides with that of
$U_{n,k,D_n}$, and the bootstrap arguments developed for the U-statistics carry over without
modification.

Let $\bigl(V_{n,k,D_n}\bigr)^{*b}$ denote the incomplete V-statistic computed on the $b$th bootstrap
resample, and define the bootstrap variance estimator

\[
\left(S^{V}_{n,k,D_n,B}\right)^2
:=
\frac{1}{B}\sum_{b=1}^B
\left(
\bigl(V_{n,k,D_n}\bigr)^{*b} - V_{n,k,D_n}
\right)^2.
\]

We construct the cheap bootstrap bagging interval
\begin{equation}\label{eq: V_interval}
    I^{V}_{n,k,D_n,B}
:=
\left[
V_{n,k,D_n}
\pm
t_{B,1-\alpha/2}\, S^{V}_{n,k,D_n,B}
\right]
\end{equation}

\noindent with the analogous theoretical relaxation obtained by adding an $o_p(n^{-1/2})$
remainder term, as in the U-statistic case:
\begin{equation}\label{eq: V_tilde_interval}
    \tilde I^{V}_{n,k,D_n,B}
:=
\left[
V_{n,k,D_n}
\pm
t_{B,1-\alpha/2}\, S^{V}_{n,k,D_n,B}
+ o_p\!\left(n^{-1/2}\right)
\right].
\end{equation}

The following theorem establishes the asymptotic validity of the cheap bootstrap bagging interval based on the incomplete V-statistic.

\begin{theorem}\label{thm: V_stat_CI}
Suppose Assumptions \ref{asm: h_bounded}-\ref{asm: D_n growth} hold. Then for any $B\geq 1$, we have:
\[
\liminf_{n\to\infty}
\pr\!\left(W_k \in \tilde I^{V}_{n,k,D_n,B}\right)
\;\ge\; 1-\alpha .
\]
Moreover, if the non-degeneracy condition
$\liminf_{k\to\infty} k^2 V(g_k(\xi))>0$ holds, then
\[
\lim_{n\to\infty}
\pr\!\left(W_k \in I^{V}_{n,k,D_n,B}\right)
\;=\; 1-\alpha .
\]
\end{theorem}
Theorem~\ref{thm: V_stat_CI} shows that replacing the U-statistic by its V-statistic version does not
affect the asymptotic validity of the cheap bootstrap bagging interval.
% From a practical perspective, this provides additional flexibility: bagging procedures
% implemented via sampling with replacement can be analyzed directly, without requiring
% explicit symmetrization or conversion to U-statistics.

\subsection{Development Roadmap of Main Coverage Results}
\label{sec: main_results_complete_U}

This subsection develops the auxiliary results that underlie our main coverage theorems. We first study the complete bagged estimator \(U_{n,k}\),
which corresponds to the idealized case \(D=\infty\), where the base learner \(H_k\) is
averaged over all \(\binom{n}{k}\) subsamples. Although this complete U-statistic is
not the estimator computed in practice, it provides the natural benchmark for the
analysis. Once the cheap bootstrap bagging interval is shown to be valid for the fully bagged
estimator, the remaining task is to control the additional Monte Carlo error introduced
by using only \(D_n\) randomly selected subsamples, which we will elaborate later in this subsection.

The proof strategy for $D = \infty$ relies on the general cheap bootstrap principle from Theorem 1 of \cite{Lam_cheap_bootstrap}, adapted to the
present setting (detailed in Section \ref{sec: Cheap_Bootstrap_Validity} of the Appendix), and also borrowing and generalizing the approaches in deriving central limit results in \cite{lam2018bounding}. The key idea is that in order to justify a cheap bootstrap bagging interval, it is enough to establish two matching asymptotic
approximations. First, the original estimator must satisfy a central limit theorem around
its population target. Second, the bootstrap analogue must satisfy a conditional central
limit theorem, given the data, with an asymptotically equivalent scale. In our setting,
these two requirements amount to proving a CLT for \(U_{n,k}-W_k\) and a corresponding
conditional bootstrap CLT for \(U^*_{n,k}-U_{n,k}\).

The analysis therefore proceeds in three steps. We first establish the asymptotic
normality of the complete U-statistic \(U_{n,k}\) when the kernel order \(k\) is
allowed to grow with the sample size. We then show that the bootstrap complete
U-statistic reproduces the same first-order behavior conditionally on the observed
data. Finally, these two results are combined to obtain the coverage guarantee for the
cheap bootstrap bagging interval centered at \(U_{n,k}\).

 Recall the Hájek projection $g_k(\xi)$ defined in \eqref{eq: Hajek_H_k}. Our next result identifies the leading stochastic fluctuation of the complete bagged
estimator. Under the stated growth condition on \(k\), the Hájek projection governs the
first-order behavior of \(U_{n,k}\), and the remaining higher-order terms are
asymptotically negligible at the relevant scale.

\begin{theorem}
\label{thm: clt_bagging_modified}
Under Assumptions~\ref{asm: h_bounded} and \ref{asm: rate_of_k}, we have
\[
    \sqrt{n}\,(U_{n,k} - W_k)
    \;=\;
    k \sqrt{V(g_k(\xi))}\, Z_{n,k}
    \;+\; o_p\!\left(k^{-\epsilon}\right),
\]
where \(Z_{n,k}\) is a sequence of random variables satisfying
\(
    \E[Z_{n,k}] = 0 \text{ and }
    \Var(Z_{n,k}) = 1.
\)
Moreover, for any subsequence along which \(k^{2+\epsilon} V(g_k(\xi)) > 1\), we have
\[
    Z_{n,k} \;\xrightarrow{d}\; N(0,1).
\]
\end{theorem}

Theorem~\ref{thm: clt_bagging_modified} provides the estimator-side approximation
needed for cheap bootstrap inference. It shows that, along non-degenerate subsequences,
the complete bagged estimator is asymptotically Gaussian with first-order scale
determined by \(k^2 V(g_k(\xi))/n\). The form and argument for Theorem~\ref{thm: clt_bagging_modified} leverage those of Theorem 3 in  
\cite{lam2018bounding}. However, \cite{lam2018bounding} focuses on optimality bounds for stochastic optimization problems and their specific structures, while our Theorem~\ref{thm: clt_bagging_modified} is general.
The next theorem establishes the corresponding
bootstrap approximation. Namely, conditionally on the observed data, the bootstrap
fluctuation \(U^*_{n,k}-U_{n,k}\) has the same limiting Gaussian behavior and the same
asymptotic scale.

\begin{theorem}
\label{thm: growing_k_main_theorem}
Consider data \(\xi_1,\ldots,\xi_n\), and let \(P_n\) denote the empirical distribution.
Let \(\xi_1^*,\ldots,\xi_n^* \sim P_n\) be an i.i.d.\ bootstrap resample. Suppose
Assumptions~\ref{asm: h_bounded}--\ref{asm: rate_of_k} hold. Consider any
subsequence along which \(k^{2+\epsilon} V(g_k(\xi))>1\), where \(\epsilon\) is as in
Assumption~\ref{asm: rate_of_k}. Then, along that subsequence, there exists a
\(P_n\)-measurable \(\sigma^*_{n,k}\) such that
\begin{equation}
\label{eqn: Full_U bootstrap convergence}
    \left|
    \pr^*\!\left(
        \frac{U_{n,k}^* - U_{n,k}}{\sigma^*_{n,k}} \le x
    \right)
    - \Phi(x)
    \right|
    \;\xrightarrow[]{p}\; 0,
\end{equation}
and moreover,
\begin{equation}
\label{eqn: growing_k_main_theorem: variance limit}
    \frac{n\,\sigma^{*2}_{n,k}}{k^2 V(g_k(\xi))}
    \;\xrightarrow[]{p}\; 1.
\end{equation}
\end{theorem}

Together, Theorems~\ref{thm: clt_bagging_modified} and
\ref{thm: growing_k_main_theorem} show that the sampling distribution of the complete
bagged estimator and its conditional bootstrap counterpart agree to first order. This
is precisely the structure needed for the cheap bootstrap. Now, define the scale estimator
\[
    S_{n,k,B}^2
    \;:=\;
    \frac{1}{B}\sum_{b=1}^B
    \left(U^{*b}_{n,k}-U_{n,k}\right)^2 ,
\]
where \(U^{*b}_{n,k}\) denotes the complete U-statistic computed from the \(b\)-th
bootstrap resample. The following corollary converts the preceding estimator and
bootstrap approximations into a coverage statement for \(W_k\).

\begin{corollary}
\label{corr: growing_k_main_CI}
Consider the cheap bootstrap bagging intervals
\[
\tilde I_{n,k,\infty,B}
:=
\left[
U_{n,k} \pm t_{B,1-\alpha/2}\, S_{n,k,B}
+ o_p\!\left(n^{-1/2}\right)
\right],
\qquad
I_{n,k,\infty,B}
:=
\left[
U_{n,k} \pm t_{B,1-\alpha/2}\, S_{n,k,B}
\right].
\]
Under Assumptions~\ref{asm: h_bounded}--\ref{asm: rate_of_k}, we have
\[
    \liminf_{n\to\infty}
    \pr\!\left(W_k \in \tilde I_{n,k,\infty,B}\right)
    \;\ge\; 1-\alpha.
\]
Furthermore, under the non-degeneracy condition
\(
    \liminf_{k\to\infty} k^2 V(g_k(\xi))>0,
\)
we have
\[
    \lim_{n\to\infty}
    \pr\!\left(W_k \in I_{n,k,\infty,B}\right)
    \;=\; 1-\alpha.
\]
\end{corollary}

Corollary~\ref{corr: growing_k_main_CI} establishes the validity of the cheap bootstrap bagging interval in the complete-bagging benchmark. Importantly, the result allows the number of
bootstrap replicates \(B\) to remain fixed; the uncertainty from using only \(B\)
replicates is reflected in the Student-\(t\) critical value with \(B\) degrees of
freedom. The two intervals in Corollary~\ref{corr: growing_k_main_CI} correspond to two possible
asymptotic regimes. The distinction is governed by the variance term
\(k^2 V(g_k(\xi))\), which determines the first-order scale of the complete
U-statistic. If this term is non-degenerate, then the Hájek projection yields a
standard Gaussian first-order limit, and the ordinary interval
\(I_{n,k,\infty,B}\) achieves asymptotically exact \((1-\alpha)\) coverage. If, on the other hand, the first-order variance degenerates, the nominal interval may
collapse at the \(\sqrt n\) scale. The implementable interval $I_{n,k,\infty,B}$ then provides guaranteed valid coverage only up to an additional \(o_p(n^{-1/2})\) error.
% , which
% provides a conservative safeguard in this degenerate regime. 

Finally, we use the above complete-bagging results to obtain the finite-ensemble guarantee in Theorem \ref{thm: growing_k_main_CI_D_n}. Conditional on the data, the finite ensemble averages \(D_n\) independent kernel evaluations whose mean is the complete bagged estimator. Under our assumptions, each evaluation has conditional variance \(O(1/k)\), so the resulting Monte Carlo error has conditional mean square \(O(1/(kD_n))\), with both bounds uniform in the data almost surely. Assumption \ref{asm: D_n growth}, namely \(n/(D_n k^{1-2\epsilon})\to0\), therefore ensures that this approximation error is negligible at the scale required by the coverage argument.

The same error bound applies to each bootstrap ensemble. Since the interval's scale is computed from the root mean square of the differences between the bootstrap estimates and the original estimate, we must control the finite-bagging error in both. With a fixed number of bootstrap replications, these errors are asymptotically negligible in both the interval's center and its bootstrap scale. The coverage conclusions of Corollary \ref{corr: growing_k_main_CI} therefore extend to the finite ensemble: exact asymptotic coverage under the stated non-degeneracy condition, and the relaxed coverage guarantee otherwise. This yields an implementable procedure with an explicit sufficient condition relating the ensemble size, subsample size, and original data size. The detailed arguments are provided in the proofs of Theorems \ref{thm: growing_k_main_CI_D_n} and \ref{thm: V_stat_CI} in Section \ref{sec:proof main} of the Appendix.
% This finite-\(D_n\) approximation is
% handled  (presented in Section \ref{sec:proof main} of the Appendix). Thus, in the reminder of this subsection, we focus on the complete bagged estimator.

\section{Applications}\label{sec:Applications}
In this section, we illustrate the applicability of our cheap bootstrap bagging intervals 
across two representative settings, namely random forests and stochastic optimization that we mentioned in our introduction. 
We examine how the general framework developed in
Section~\ref{sec:main results} applies in these concrete settings under natural and
interpretable conditions.
Together with the preceding theoretical results, we demonstrate how our proposed interval yields asymptotically valid inference for a broad class of bagged estimators. 
% We specifically demonstrate the applicability of this framework to random forests and stochastic optimization  below.

\subsection{Random Forests}
Random forests \cite{Breiman2001_Random_Forest, Tim_RF} provide a canonical example of bagging for machine learning. A decision tree partitions the covariate space into leaves and, for a fixed prediction point \(x\), outputs the average response of the training observations that fall in the leaf containing \(x\). A random forest aggregates many randomized decision trees built on resampled data, thereby stabilizing the highly variable predictions of a single tree. This aggregation viewpoint is central both to their empirical success and to their statistical analysis.

To place random forests within our framework, fix a prediction point \(x\in\mathcal X\), and let \(Z_i=(X_i,Y_i)\) denote the training observations. We write
\begin{equation}\label{eq: Hk_RF}
H_k(x;Z_1,\ldots,Z_k)    
\end{equation}
for the prediction at \(x\) produced by a randomized decision tree trained on a subsample of size \(k\). The associated bagged target is then
\[
W_k(x):=\E\!\left[H_k(x;Z_1,\ldots,Z_k)\right].
\]

It is important to distinguish the target \(W_k(x)\) from the underlying regression function
\[
\mu(x):=\E[Y\mid X=x].
\]
The relationship between \(W_k(x)\) and \(\mu(x)\) is nontrivial, because a random-forest prediction is built from data-adaptive neighborhoods whose geometry depends on the tree-growing mechanism. Establishing consistency and asymptotic normality for \(\mu(x)\) therefore requires additional arguments controlling bias, honesty, and regularity of the tree construction; this is developed in detail, for example, by \cite{athey_RF, IJ_CLT}. In the present paper, we do not pursue that additional layer of analysis. Instead, we work directly with the quantity \(W_k(x)\) and construct confidence intervals for this bagged target within our general cheap bootstrap framework.

To verify the assumptions of Section~\ref{sec:main results}, we impose the following idealized leaf-size condition at the fixed prediction point \(x\).

\begin{assumption}\label{asm:rf_leaf_size}
There exist constants \(M<\infty\) and \(\lambda\in(0,1]\) such that the following hold for every \(k\):
\begin{enumerate}
    \item The response is uniformly bounded almost surely:
    \[
    |Y|\le M \qquad \P\text{-a.s.}
    \]
    \item The tree predictor at \(x\) is a leaf average of the form
    \[
    H_k(x;Z_1,\ldots,Z_k)
    =
    \frac{1}{N_k(x)}
    \sum_{i\in L_k(x)} Y_i,
    \]
    where \(L_k(x)\subseteq \{1,\ldots,k\}\) is the set of training points that fall in the prediction leaf containing \(x\), and \(N_k(x):=|L_k(x)|\).
    \item The prediction leaf size is uniformly nonvanishing:
    \[
    N_k(x)\ge \lambda k
    \qquad \text{almost surely.}
    \]
    \item If the first \(c\) training observations are replaced, yielding a new leaf index set \(L_k'(x)\), then
    \[
    |L_k(x)\,\triangle\,L_k'(x)|\le c
    \qquad \text{almost surely,}
    \]
    where \(\triangle\) denotes symmetric difference.
\end{enumerate}
\end{assumption}

Assumption~\ref{asm:rf_leaf_size} imposes a set of simple structural conditions on the tree predictor at the fixed prediction point \(x\). The first two conditions are general regularity requirements: the response is uniformly bounded, and the prediction is given by an average of responses over the leaf containing \(x\). The third condition requires that this leaf not be too small, specifically that its size is at least a fixed fraction of the subsample size. This prevents the leaf average from becoming overly sensitive to the contribution of a small number of observations. The final condition controls how much the prediction leaf can change when \(c\) training observations are replaced. The symmetric-difference bound means that the new leaf differs from the original one in at most \(c\) indices, so only a small portion of the average can be affected. Together, these conditions ensure both boundedness of the kernel and the \(O(c/k)\) sensitivity property required by our general theory.
For the remainder of this subsection, we fix \(x\) and suppress its dependence, writing \(H_k\) and \(W_k\) for \(H_k(x;\cdot)\) and \(W_k(x)\), respectively.

\begin{proposition}\label{prop:rf_assump23}
Suppose Assumption~\ref{asm:rf_leaf_size} holds. Then, for each fixed prediction point \(x\), the kernel \(H_k\) defined in  (\ref{eq: Hk_RF}) satisfies Assumptions~\ref{asm: h_bounded} and~\ref{asm: sup bound on H diff}.
\end{proposition}

Thus, for the decision tree kernel it follows from the above Proposition \ref{prop:rf_assump23} and Theorem \ref{thm: growing_k_main_CI_D_n} that:

\begin{theorem}\label{thm: RF_main} Suppose the decision tree kernel $H_k$ defined in \eqref{eq: Hk_RF} satisfies Assumption \ref{asm:rf_leaf_size}. Additionally, suppose the subsample size $k$ and the bagging size $D_n$ satisfy Assumptions \ref{asm: rate_of_k} and \ref{asm: D_n growth} respectively. Then, for any $B\geq 1$, we have
\[
    \liminf_{n\to\infty}
    \pr\!\left(W_k \in \tilde I_{n,k,D_n,B}\right)
    \;\ge\; 1-\alpha.
\]
Moreover, if the non-degeneracy condition
$\liminf_{k\to\infty} k^2 V(g_k(\xi))>0$ holds, then
\[
    \lim_{n\to\infty}
    \pr\!\left(W_k \in I_{n,k,D_n,B}\right)
    \;=\; 1-\alpha.
\]
where $W_k = \E[H_k]$ is the population target,   $I_{n,k,D_n,B}$ is the confidence interval obtained from Algorithm \ref{alg:cheap_bootstrap_bagging}, and $\tilde I_{n,k,D_n,B}$ is the relaxed version as defined in \eqref{eqn: I and I_tilde}.
% to accommodate for the degenerate  case when $\liminf_{k\to\infty} k^2 V(g_k(\xi))=0$.
\end{theorem}

\subsection{Optimal Value Estimation in Stochastic Optimization}

The estimation of the optimal value in stochastic optimization provides another important setting in which bagging arises naturally. Consider the problem
\begin{equation}\label{eq:stoch_opt_Hstar}
H^*
\;:=\;
\min_{x\in\mathcal X}
\E_{\P}\!\left[h(x,\xi)\right],
\end{equation}
where \(\mathcal X\) is a feasible decision set and \(\xi\in\Xi\) is a random input with unknown distribution \(\P\). In practice, \(\P\) is not observed directly; instead, one observes an i.i.d.\ sample \(\xi_1,\ldots,\xi_n\sim\P\).

A standard data-driven approximation to \eqref{eq:stoch_opt_Hstar} is the sample average approximation (SAA) objective value
\begin{equation}\label{eq:SAA_Hn}
H_n
\;:=\;
\min_{x\in\mathcal X}
\frac{1}{n}\sum_{i=1}^n h(x,\xi_i),
\end{equation}
which is computed from the observed data. A basic and well-known property of \(H_n\) is its optimistic bias:
\begin{equation}\label{eq:jensen_H}
W_n
\;:=\;
\E[H_n]
\;\le\;
H^*,
\end{equation}
where the expectation is taken with respect to the sampling distribution of the data. This inequality follows from Jensen's inequality and has long been used in stochastic programming and simulation optimization to construct lower bounds on the true optimal value \cite{mak1999monte,shapiro2003monte,glasserman2004monte}. Moreover, under standard regularity conditions, \(W_n\to H^*\) as \(n\to\infty\).

The inequality \eqref{eq:jensen_H} suggests reducing inference for \(H^*\) to inference for the mean of SAA objective values. This motivates resampling-based approaches that repeatedly solve smaller SAA problems and aggregate the resulting objective values. This approach can be contrasted to batching \cite{mak1999monte} that partitions the dataset into disjoint groups and solves one SAA problem per batch, and has been shown to provide a more efficient alternative by allowing data reuse and hence tighter confidence bounds \cite{lam2018bounding} for certain classes of optimization problems. 

% While simple, batching prevents reuse of observations across groups and can therefore be statistically inefficient.

% Bagging provides a more efficient alternative. Rather than partitioning the data once, it repeatedly forms subsamples and aggregates the resulting SAA objective values. This preserves the lower-bound structure in \eqref{eq:jensen_H} while allowing substantial data reuse, and has been shown to yield tighter bounds in stochastic-optimization settings \cite{lam2018bounding}.

Let \(k\) denote a subsample size. Define the SAA kernel
\begin{equation}\label{eq:SO_kernel_Hk}
H_{k}(\xi_1,\ldots,\xi_{k})
\;:=\;
\min_{x\in\mathcal X}
\frac{1}{k}\sum_{i=1}^{k} h(x,\xi_i),
\end{equation}
and the corresponding population target
\[
W_{k}
\;:=\;
\E\!\left[H_{k}(\xi_1,\ldots,\xi_k)\right].
\]
By construction, \(W_{k}\le H^*\) for every \(k\), and if \(k\to\infty\), this lower bound tightens asymptotically. We now verify that this stochastic-optimization setting falls within the scope of the general theory developed in Section~\ref{sec:main results}.

\begin{assumption}\label{asm: h_bounded_SO}
The loss function \(h(x,\xi)\) is uniformly bounded almost surely; that is, there exists \(M<\infty\) such that
\[
\sup_{x\in\mathcal X}|h(x,\xi)| \le M
\qquad \P\text{-a.s.}
\]
\end{assumption}

Under this condition, the induced SAA kernel satisfies the boundedness and stability requirements used in the main theory:
\begin{proposition}\label{prop: SO_bagging_validity}
Suppose Assumption~\ref{asm: h_bounded_SO} holds. Then the kernel \(H_k\) defined in \eqref{eq:SO_kernel_Hk} satisfies Assumptions~\ref{asm: h_bounded} and \ref{asm: sup bound on H diff}.
\end{proposition}
Thus, for the bagging kernel it follows from the above Proposition \ref{prop: SO_bagging_validity} and Theorem \ref{thm: growing_k_main_CI_D_n} that:

\begin{theorem}\label{thm: SAA_main} Suppose the SAA kernel $H_k$ defined in \eqref{eq:SO_kernel_Hk} satisfies Assumption \ref{asm: h_bounded_SO}. Additionally, suppose the subsample size $k$ and the bagging size $D_n$ satisfy Assumptions \ref{asm: rate_of_k} and \ref{asm: D_n growth} respectively. 
Then, for any $B\geq 1$, we have
\[
    \liminf_{n\to\infty}
    \pr\!\left(W_k \in \tilde I_{n,k,D_n,B}\right)
    \;\ge\; 1-\alpha.
\]
Moreover, if the non-degeneracy condition
$\liminf_{k\to\infty} k^2 V(g_k(\xi))>0$ holds, then
\[
    \lim_{n\to\infty}
    \pr\!\left(W_k \in I_{n,k,D_n,B}\right)
    \;=\; 1-\alpha.
\]
where $W_k = \E[H_k]$ is the population target,   $I_{n,k,D_n,B}$ is the confidence interval obtained from Algorithm \ref{alg:cheap_bootstrap_bagging}, and $\tilde I_{n,k,D_n,B}$ is the relaxed version as defined in \eqref{eqn: I and I_tilde}.
% to accommodate for the degenerate  case when $\liminf_{k\to\infty} k^2 V(g_k(\xi))=0$.
\end{theorem}

\section{Numerical Examples}\label{sec:Numerical_Examples}
We consider two numerical examples that illustrate the proposed procedure in two different settings. 
Section~\ref{sec: decision_tree} considers a random forest regression problem, where each observation is a covariate--response pair and the kernel \(H_k\) is the prediction at a fixed test point from a decision tree trained on \(k\) observations. 
Section~\ref{sec: SLO} considers a stochastic linear optimization problem over a simplex, where the kernel \(H_k\) is the sample-average-approximation (SAA) optimal value computed from \(k\) observations.

In both examples, the underlying data-generating distribution and the symmetric kernel \(H_k\) together determine the population target
\[
    W_k=\E[H_k(\xi_1,\ldots,\xi_k)].
\]
Our objective is to evaluate the finite-sample performance of the proposed cheap bootstrap bagging intervals for this target.

We use a common simulation framework across the two examples. We vary the sample size \(n\), the subsample size \(k\), the number of bagging draws \(D\), and the number of outer bootstrap resamples \(B\). We parameterize the bagging budget as
\[
    D=n\times D_{\mathrm{mult}},
\]
where \(D_{\mathrm{mult}}\) controls the computational effort devoted to incomplete bagging. For each parameter configuration, we compute the bagged estimator, construct the corresponding cheap bootstrap bagging interval, and record its coverage and length.

Since \(W_k\) is generally unavailable in closed form, we approximate it using a sufficiently large number of independent Monte Carlo evaluations of \(H_k\). Empirical coverage probabilities and average interval lengths are then estimated by independently repeating each simulation setting. The specific data-generating distribution, kernel, parameter grid, and simulation details are described separately for the random forest example in Section~\ref{sec: decision_tree} and the stochastic optimization example in Section~\ref{sec: SLO}.

\subsection{Random Forest Regression}\label{sec: decision_tree}
We consider a random forest regression example. In this example, each observation is
\(
    \xi=(X,Y),
\)
where \(X\in[0,1]^5\) and \(Y\in\mathbb R\). The covariates have independent uniform coordinates,
\[
    X_i \sim \mathrm{Unif}(0,1), \qquad i=1,\ldots,5,
\]
and the response is generated according to
\begin{equation}\label{eq:decision_tree_dgp}
    Y
    =
    \sin(2\pi X_1)
    +
    X_2^2
    +
    0.5X_3
    -
    X_4
    +
    \varepsilon,
    \qquad
    \varepsilon\sim N(0,0.5^2).
\end{equation}
Thus, the conditional mean function is
\[
    m(x)
    =
    \E[Y\mid X=x]
    =
    \sin(2\pi x_1)
    +
    x_2^2
    +
    0.5x_3
    -
    x_4 .
\]
We evaluate the fitted decision tree at the fixed test point $x_0=(0.5,0.5,0.5,0.5,0.5).$
At this point, the true regression value is \(m(x_0)=0\). However, the target of our inference procedure is not the true regression function \(m(x_0)\), but rather the population bagged target induced by the decision-tree learning rule.

For a subsample of size \(k\), we define the kernel \(H_k\) as follows. Given
\(
    \xi^{(1)},\ldots,\xi^{(k)}
    =
    (X^{(1)},Y^{(1)}),\ldots,(X^{(k)},Y^{(k)}),
\)
we fit a decision tree regressor to the \(k\) observations and return its prediction at \(x_0\). More precisely,
\[
    H_k(\xi^{(1)},\ldots,\xi^{(k)})
    =
    \widehat m_k(x_0),
\]
where \(\widehat m_k\) is the decision tree fitted on the subsample. In the implementation, the maximum depth is fixed at \(4\), and the minimum number of observations required to split an internal node is set to $\lceil k/10\rceil$. The target estimand is
\[
    W_k
    =
    \E\left[
    H_k(\xi_1,\ldots,\xi_k)
    \right]
    =
    \E\left[
    \widehat m_k(x_0)
    \right],
\]
where the expectation is over an independent training subsample of size \(k\). This quantity is the expected prediction of a decision tree trained on \(k\) observations and evaluated at \(x_0\). The unknown $W_k$ is approximated numerically using \(3.2\times 10^5\) independent Monte Carlo evaluations of \(H_k\). For this example, we consider sample sizes up to \(n=4000\), and \(D_{\mathrm{mult}}\) ranges between \(0.01\) to \(1.1\). Each simulation setting is repeated 480 times to estimate the empirical coverage probability and the average interval length.

% The main qualitative patterns are similar to those observed in the previous examples. 
Table~\ref{tab:DT_plot_summary} summarizes the numerical comparisons for the random forest regression problem in \eqref{eq:decision_tree_dgp}. First, Figure~\ref{fig:DT_Btrend_U} studies the effect of the number of outer cheap bootstrap resamples \(B\). Empirical coverage approaches the nominal \(95\%\) level for sufficiently large \(n\), including small values of \(B\). Increasing \(B\) reduces the variability of the bootstrap variance estimate and tends to shorten the resulting intervals.

\renewcommand{\arraystretch}{1.2}
\begin{table}[t]
\centering
\small
\begin{tabular}{
>{\centering\arraybackslash}m{0.15\textwidth}
>{\centering\arraybackslash}m{0.11\textwidth}
>{\centering\arraybackslash}m{0.22\textwidth}
>{}m{0.40\textwidth}
}
\toprule
Figure & Trend on & Fixed parameters & \multicolumn{1}{c}{Comments} \\
\midrule
Figure~\ref{fig:DT_Btrend_U}
& \(B\)
& \shortstack{\(\text{stat}=U\)\\[7pt] \(D_{\mathrm{mult}}=1.1\)}
& Coverage approaches \(95\%\) as \(n\) grows, even for small \(B\). Interval length decreases mainly when \(B\) is increased from very small values; the additional gain is modest after a few bootstrap resamples. \\
\midrule

Figure~\ref{fig:DT_Dtrend_U}
& \(D_{\mathrm{mult}}\)
& \shortstack{\(\text{stat}=U\)\\[7pt] \(B=5\)}
& Very small \(D_{\mathrm{mult}}\) gives wider and more conservative intervals. Once \(D_{\mathrm{mult}}\) is moderately large, interval lengths become much smaller and further increases in \(D_{\mathrm{mult}}\) have limited effect. \\
\midrule

Figure~\ref{fig:DT_ktrend}
& \(k, D_{\mathrm{mult}}\)
& \shortstack{\(n=4000\)\\[7pt] \(B=5\)\\[7pt] \(\text{stat}=U\)}
& Coverage remains close to \(95\%\) for most configurations. The smallest bagging budget gives wider intervals, while the dependence on \(k\) is nonmonotone because \(k\) is varied with \(n\) fixed rather than in the ideal growing-\(n,k\) regime. \\
\midrule

Figure~\ref{fig:DT_stat_trend}
& stat
& \shortstack{\(B=5\)\\[7pt] \(D_{\mathrm{mult}}=1.1\)}
& The incomplete U- and V-statistic implementations have similar coverage and interval length for this decision-tree regression problem. \\
\bottomrule
\end{tabular}
\caption{Summary of the numerical plots for the random forest regression problem in \eqref{eq:decision_tree_dgp}.}
\label{tab:DT_plot_summary}
\end{table}

Figure~\ref{fig:DT_Dtrend_U} studies the effect of the bagging budget \(D_{\mathrm{mult}}\). When \(D_{\mathrm{mult}}\) is very small, the implemented bagged estimator contains substantial incomplete-bagging Monte Carlo error. This additional computational noise leads to wider intervals and can make the empirical coverage conservative. As \(D_{\mathrm{mult}}\) increases, the interval length decreases until the contribution of bagging Monte Carlo error becomes relatively small.

Figure~\ref{fig:DT_ktrend} examines the effect of the subsample size \(k\) in the decision-tree example, fixing \(n=4000\). The coverage remains close to the nominal \(95\%\) level across most configurations, and the intervals are substantially shorter once \(D_{\mathrm{mult}}\) is moderately large. The smallest bagging budget, \(D_{\mathrm{mult}}=0.01\), produces wider and more conservative intervals, which is consistent with the additional Monte Carlo variation from incomplete bagging.

The dependence on \(k\) is not monotone, especially for the average interval length. Here \(k\) is increased while \(n\) is fixed, whereas the theory requires \(n\) to grow with \(k\) under suitable rate conditions. Thus, for large \(k\), the experiment begins to leave the ideal theoretical regime. The main empirical message is that, for a fixed sample size, stable performance is obtained when \(Dk\) is sufficiently large and \(k\) is not too large relative to \(n\).

Finally, Figure~\ref{fig:DT_stat_trend} compares the incomplete U- and V-statistic implementations. The U-statistic version samples subsamples without replacement, while the V-statistic version samples with replacement. In this decision-tree example, the two constructions produce similar empirical behavior, indicating that the choice between sampling with and without replacement has limited impact on the reported finite-sample performance for this setting.

\begin{figure}[t]
    \centering
    \includegraphics[width=\textwidth]{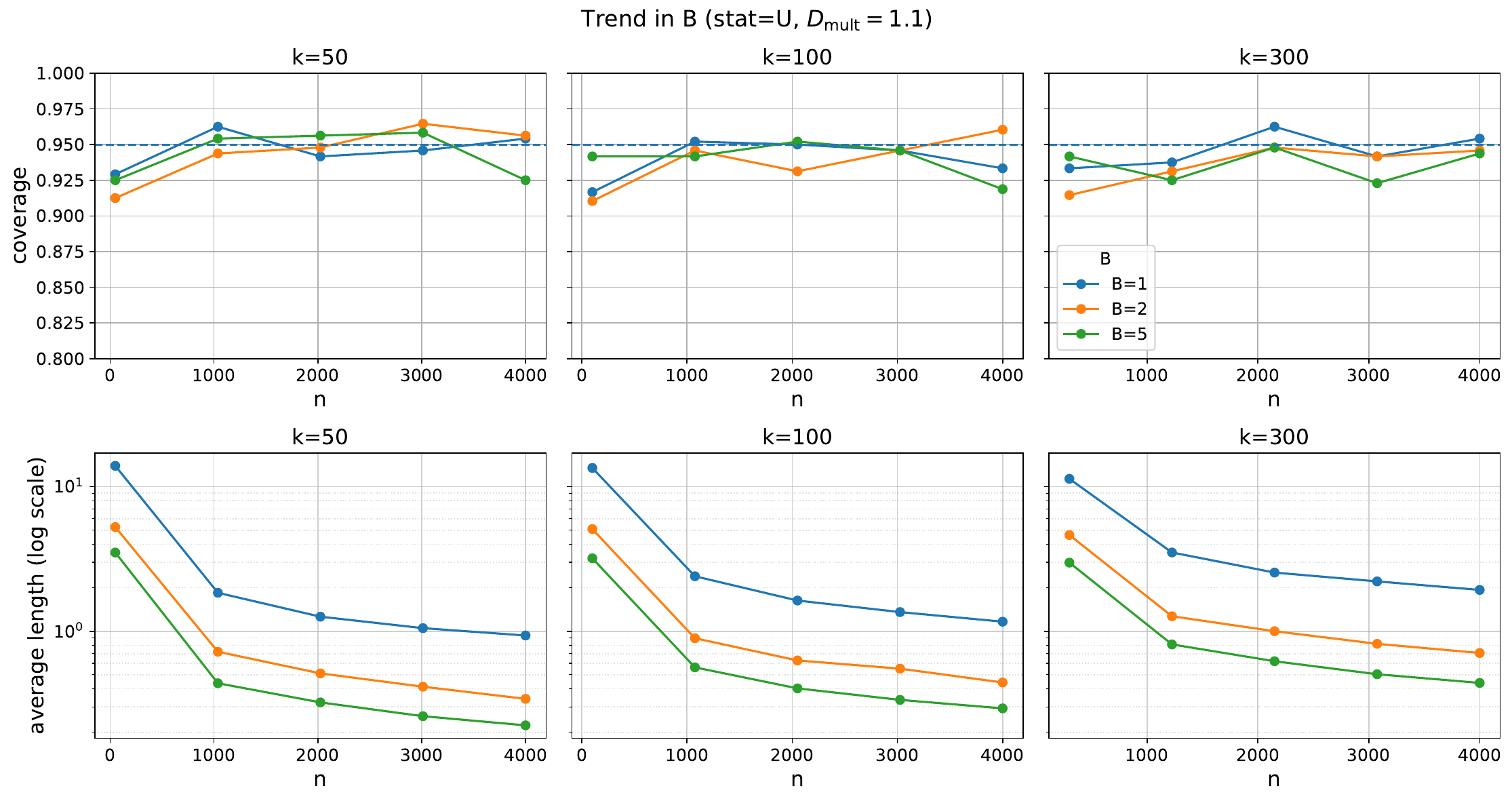}
    \caption{Trend in the number of outer bootstrap resamples \(B\) for the random forest regression problem in \eqref{eq:decision_tree_dgp}, using the incomplete U-statistic estimator and fixing \(D_{\mathrm{mult}}=1.1\). The top row reports empirical coverage, and the bottom row reports average interval length on a log scale.}
    \label{fig:DT_Btrend_U}
\end{figure}

\begin{figure}[t]
    \centering
    \includegraphics[width=\textwidth]{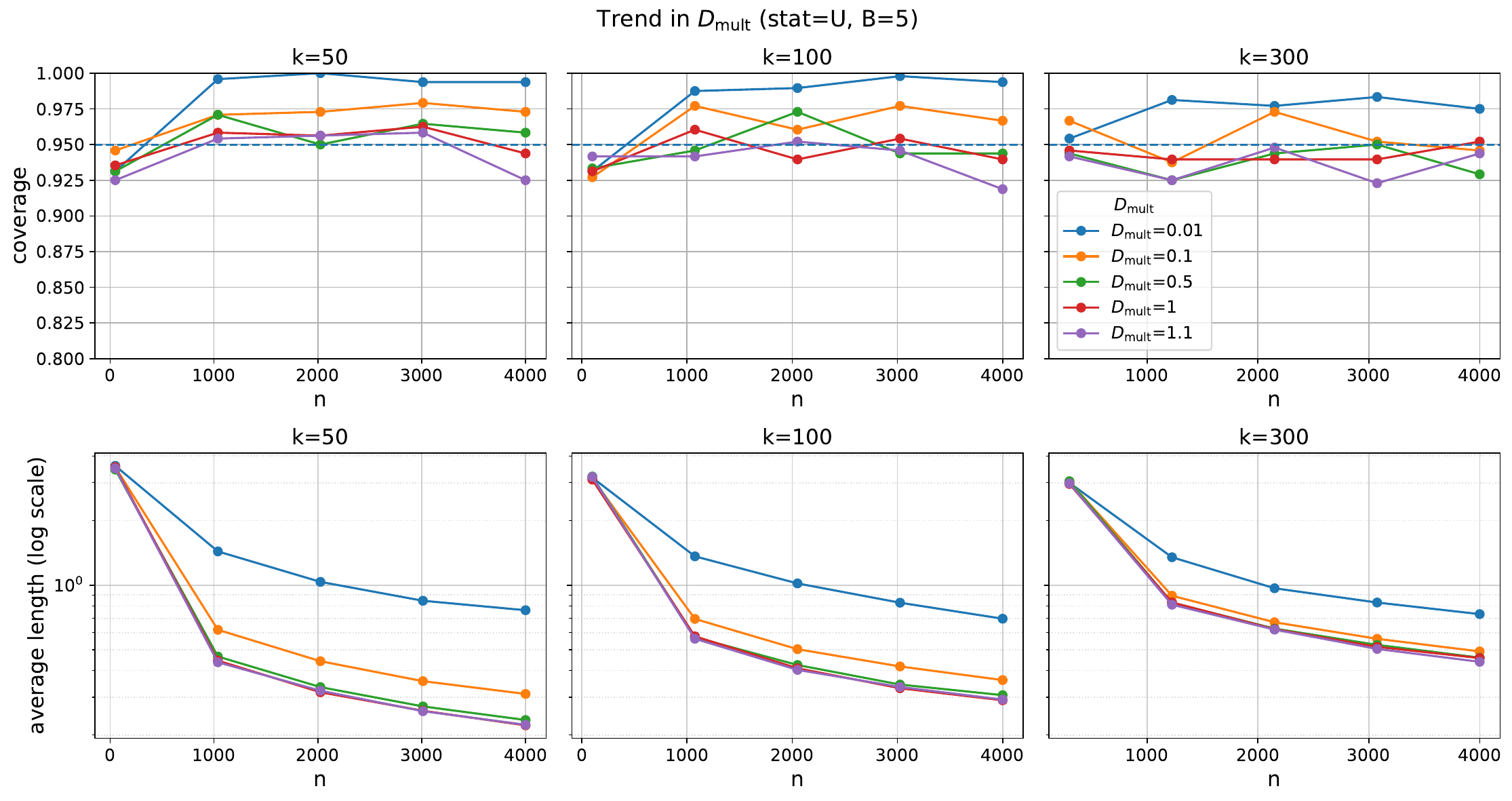}
    \caption{Trend in the bagging budget \(D_{\mathrm{mult}}\) for the random forest regression problem in \eqref{eq:decision_tree_dgp}, using the incomplete U-statistic estimator and fixing \(B=5\). The top row reports empirical coverage, and the bottom row reports average interval length on a log scale.}
    \label{fig:DT_Dtrend_U}
\end{figure}

\begin{figure}[t]
    \centering
    \includegraphics[width=\textwidth]{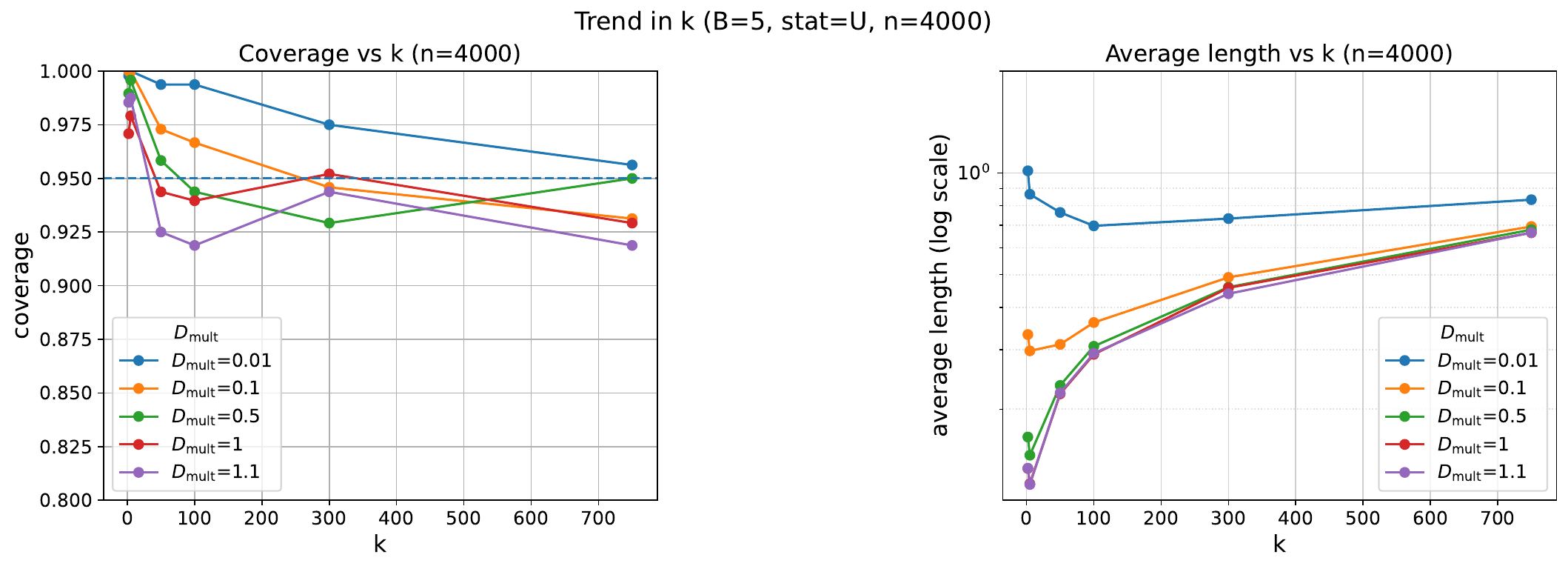}
    \caption{Trend in the subsample size \(k\) for various choices of \(D_{\mathrm{mult}}\) in the random forest regression problem in \eqref{eq:decision_tree_dgp}, fixing \(n=4000\), \(B=5\), and stat \(=U\). The left panel reports empirical coverage, and the right panel reports average interval length on a log scale.}
    \label{fig:DT_ktrend}
\end{figure}

\begin{figure}[t]
    \centering
    \includegraphics[width=\textwidth]{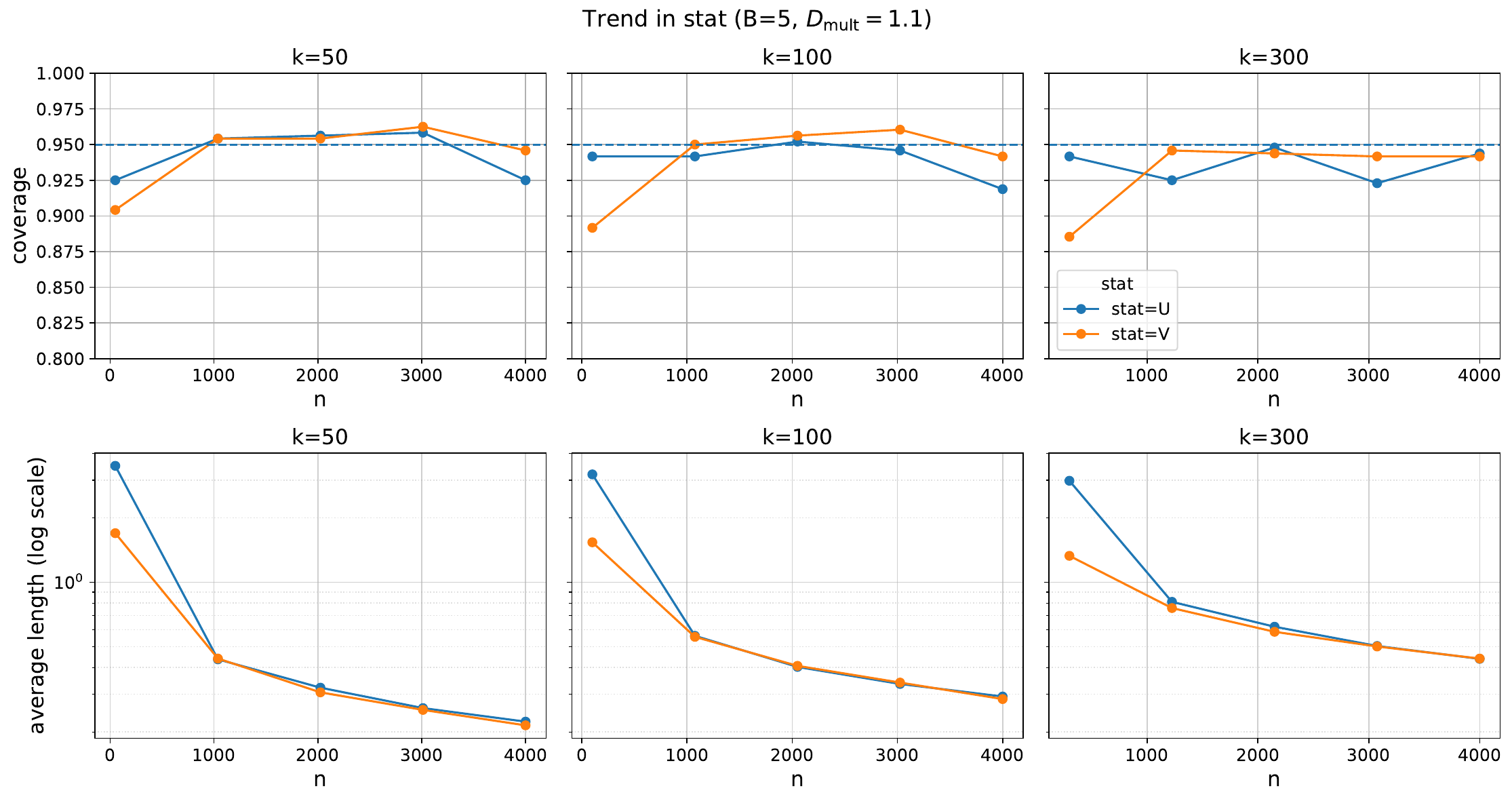}
    \caption{Comparison of the incomplete U- and V-statistic estimators for the random forest regression problem in \eqref{eq:decision_tree_dgp}, with \(B\) and \(D_{\mathrm{mult}}\) fixed. The top row reports empirical coverage, and the bottom row reports average interval length on a log scale.}
    \label{fig:DT_stat_trend}
\end{figure}

\subsection{Stochastic Optimization: Linear Program Over a Simplex}\label{sec: SLO}

We consider the stochastic linear optimization problem
\begin{equation}\label{eq:SLO_problem}
\begin{aligned}
    H^*
    =
    \min_{x\in\mathbb R^{10}} \quad
    & \E[\xi^T x] \\
    \text{s.t.} \quad
    & \sum_{i=1}^{10} x_i = 1, \\
    & x_i \geq 0, \qquad i=1,\ldots,10 .
\end{aligned}
\end{equation}
Here \(\xi=(\xi_1,\ldots,\xi_{10})\), where the coordinates are independent normal random variables satisfying
\[
    \xi_i \sim N(0,1), \qquad i\leq 5,
    \qquad\text{and}\qquad
    \xi_i \sim N(0.1,1), \qquad i\geq 6.
\]
Since the first five coordinates have mean zero and the last five coordinates have mean \(0.1\), every feasible solution satisfying \(x_i=0\) for all \(i\geq 6\) is population-optimal. Thus, problem~\eqref{eq:SLO_problem} has multiple optimal solutions. This makes it a useful stochastic optimization example for studying the behavior of bagging-based inference, since the SAA optimal value can be sensitive to sampling noise among several nearly competitive coordinates.

For a subsample of size \(k\), we define the SAA-value kernel
\[
    H_k(\xi^{(1)},\ldots,\xi^{(k)})
    =
    \min_{x\in\Delta_{10}}
    \frac{1}{k}\sum_{j=1}^k (\xi^{(j)})^T x,
\]
where
\[
    \Delta_{10}
    =
    \left\{
    x\in\mathbb R^{10}:
    \sum_{i=1}^{10}x_i=1,\ x_i\geq 0
    \right\}.
\]
Equivalently, because the feasible region is the simplex,
\[
    H_k(\xi^{(1)},\ldots,\xi^{(k)})
    =
    \min_{1\leq i\leq 10}
    \frac{1}{k}\sum_{j=1}^k \xi_i^{(j)} .
\]
The target estimand is therefore
\[
    W_k
    =
    \E\left[
    \min_{1\leq i\leq 10}
    \frac{1}{k}\sum_{j=1}^k \xi_i^{(j)}
    \right].
\]
Unlike the population optimal value \(H^*=0\), the quantity \(W_k\) is the expected SAA optimal value based on a subsample of size \(k\). It is generally below the population optimum because of the optimistic bias induced by minimizing over noisy sample averages. As in the other numerical examples, we approximate \(W_k\) using independent Monte Carlo evaluations of the kernel.

For this example, we consider sample sizes up to \(n=5000\). For each configuration, the bagging budget is set as $D = n \times D_{\mathrm{mult}}$, where \(D_{\mathrm{mult}}\) ranges from \(0.01\) to \(1.1\). Thus, the number of bagging draws \(D\) ranges from approximately \(1\%\) to \(110\%\) of the sample size. For each choice of \(n\), \(k\), \(D_{\mathrm{mult}}\), \(B\), and bagging statistic, we construct the bagged estimator, apply the cheap bootstrap procedure, and record the resulting confidence interval. The target \(W_k\) is approximated numerically using \(3.2\times 10^6\) independent Monte Carlo evaluations of \(H_k\). Each simulation setting is repeated \(800\) times to estimate the empirical coverage probability and the average interval length.

Table~\ref{tab:SLO_plot_summary} summarizes the numerical comparisons for problem~\eqref{eq:SLO_problem}. The overall behavior is consistent across the different parameter trends. First, Figure~\ref{fig:SLO_Btrend_U} shows that empirical coverage approaches the nominal \(95\%\) level for sufficiently large \(n\), including when the number of outer cheap bootstrap resamples is as small as \(B=1\). Increasing \(B\) tends to reduce the average interval length, while the coverage remains broadly stable.

\renewcommand{\arraystretch}{1.2}
\begin{table}[t]
\centering
\small
\begin{tabular}{
>{\centering\arraybackslash}m{0.15\textwidth}
>{\centering\arraybackslash}m{0.11\textwidth}
>{\centering\arraybackslash}m{0.22\textwidth}
>{}m{0.40\textwidth}
}
\toprule
Figure & Trend on & Fixed parameters & \multicolumn{1}{c}{Comments} \\
\midrule
Figure~\ref{fig:SLO_Btrend_U}
& \(B\)
& \shortstack{\(\text{stat}=U\)\\[7pt] \(D_{\mathrm{mult}}=1.1\)}
& Coverage reaches the nominal \(95\%\) level for large \(n\), including at \(B=1\). Interval length decreases as \(B\) increases. \\
\midrule

Figure~\ref{fig:SLO_Dtrend_U}
& \(D_{\mathrm{mult}}\)
& \shortstack{\(\text{stat}=U\)\\[7pt] \(B=5\)}
& Very small \(D_{\mathrm{mult}}\) gives wider intervals and conservative coverage because of additional incomplete-bagging Monte Carlo error. Larger \(D_{\mathrm{mult}}\) reduces interval length until the Monte-Carlo error from bagging is negligible.  \\
\midrule
Figure~\ref{fig:SLO_ktrend}
& \(k, D_{\mathrm{mult}}\)
& \shortstack{\(n=5000\)\\[7pt] \(B=2\)\\[7pt] \(\text{stat}=U\)}
& For small \(Dk\), intervals are wider and coverage is conservative. Coverage near \(95\%\) is observed once \(Dk\gg n\). \\
\midrule
Figure~\ref{fig:SLO_stat_trend}
& stat
& \shortstack{\(B=5\)\\[7pt] \(D_{\mathrm{mult}}=1.1\)}
& The incomplete U- and V-statistic constructions have very similar empirical behavior. \\
\bottomrule
\end{tabular}
\caption{Summary of the numerical plots for the stochastic linear optimization problem in \eqref{eq:SLO_problem}.}
\label{tab:SLO_plot_summary}
\end{table}

Figure~\ref{fig:SLO_Dtrend_U} studies the effect of the bagging budget \(D_{\mathrm{mult}}\). For very small values of \(D_{\mathrm{mult}}\), the intervals are wider and the empirical coverage can be conservative. This behavior is expected because a small number of bagging draws introduces additional randomization error into the implemented estimator. Increasing \(D_{\mathrm{mult}}\) initially reduces interval length, with the improvement becoming less pronounced once the incomplete-bagging Monte Carlo error is sufficiently small. 

Figure~\ref{fig:SLO_ktrend} shows the effect of the subsample size \(k\). The numerical results again support the role of the condition \(Dk\gg n\) in controlling incomplete-bagging error. When \(Dk\) is sufficiently large, the interval length decreases and coverage settles near the nominal level. When \(Dk\) is small, the additional computational noise remains visible, leading to wider intervals and coverage above \(95\%\).

Finally, Figure~\ref{fig:SLO_stat_trend} compares the incomplete U- and V-statistic implementations. In this example, the two constructions produce very similar empirical coverage and average interval lengths, suggesting that the distinction between sampling with and without replacement has little effect on the reported finite-sample behavior for problem~\eqref{eq:SLO_problem}.

\begin{figure}[t]
    \centering
    \includegraphics[width=\textwidth]{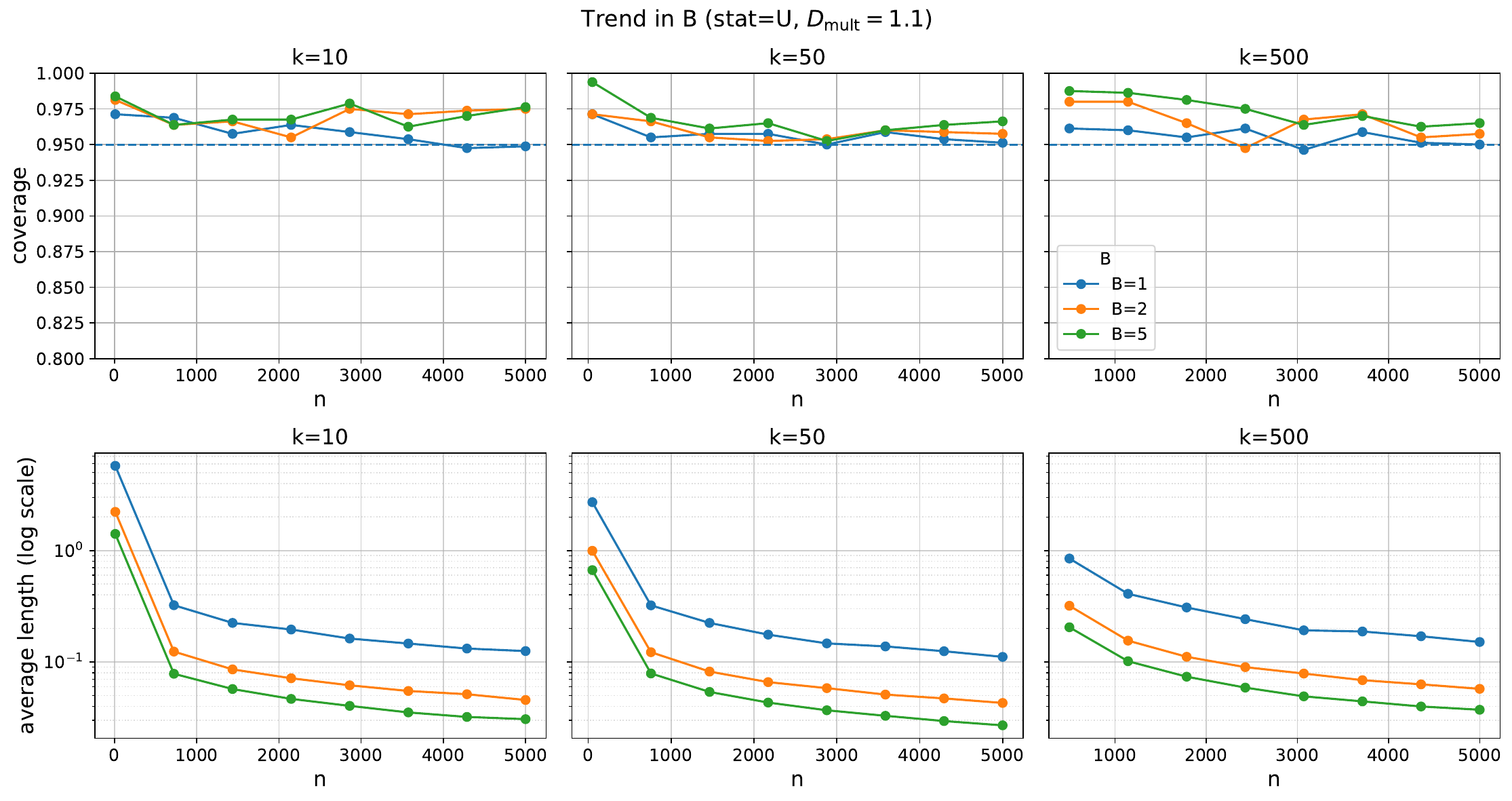}
    \caption{Trend in the number of outer bootstrap resamples \(B\) for the stochastic linear optimization problem in \eqref{eq:SLO_problem}, using the incomplete U-statistic estimator and fixing \(D_{\mathrm{mult}}=1.1\). The top row reports empirical coverage, and the bottom row reports average interval length on a log scale.}
    \label{fig:SLO_Btrend_U}
\end{figure}

\begin{figure}[t]
    \centering
    \includegraphics[width=\textwidth]{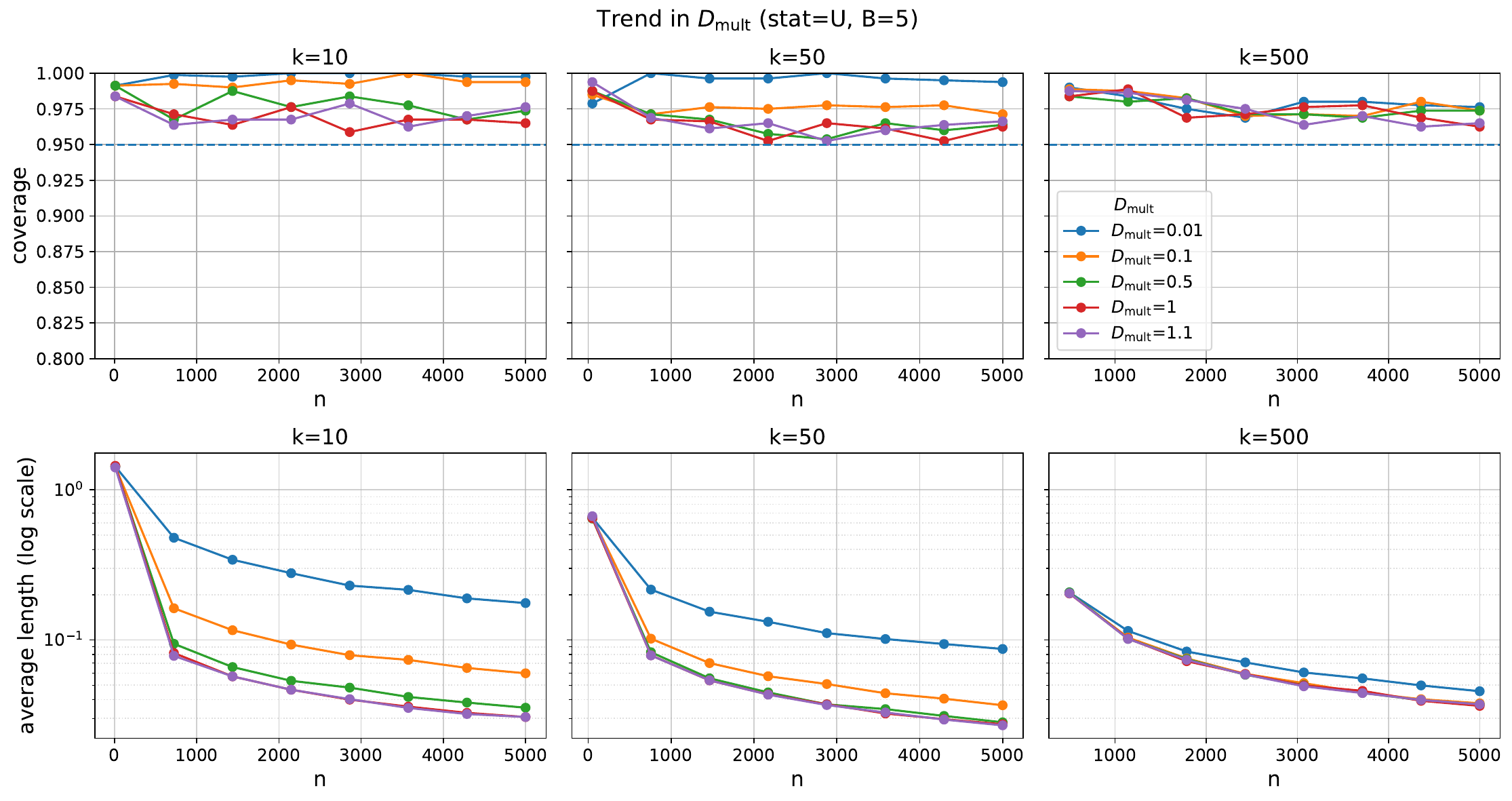}
    \caption{Trend in the bagging budget \(D_{\mathrm{mult}}\) for the stochastic linear optimization problem in \eqref{eq:SLO_problem}, using the incomplete U-statistic estimator and fixing \(B=5\). The top row reports empirical coverage, and the bottom row reports average interval length on a log scale.}
    \label{fig:SLO_Dtrend_U}
\end{figure}

\begin{figure}[t]
    \centering
    \includegraphics[width=\textwidth]{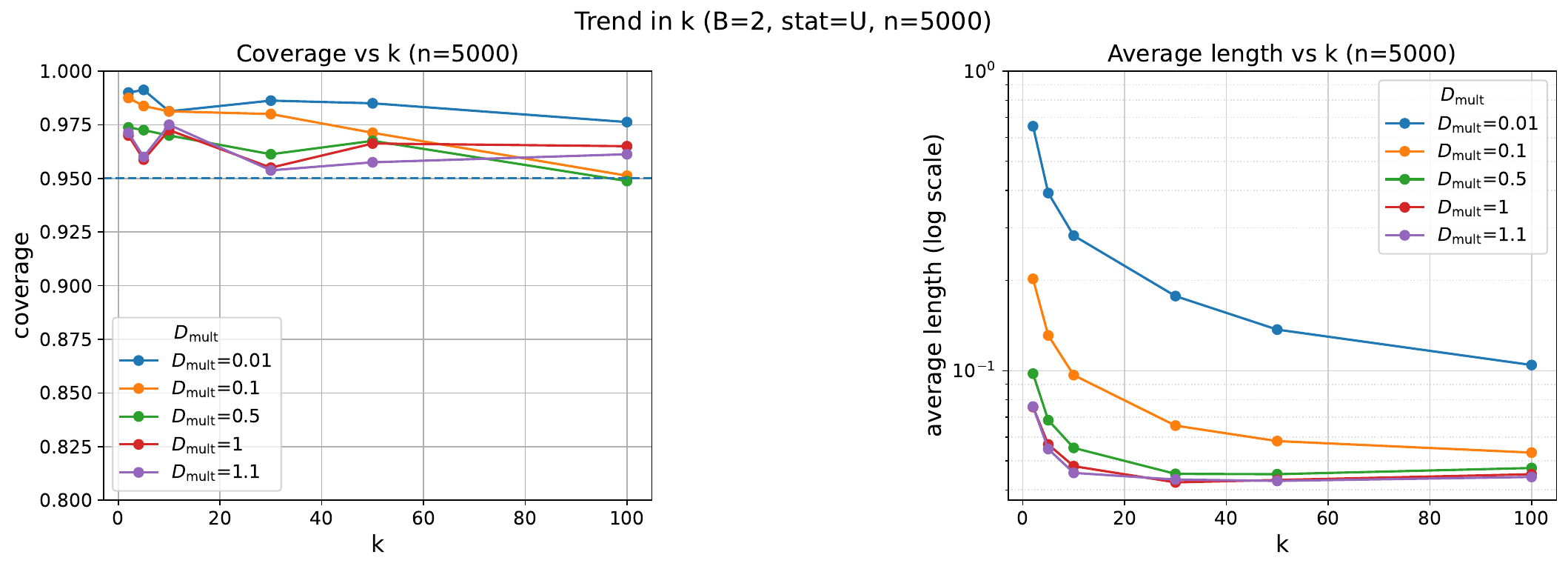}
    \caption{Trend in the subsample size \(k\) for various choices of \(D_{\mathrm{mult}}\) in the stochastic linear optimization problem in \eqref{eq:SLO_problem}, fixing \(n=5000\), \(B=2\), and stat \(=U\). The left panel reports empirical coverage, and the right panel reports average interval length on a log scale.}
    \label{fig:SLO_ktrend}
\end{figure}

\begin{figure}[t]
    \centering
    \includegraphics[width=\textwidth]{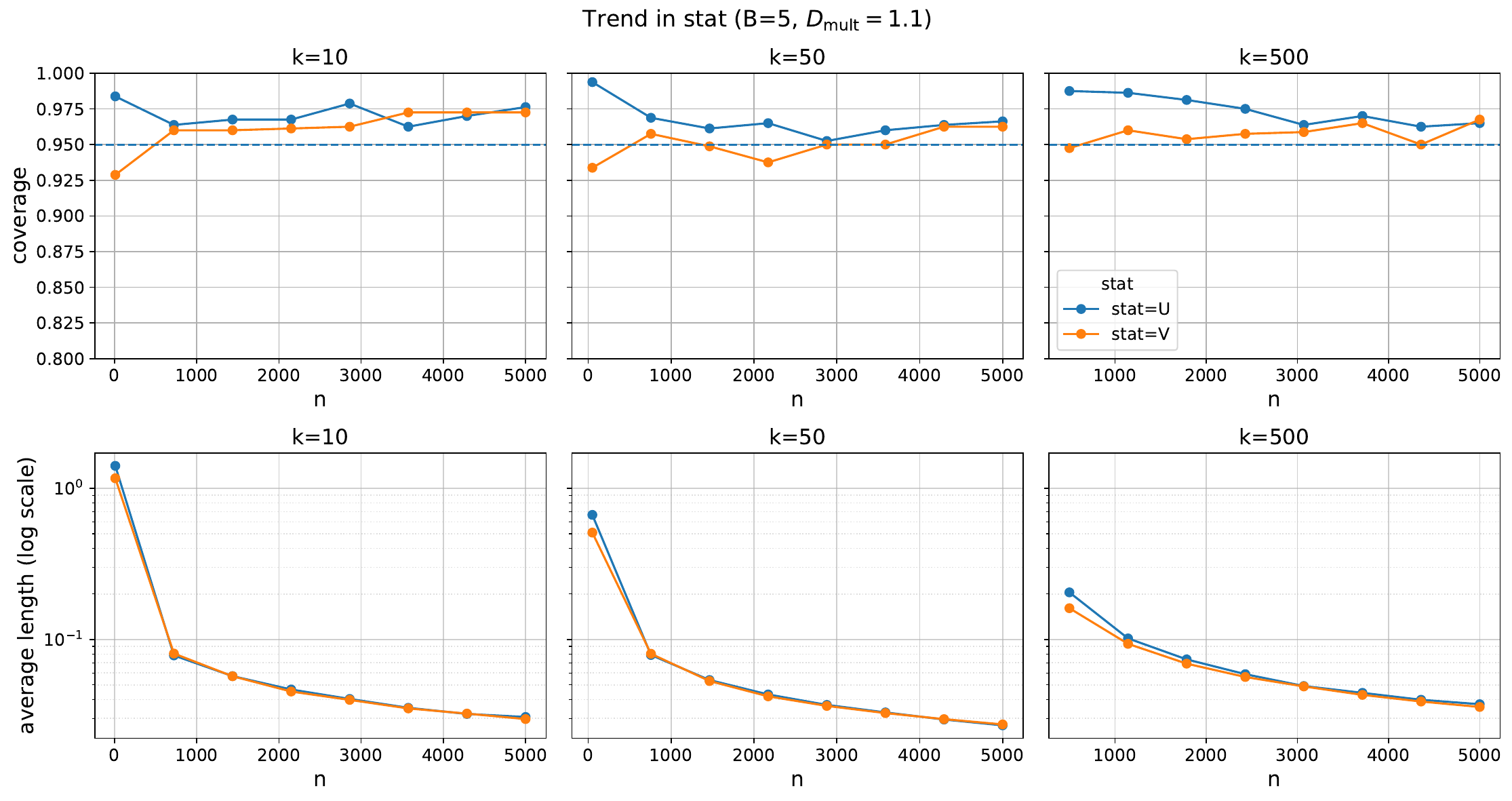}
    \caption{Comparison of the incomplete U- and V-statistic estimators for the stochastic linear optimization problem in \eqref{eq:SLO_problem}, with \(B\) and \(D_{\mathrm{mult}}\) fixed. The top row reports empirical coverage, and the bottom row reports average interval length on a log scale.}
    \label{fig:SLO_stat_trend}
\end{figure}

\subsection{Summary of Findings and Discussion}
\label{sec:Summary}

The numerical experiments in Sections~\ref{sec: decision_tree} and~\ref{sec: SLO} give a broadly consistent picture of the behavior of the proposed cheap bootstrap bagging intervals. Across both the random forest regression and stochastic optimization examples, the empirical coverage approaches the nominal \(95\%\) level as the sample size increases, including when the number of outer bootstrap resamples \(B\) is very small. In particular, the experiments show that useful coverage can be obtained even with \(B=1\), consistent with the main motivation of the cheap bootstrap construction.

The main effect of increasing \(B\) is on interval length and stability rather than on coverage. For very small \(B\), the Student-\(t\) critical value is larger and the scale estimate based on the bootstrap recomputations is more variable, leading to wider intervals. Increasing \(B\) reduces the average interval length, but most of this improvement occurs over the first few bootstrap resamples. In both examples, the additional gains become comparatively modest once \(B\) reaches a small value such as \(B=5\). Thus, the experiments suggest that a small fixed number of outer bootstrap resamples can already provide a favorable statistical--computational tradeoff.

The theoretical requirement on the bagging budget is that the incomplete-bagging Monte Carlo error be negligible relative to the sampling error. Thus, \(D\) remains a user-chosen computational parameter, and our approach is applicable as long as \(Dk\) is sufficiently large relative to \(n\). In our numerical experiments, the empirical performance was already satisfactory when \(Dk\) was approximately of order \(n\). Moreover, increasing \(D\) beyond a moderate fraction of \(n\) produced only limited additional reduction in interval length. Based on these results, we recommend choosing \(D\) to be a moderate fraction of \(n\), with \(D \approx 0.75n\) as a practical default when computationally feasible. Smaller values of \(D\) can still be used when computation is constrained, but they may lead to wider and more conservative intervals. The numerical results also show little practical difference between the incomplete U- and V-statistic implementations. Both versions produced similar coverage and interval lengths across the examples considered. Therefore, the procedure can be applied with either subsampling without replacement or resampling with replacement, depending on the convention most natural for the application.

Finally, our inference target throughout this paper is \(W_k\), the population expectation of the base kernel at subsample size \(k\). In some applications, \(W_k\) may itself be an approximation to a more fundamental population quantity, such as the true optimal value in stochastic optimization or the limiting performance of an ideal learning procedure. The confidence intervals developed here quantify uncertainty around \(W_k\) and the finite-\(D\) bagged estimator of \(W_k\). They do not, by themselves, quantify the gap between \(W_k\) and such downstream targets. Analyzing that additional approximation or bias requires separate problem-specific arguments. 

Overall, the proposed cheap bootstrap procedure provides a simple and scalable route to uncertainty quantification for bagged estimators. It preserves the conceptual appeal of bootstrap inference, avoids the large computational burden of nested resampling, and delivers asymptotically valid confidence intervals for the bagged target under explicit conditions on the bagging budget. The method is especially useful in settings where base-learner evaluations are expensive and the practitioner seeks a principled confidence interval using only a limited number of additional resampling computations.

\clearpage

\printbibliography

@ARTICLE{lam2018bounding,
  title        = "Bounding optimality gap in stochastic optimization via
                  bagging: Statistical efficiency and stability",
  author       = "Lam, Henry and Qian, Huajie",
  year         =  2018,
  primaryClass = "math.OC",
  eprint       = "1810.02905"
}

@incollection{hoeffding1992,
  title={A class of statistics with asymptotically normal distribution},
  author={Hoeffding, Wassily},
  booktitle={Breakthroughs in statistics: Foundations and basic theory},
  pages={308--334},
  year={1992},
  publisher={Springer}
}

@ARTICLE{athey_RF,
  title        = "Estimation and inference of heterogeneous treatment effects
                  using random forests",
  author       = "Wager, Stefan and Athey, Susan",
  year         =  2015,
  primaryClass = "stat.ME",
  eprint       = "1510.04342"
}

@article{Lam_cheap_bootstrap,
  title={A Cheap Bootstrap Method for Fast Inference},
  author={Lam, Henry},
  journal={arXiv preprint arXiv:2202.00090},
  year={2022}
}

@book{Vaart1998,
  title = {Asymptotic Statistics},
  ISBN = {9780521784504},
  url = {http://dx.doi.org/10.1017/CBO9780511802256},
  DOI = {10.1017/cbo9780511802256},
  publisher = {Cambridge University Press},
  author = {van der Vaart,  A. W. },
  year = {1998},
  month = oct 
}

@article{IJ_CLT,
  author    = {Stefan Wager and Trevor Hastie and Bradley Efron},
  title     = {Confidence Intervals for Random Forests: The Jackknife and the Infinitesimal Jackknife},
  journal   = {Journal of Machine Learning Research},
  volume    = {15},
  pages     = {1625--1651},
  year      = {2014},
  url       = {http://jmlr.org/papers/v15/wager14a.html}
}

@article{IJ_first,
  title = {Estimation and Accuracy After Model Selection},
  volume = {109},
  ISSN = {1537-274X},
  url = {http://dx.doi.org/10.1080/01621459.2013.823775},
  DOI = {10.1080/01621459.2013.823775},
  number = {507},
  journal = {Journal of the American Statistical Association},
  publisher = {Informa UK Limited},
  author = {Efron,  Bradley},
  year = {2014},
  month = jul,
  pages = {991–1007}
}

@book{Shapiro2009,
  title = {Lectures on Stochastic Programming: Modeling and Theory},
  ISBN = {9780898718751},
  url = {http://dx.doi.org/10.1137/1.9780898718751},
  DOI = {10.1137/1.9780898718751},
  publisher = {Society for Industrial and Applied Mathematics},
  author = {Shapiro,  Alexander and Dentcheva,  Darinka and Ruszczyński,  Andrzej},
  year = {2009},
  month = jan 
}

@article{Kleywegt2002,
  title = {The Sample Average Approximation Method for Stochastic Discrete Optimization},
  volume = {12},
  ISSN = {1095-7189},
  url = {http://dx.doi.org/10.1137/S1052623499363220},
  DOI = {10.1137/s1052623499363220},
  number = {2},
  journal = {SIAM Journal on Optimization},
  publisher = {Society for Industrial & Applied Mathematics (SIAM)},
  author = {Kleywegt,  Anton J. and Shapiro,  Alexander and Homem-de-Mello,  Tito},
  year = {2002},
  month = jan,
  pages = {479–502}
}

@article{Higle1991,
  title = {Stochastic Decomposition: An Algorithm for Two-Stage Linear Programs with Recourse},
  volume = {16},
  ISSN = {1526-5471},
  url = {http://dx.doi.org/10.1287/moor.16.3.650},
  DOI = {10.1287/moor.16.3.650},
  number = {3},
  journal = {Mathematics of Operations Research},
  publisher = {Institute for Operations Research and the Management Sciences (INFORMS)},
  author = {Higle,  Julia L. and Sen,  Suvrajeet},
  year = {1991},
  month = aug,
  pages = {650–669}
}

@article{Breiman2001_Random_Forest,
  title = {Random Forests},
  volume = {45},
  ISSN = {0885-6125},
  url = {http://dx.doi.org/10.1023/A:1010933404324},
  DOI = {10.1023/a:1010933404324},
  number = {1},
  journal = {Machine Learning},
  publisher = {Springer Science and Business Media LLC},
  author = {Breiman,  Leo},
  year = {2001},
  pages = {5–32}
}

@article{mak1999monte,
  title={Monte Carlo Bounding Techniques for Determining Solution Quality in Stochastic Programs},
  author={Mak, Wai-Kei and Morton, David P and Wood, R Kevin},
  journal={Operations research letters},
  volume={24},
  number={1-2},
  pages={47--56},
  year={1999},
  publisher={Elsevier}
}

@article{shapiro2003monte,
  title={Monte Carlo Sampling Methods},
  author={Shapiro, Alexander},
  journal={Handbooks in operations research and management science},
  volume={10},
  pages={353--425},
  year={2003},
  publisher={Elsevier}
}

@book{glasserman2004monte,
  title={Monte Carlo Methods in Financial Engineering},
  author={Glasserman, Paul},
  volume={53},
  year={2004},
  publisher={Springer}
}

@article{buhlmann2002analyzing,
  title={Analyzing {B}agging},
  author={B{\"u}hlmann, Peter and Yu, Bin},
  journal={The Annals of Statistics},
  volume={30},
  number={4},
  pages={927--961},
  year={2002},
  publisher={Institute of Mathematical Statistics},
  doi={10.1214/aos/1031689014}
}

@inproceedings{lam2023bootstrap,
  title={Bootstrap in high dimension with low computation},
  author={Lam, Henry and Liu, Zhenyuan},
  booktitle={International Conference on Machine Learning},
  pages={18419--18453},
  year={2023},
  organization={PMLR}
}

@article{huang2023efficient,
  title={Efficient uncertainty quantification and reduction for over-parameterized neural networks},
  author={Huang, Ziyi and Lam, Henry and Zhang, Haofeng},
  journal={Advances in neural information processing systems},
  volume={36},
  pages={64428--64467},
  year={2023}
}

@BOOK{Boucheron2013,
  title     = "Concentration inequalities",
  author    = "Boucheron, Stephane and Lugosi, Gabor and Massart, Pascal",
  publisher = "Oxford University Press",
  month     =  mar,
  year      =  2013,
  address   = "London, England",
  language  = "en"
}

@article{effron_stein_1981,
 author = {B. Efron and C. Stein},
 journal = {The Annals of Statistics},
 number = {3},
 pages = {586--596},
 publisher = {Institute of Mathematical Statistics},
 title = {The Jackknife Estimate of Variance},
 urldate = {2026-02-25},
 volume = {9},
 year = {1981}
}

@INPROCEEDINGS{Tim_RF,
  author={Tin Kam Ho},
  booktitle={Proceedings of 3rd International Conference on Document Analysis and Recognition}, 
  title={Random decision forests}, 
  year={1995},
  volume={1},
  number={},
  pages={278-282 vol.1},
  doi={10.1109/ICDAR.1995.598994}}

@article{mentch_hooker_2016,
  author  = {Lucas Mentch and Giles Hooker},
  title   = {Quantifying Uncertainty in Random Forests via Confidence Intervals and Hypothesis Tests},
  journal = {Journal of Machine Learning Research},
  year    = {2016},
  volume  = {17},
  number  = {26},
  pages   = {1--41},
  url     = {http://jmlr.org/papers/v17/14-168.html}
}

@article{diciccio_efron_1996,
author = {Thomas J. DiCiccio and Bradley Efron},
title = {{Bootstrap confidence intervals}},
volume = {11},
journal = {Statistical Science},
number = {3},
publisher = {Institute of Mathematical Statistics},
pages = {189 -- 228},
year = {1996},
doi = {10.1214/ss/1032280214},
URL = {https://doi.org/10.1214/ss/1032280214}
}

@ARTICLE{efron1979bootstrap,
  title     = "Bootstrap methods: Another look at the jackknife",
  author    = "Efron, B",
  journal   = "Ann. Stat.",
  publisher = "Institute of Mathematical Statistics",
  volume    =  7,
  number    =  1,
  pages     = "1--26",
  month     =  jan,
  year      =  1979
}

@BOOK{Efron1994_bootstrap_intro,
  title     = "An introduction to the bootstrap",
  author    = "Efron, Bradley and Tibshirani, R J",
  publisher = "Chapman and Hall/CRC",
  month     =  may,
  year      =  1994
}

@book{hall1992bootstrap,
  title     = {The Bootstrap and Edgeworth Expansion},
  author    = {Hall, Peter},
  year      = {1992},
  publisher = {Springer},
  address   = {New York},
  series    = {Springer Series in Statistics}
}

@article{breiman1996bagging,
  title   = {Bagging Predictors},
  author  = {Breiman, Leo},
  journal = {Machine Learning},
  volume  = {24},
  number  = {2},
  pages   = {123--140},
  year    = {1996},
  doi     = {10.1007/BF00058655}
}

@article{lam2023resampling,
  title={Cheap bootstrap for efficient uncertainty quantification of stochastic gradient descent},
  author={Lam, Henry and Wang, Zitong},
  journal={To appear in Journal of Machine Learning Research},
  year={2026}
}

@article{ohlendorff2025cheap,
  title={Cheap subsampling bootstrap confidence intervals for fast and robust inference},
  author={Ohlendorff, Johan Sebastian and Munch, Anders and S{\o}rensen, Kathrine Kold and Gerds, Thomas Alexander},
  journal={arXiv preprint arXiv:2501.10289},
  year={2025}
}

@inproceedings{lam2022cheap,
  title={Cheap bootstrap for input uncertainty quantification},
  author={Lam, Henry},
  booktitle={2022 Winter Simulation Conference (WSC)},
  pages={2318--2329},
  year={2022},
  organization={IEEE}
}

\section{Appendix}

\subsection{Notation}\label{sec: Notation}

In this section we introduce the main notation used throughout the paper and, for ease of readability, summarize the key quantities that will appear repeatedly in the theoretical developments. We begin with the classical U- and V-statistics, which serve as the population-level analogues of complete bagging.

Let $T_k:\Xi^k \to \mathbb{R}$ be a symmetric base learner acting on $k$ observations. Symmetry means that $T_k(\xi_{i_1},\dots,\xi_{i_k})$ is invariant under permutations of its arguments.

\paragraph{U-statistic.}
The associated U-statistic is defined as
\[
U_n(T_k)
:= \frac{1}{\binom{n}{k}}
\sum_{1 \le i_1 < \cdots < i_k \le n}
T_k(\xi_{i_1},\dots,\xi_{i_k}),
\]
which averages the kernel over all distinct subsets of size $k$. This corresponds to \emph{complete bagging without replacement}, and satisfies
\[
\E[U_n(T_k)] = W_k := \E\!\left[T_k(\xi_1,\dots,\xi_k)\right].
\]

\paragraph{V-statistic.}
The corresponding V-statistic is defined as
\[
V_n(T_k)
:= \frac{1}{n^k}
\sum_{1 \le i_1,\dots,i_k \le n}
T_k(\xi_{i_1},\dots,\xi_{i_k}),
\]
which averages over all ordered $k$-tuples, allowing repeated indices. This corresponds to \emph{complete bagging with replacement}. 

\medskip

If the argument $T_k$ is omitted, we write $U_{n,k}$ (and $V_{n,k}$) to denote the U-statistic (and V-statistic) corresponding to the canonical kernel $H_k$ introduced in Section~\ref{sec:problem_setup}. In particular,
\[
U_{n,k} := U_n(H_k),
\qquad
V_{n,k} := V_n(H_k),
\]
and both aim to estimate
\[
W_k = \E\!\left[ H_k(\xi_1,\dots,\xi_k) \right].
\]

\paragraph{Incomplete bagging.}
Since computing $U_{n,k}$ or $V_{n,k}$ is typically infeasible when $n$ and $k$ are moderate (due to the combinatorial growth of $\binom{n}{k}$ or $n^k$), we also consider the incomplete bagged estimator. Let $I^{(1)},\dots,I^{(D_n)}$ denote $D_n$ independently sampled index sets of size $k$ (with or without replacement, depending on the scheme). The incomplete bagged estimator is
\[
U_{n,k,D_n}\ (\text{ or } V_{n,k,D_n})
:= \frac{1}{D_n}
\sum_{d=1}^{D_n}
T_k\big(\xi_{i^{(d)}_1},\dots,\xi_{i^{(d)}_k}\big),
\]
which approximates $U_{n,k}$ (or $V_{n,k}$) depending on the sampling design.

\medskip

We next summarize the notation used for the general bagged estimators and for the proofs in Section~\ref{sec:main results} in Table~\ref{tab:notation_summary}.
\renewcommand{\arraystretch}{2} % Increases row height
\begin{table}[httb]
\centering
\begin{tabular}{|c|c|}
\hline
\textbf{Symbol} & \textbf{Definition}  \\
\hline
$H_k(\xi_1, \ldots, \xi_k)$ & The baseline predictor \\
\hline
$W_k$ & $\mathbb{E}[H_k(\xi_1, \ldots, \xi_k)]$  \\
\hline
$g_k(\xi)$ & $\mathbb{E}[H_k(\xi_1, \ldots, \xi_k) \mid \xi_1 = \xi]-W_k$  \\
\hline
$U_{n,k}$ & $\frac{1}{\binom{n}{k}}\sum_{\beta\in B} H_k(\xi_{\beta_1},...,\xi_{\beta_k})$ \\
\hline
$V_{n,k}$ & $\frac{1}{n^k}\sum_{\beta\in \tilde B} H_k(\xi_{\beta_1},...,\xi_{\beta_k})$  \\
\hline
$\mathring{H}_{n,k}$ & $W_k + \sum_{i = 1}^k g_k(\xi_i)$  \\
\hline
$\U_{n,k}$ & $W_k + \frac kn\sum_{i = 1}^n g_k(\xi_i)$  \\
\hline
$\zeta_{k,r}$ & $ cov\left[H_k(\xi_1,...,\xi_r,\xi_{r+1}...,\xi_k), \ H_{k}(\xi_1,...,\xi_r,\tilde \xi_{r+1},..., \tilde \xi_k)\right]$ \\
\hline
${\rho_{n,k}}$ &$\E\left[|g_k(\xi)  |^3\right]$\\
\hline
\end{tabular}
\caption{Summary of Notation for sample quantities}
\label{tab:notation_summary}
\end{table}
\renewcommand{\arraystretch}{1}

\subsection{Cheap Bootstrap Validity}\label{sec: Cheap_Bootstrap_Validity}
In this section, we present the general framework for establishing the validity of the cheap bootstrap. Theorem~\ref{thm: main_cheap_boot} is a version of the main cheap bootstrap validity result, adapted from \cite{Lam_cheap_bootstrap} to suit the present setting.

\begin{assumption}[Bootstrap validity]\label{Asm: bootstrap_validity}
There exist positive scaling sequences $\sigma_n$ and $\sigma_n^*$ (data-dependent) such that,
as $n\to\infty$:
\begin{enumerate}
\item[(i)] (Estimator CLT)
\begin{equation}\label{eq:estimator_CLT_star}
\frac{\hat\psi_n-\psi(\P)}{\sigma_n}\ \xrightarrow[]{d}\ N(0,1).
\end{equation}

\item[(ii)] (Conditional bootstrap CLT)
For each fixed $b$,
\begin{equation}\label{eq:bootstrap_CLT_star}
\frac{\hat\psi_n^{*b}-\hat\psi_n}{\sigma_n^*}\ \xrightarrow[]{d}\ N(0,1)
\end{equation}
\emph{conditionally on the data in probability}, i.e., for every $x\in\R$,
\begin{equation}\label{eq:conditional_CLT_meaning}
P\!\left(\left.\frac{\hat\psi_n^{*b}-\hat\psi_n}{\sigma_n^*}\le x\ \right|\ \P_n\right)
\ \xrightarrow[]{p}\ \Phi(x),
\end{equation}
where $\Phi$ is the standard normal cdf and $\sigma_n/\sigma_n^*\ \xrightarrow[]{p}\ 1$
\end{enumerate}

\end{assumption}
For notational convenience, any explicit $\sqrt{n}$ scaling is absorbed into $\sigma_n$ and $\sigma_n^*$. The above assumption is a sufficient condition that ensures \eqref{eq:cheap_boot_CI_general} achieves asymptotically exact coverage even when $B$ is kept fixed:

\begin{theorem}\label{thm: main_cheap_boot}
Under Assumption~\ref{Asm: bootstrap_validity}, for any fixed integer $B\ge 1$,
\begin{equation}\label{eq:t_limit_general}
\frac{\hat\psi_n-\psi(\P)}{S_{n,B}}\ \xrightarrow[]{d}\ t_B,
\end{equation}
where $t_B$ denotes a Student-$t$ distribution with $B$ degrees of freedom and $S_{n,B}$ is defined in
\eqref{eq:S_nb_def}. Consequently, the interval $I_{n,B}(\psi)$ in \eqref{eq:cheap_boot_CI_general} satisfies
\begin{equation}\label{eq:coverage_limit_general}
\pr\bigl(\psi(\P)\in I_{n,B}(\psi)\bigr)\ \to\ 1-\alpha.
\end{equation}
\end{theorem}

To show the above result, similar to \cite{Lam_cheap_bootstrap}, we have:
\begin{proposition}\label{prop: main}
    Under Assumption \ref{Asm: bootstrap_validity}, we have:
    $$\left(\frac{\hat\psi_n-\psi}{\sigma_n}, \frac{\psi^{*1}_n-\hat\psi_n}{\sigma_n},...,\frac{\psi^{*B}_n-\hat\psi_n}{\sigma_n}\right) \xrightarrow[]{d} \left(Z_0,Z_1,...,Z_B\right)$$
    where $Z_b,\ b = 0,...,B$ are i.i.d $N(0,1)$ variables.
\end{proposition}

We next prove the above results:

\begin{proof}[Proof of Theorem \ref{thm: main_cheap_boot}] Using Proposition \ref{prop: main} and continuous mapping theorem,
    \begin{align*}
        T &= \frac{\hat \psi_n-\psi}{S}\\
        &= \frac{\hat \psi_n-\psi}{\sqrt{\frac 1B \sum_{b=1}^B\left(\psi^{*b}_n-\hat\psi_n\right)}}
        = \frac{\frac{\hat \psi_n-\psi}{\sigma_n}}{\sqrt{\frac 1B \sum_{b=1}^B\left(\frac{\psi^{*b}_n-\hat\psi_n}{\sigma_n}\right)^2}}\\
        & \xrightarrow[]{d} \frac{Z_0}{\sqrt{\frac 1B \sum_{b=1}^BZ_b^2}} \overset{d}{=} t_B
    \end{align*}
\end{proof}

We next prove the Proposition \ref{prop: main}:
\begin{proof}[Proof of Proposition \ref{prop: main}]
For any $x_b\in\R,\ b\in\{0,...,B\}$, we will show :
$$P\left(\frac{\hat\psi_n-\psi}{\sigma_n}\leq x_0, \frac{\psi^{*1}_n-\psi}{\sigma_n}\leq x_1, \cdots , \frac{\hat\psi^{*B}_n-\psi}{\sigma_n}\leq x_B\right)\rightarrow \prod_{b=0}^B \Phi(x_b)\quad \text{ as } n\to \infty$$
We have:

\begin{align*}
&\left|
P\left(
\frac{\hat\psi_n-\psi}{\sigma_n}\leq x_0,\,
\frac{\hat\psi^{*1}_n-\hat\psi_n}{\sigma_n}\leq x_1,\,
\ldots,\,
\frac{\hat\psi^{*B}_n-\hat\psi_n}{\sigma_n}\leq x_B
\right)
-
\prod_{b=0}^B\Phi(x_b)
\right| \\
=\;&
\left|
E\left[
I\left(
\frac{\hat\psi_n-\psi}{\sigma_n}\leq x_0,\,
\frac{\hat\psi^{*1}_n-\hat\psi_n}{\sigma_n}\leq x_1,\,
\ldots,\,
\frac{\hat\psi^{*B}_n-\hat\psi_n}{\sigma_n}\leq x_B
\right)
\right]
-
\prod_{b=0}^B\Phi(x_b)
\right| \\
=\;&
\left|
E\left[
E\left[
I\left(
\frac{\hat\psi_n-\psi}{\sigma_n}\leq x_0,\,
\frac{\hat\psi^{*1}_n-\hat\psi_n}{\sigma_n}\leq x_1,\,
\ldots,\,
\frac{\hat\psi^{*B}_n-\hat\psi_n}{\sigma_n}\leq x_B
\right)
\,\Big|\, \hat P_n
\right]
\right]
-
\prod_{b=0}^B\Phi(x_b)
\right| \\
=\;&
\left|
\begin{aligned}
& E\left[
I\left(\frac{\hat\psi_n-\psi}{\sigma_n}\leq x_0\right)
\prod_{b=1}^B
P\left(
\frac{\hat\psi^{*b}_n-\hat\psi_n}{\sigma_n}\leq x_b
\,\Big|\, \hat P_n
\right)
\right] \\
&\quad -
\prod_{b=0}^B\Phi(x_b)
\end{aligned}
\right| \\
\leq\;&
\left|
\begin{aligned}
& E\left[
I\left(\frac{\hat\psi_n-\psi}{\sigma_n}\leq x_0\right)
\prod_{b=1}^B
P\left(
\frac{\hat\psi^{*b}_n-\hat\psi_n}{\sigma_n}\leq x_b
\,\Big|\, \hat P_n
\right)
\right] \\
&\quad -
P\left(\frac{\hat\psi_n-\psi}{\sigma_n}\leq x_0\right)
\prod_{b=1}^B\Phi(x_b)
\end{aligned}
\right| \\
&+
\left|
P\left(\frac{\hat\psi_n-\psi}{\sigma_n}\leq x_0\right)
\prod_{b=1}^B\Phi(x_b)
-
\prod_{b=0}^B\Phi(x_b)
\right| \\
=:\;& A+B .
\end{align*}

Now,
\begin{align*}
A
=\;&
\left|
\begin{aligned}
& E\left[
I\left(\frac{\hat\psi_n-\psi}{\sigma_n}\leq x_0\right)
\prod_{b=1}^B
P\left(
\frac{\hat\psi^{*b}_n-\hat\psi_n}{\sigma_n}\leq x_b
\,\Big|\,\hat P_n
\right)
\right] \\
&\quad -
P\left(\frac{\hat\psi_n-\psi}{\sigma_n}\leq x_0\right)
\prod_{b=1}^B\Phi(x_b)
\end{aligned}
\right| \\
=\;&
\left|
E\left[
I\left(\frac{\hat\psi_n-\psi}{\sigma_n}\leq x_0\right)
\left\{
\prod_{b=1}^B
P\left(
\frac{\hat\psi^{*b}_n-\hat\psi_n}{\sigma_n}\leq x_b
\,\Big|\,\hat P_n
\right)
-
\prod_{b=1}^B\Phi(x_b)
\right\}
\right]
\right| \\
\leq\;&
E\left[
I\left(\frac{\hat\psi_n-\psi}{\sigma_n}\leq x_0\right)
\left|
\prod_{b=1}^B
P\left(
\frac{\hat\psi^{*b}_n-\hat\psi_n}{\sigma_n}\leq x_b
\,\Big|\,\hat P_n
\right)
-
\prod_{b=1}^B\Phi(x_b)
\right|
\right] \\
\leq\;&
E\left[
\underbrace{
\left|
\prod_{b=1}^B
P\left(
\frac{\hat\psi^{*b}_n-\hat\psi_n}{\sigma_n}\leq x_b
\,\Big|\,\hat P_n
\right)
-
\prod_{b=1}^B\Phi(x_b)
\right|
}_{=:A_n}
\right]
\;\longrightarrow\;0 .
\end{align*}

Because, $A_n\xrightarrow[]{p} 0$ by Assumption \ref{Asm: bootstrap_validity} , and $|A_n|\leq 1$ a.s. Thus, by DCT, $E[|A_n|]\to 0$.  
Also,
\begin{align*}
    B =\ & \left|P\left(\frac{\hat\psi_n-\psi}{\sigma_n}\leq x_0\right)\prod_{b=1}^B\Phi(x_b) - \prod_{b=0}^B\Phi(x_b)\right|\\
    \leq \ &\left|P\left(\frac{\hat\psi_n-\psi}{\sigma_n}\leq x_0\right) - \Phi(x_0)\right|\cdot\prod_{b=1}^B\Phi(x_b) \\
    \leq \ &\left|P\left(\frac{\hat\psi_n-\psi}{\sigma_n}\leq x_0\right) - \Phi(x_0)\right|\to 0 
\end{align*}
again, by Assumption \ref{Asm: bootstrap_validity}.
\end{proof}

\subsection{Auxiliary Results}
\subsubsection{Conditional Slutsky}\label{sec: Conditional_Slutsky}

The usual Slutsky's Theorem goes as follows:
\begin{theoremapx}(\textbf{Slutsky})
    Let $\xi_n$ and $Y_n$ be two sequences of real valued random variables such that $\xi_n \xrightarrow[]{d}X$ and $Y_n\xrightarrow[]{p}0$. Then $\xi_n+Y_n \xrightarrow[]{d}X$.
\end{theoremapx}

For our results, we state and prove a conditional version of the same:
\begin{theoremapx}\label{thm: conditional_Slutsky}(\textbf{Conditional Slutsky})
    Let $\xi_n$ and $Y_n$ be two sequences of real valued random variables. Consider a third sequence of random variables $Z_n$. Suppose: 
    \begin{enumerate}
        \item Conditional on $Z_n$, $\xi_n \xrightarrow[]{d} X$ in probability, i.e,
        \begin{equation}\label{eq: slutsky_condition_1}
            \big| \pr\left(\xi_n\leq x\mid Z_n\right) - F(x) \big| \xrightarrow[]{p} 0
        \end{equation}
        
        where $F(\cdot)$ is the c.d.f of $X$ and $x$ is any continuity point of $F$.
        \item Conditional on $Z_n$, $Y_n \xrightarrow[]{p} 0$ in probability, i.e, for any $\delta>0$:
        \begin{equation}\label{eq: slutsky_condition_2}
            \pr(|Y_n|>\delta \mid Z_n)\xrightarrow[]{p} 0
        \end{equation}
        
    \end{enumerate}
    Then we have:
    \begin{enumerate}[i)]
        \item Conditional on $Z_n$, $\xi_n +Y_n \xrightarrow[]{d} X$ in probability, i.e,
        $$\big| \pr\left(\xi_n+Y_n\leq x\mid Z_n\right) - F(x) \big| \xrightarrow[]{p} 0$$
        for any $x$ that is a continuity point of $F$.
        \item Conditional on $Z_n$, $\xi_n \cdot Y_n \xrightarrow[]{p} 0$ in probability, i.e,
        \begin{equation}\label{eq: slutsky_2}
            \pr(|\xi_n\cdot Y_n|>\delta \mid Z_n)\xrightarrow[]{p} 0
        \end{equation}
    \end{enumerate}
         
\end{theoremapx}

\begin{proof} We first prove (i). Fix an $\epsilon>0$ and an $x$ which is a continuity point of $F$ and show that:
    
    \begin{equation}
        \pr\left[\big| \pr\left(\xi_n+Y_n\leq x\mid Z_n\right) - F(x) \big| > \epsilon\right] \to 0
    \end{equation}
By continuity of $F$ at $x$, we have that $\exists \delta>0$ such that:
\begin{equation}\label{eq: slutsky_delta_choice}
    |x-y|\leq \delta \implies |F(x)-F(y)| \leq \frac{\epsilon}{5}
\end{equation} Further, since $F$ is a cdf, it has at most countably many discontinuities. Thus, $\delta$ may be chosen such that $x \pm \delta$ are also  continuity points of $F$.
\begin{align*}
    \pr(\xi_n+Y_n \leq x \mid Z_n) & \leq  \pr(|Y_n| > \delta \mid Z_n) + \pr(|Y_n| \leq \delta,\  \xi_n+Y_n \leq x \mid Z_n) \\
    & \leq  \pr(|Y_n| > \delta \mid Z_n) + \pr(\xi_n \leq x+\delta \mid Z_n)\quad a.s. \numberthis \label{slutsky_1}
\end{align*}
Similarly, we also have:
\begin{align*}
    \pr(\xi_n \leq x-\delta \mid Z_n) & \leq  \pr(|Y_n| > \delta \mid Z_n) + \pr(|Y_n| \leq \delta,\  \xi_n \leq x-\delta \mid Z_n) \\
    & \leq  \pr(|Y_n| > \delta \mid Z_n) + \pr(\xi_n +Y_n\leq x \mid Z_n)\quad a.s. \numberthis \label{slutsky_2}
\end{align*}

Combining \ref{slutsky_1} and \ref{slutsky_2}, we have:
\begin{align*}
    & \pr(\xi_n \leq x-\delta \mid Z_n) - F(x) - \pr(|Y_n| > \delta \mid Z_n) \\
    \leq\ & \pr(\xi_n +Y_n\leq x \mid Z_n) - F(x)\\
    \leq\ & \pr(\xi_n \leq x+\delta \mid Z_n) - F(x) + \pr(|Y_n| > \delta \mid Z_n)\quad a.s
\end{align*}
Thus,
\begin{align*}
    &\Big|\pr(\xi_n +Y_n\leq x \mid Z_n) - F(x)\Big|\\
    \leq\ & \pr(|Y_n| > \delta \mid Z_n) + \max\left\{ \Big| \pr(\xi_n \leq x+\delta \mid Z_n) - F(x)\Big|,\  \Big|\pr(\xi_n \leq x-\delta \mid Z_n) - F(x)\Big|\right\}\\
    \leq\ & \pr(|Y_n| > \delta \mid Z_n) \\
    +\ &\ \Big| \pr(\xi_n \leq x+\delta \mid Z_n) - F(x+\delta)\Big|+ \Big|F(x+\delta)-F(x)\Big| \\
    +\ &\ \Big|\pr(\xi_n \leq x-\delta \mid Z_n) - F(x-\delta)\Big|+ \Big|F(x-\delta)-F(x)\Big|\quad\quad  a.s 
\end{align*}
Finally, by union bound, we have:
\begin{align*}
    & \pr\left[\Big|\pr(\xi_n +Y_n\leq x \mid Z_n) - F(x)\Big|>\epsilon\right]\\
    \leq\ & \pr\left[\pr(|Y_n| > \delta \mid Z_n)>\frac{\epsilon}{5}\right]\numberthis\label{inprob}\\
    + & \pr\left[\Big| \pr(\xi_n \leq x+\delta \mid Z_n) - F(x+\delta)\Big|>\frac{\epsilon}{5}\right]\numberthis\label{indist_1}\\
    + & \pr\left[\Big|F(x+\delta)-F(x)\Big|>\frac{\epsilon}{5}\right]\numberthis\label{gap_1}\\
    + & \pr\left[\Big|\pr(\xi_n \leq x-\delta \mid Z_n) - F(x-\delta)\Big|>\frac{\epsilon}{5}\right]\numberthis\label{indist_2}\\
    + & \pr\left[\Big|F(x-\delta)-F(x)\Big|>\frac{\epsilon}{5}\right]\numberthis\label{gap_2}
\end{align*}
Now, (\ref{inprob}) $\to 0$ due to \ref{eq: slutsky_condition_2}. (\ref{indist_1}\&\ref{indist_2}) $\to 0$ due to \ref{eq: slutsky_condition_1}. Finally, (\ref{gap_1} \& \ref{gap_2}) are identically $0$ due to choice of $\delta$ in (\ref{eq: slutsky_delta_choice}).

For (ii): Fix any $\delta$. Given any  $\epsilon, \exists M_{\epsilon}$ such that $P(|X|>M_{\epsilon})\leq \frac{\epsilon}{4}$
\begin{align*}
    P(|\xi_n\cdot Y_n|>\delta \mid Z_n)& \leq P(|\xi_n|\cdot |Y_n|>\delta \mid Z_n) \\
    &\leq  P(|\xi_n|>M_{\epsilon} \mid Z_n) + P(|Y_n|>\delta/M_{\epsilon} \mid Z_n)
\end{align*}
Thus,
\begin{align*}
   \pr\big(  \pr(|\xi_n\cdot Y_n|>\delta \mid Z_n) >\epsilon  \big)& \leq  \pr\big(\pr(|\xi_n|>M_{\epsilon} \mid Z_n) + \pr(|Y_n|>\delta/M_{\epsilon} \mid Z_n) >\epsilon  \big)\\
   & \leq \pr\big(\pr(|\xi_n|>M_{\epsilon} \mid Z_n)  >\epsilon/2  \big) +  \pr\big(\pr(|Y_n|>\delta/M_{\epsilon} \mid Z_n) >\epsilon/2  \big)
\end{align*}
Now, $\pr\big(\pr(|Y_n|>\delta/M_{\epsilon} \mid Z_n) >\epsilon/2  \big) \to 0$ due to (\ref{eq: slutsky_condition_2}). For the first term,
\begin{align*}
    &\pr\big(\pr(|\xi_n|>M_{\epsilon} \mid Z_n)  >\epsilon/2  \big) \\
    \leq & \pr\left(\big|\pr(|\xi_n|>M_{\epsilon} \mid Z_n) - \pr(|X|>M_{\epsilon})\big| >\epsilon/4  \big) + \pr\big( \pr(|X|>M_{\epsilon}) >\epsilon/4  \right)
\end{align*}
$\pr\big( \pr(|X|>M_{\epsilon}) >\epsilon/4  \big)$ is identically $0$ by choice of $M_{\epsilon}$ and the first term $\to 0$ due to \ref{eq: slutsky_condition_2}.

\end{proof}

\subsubsection{Properties of U and V statistics}

In this section, we collect some general properties of the associated U- and V-statistics that will be used throughout the paper. Theorem~\ref{thm: diff_U_V} establishes that, for a fixed kernel $T_k$, the corresponding U- and V-statistics are asymptotically close. Theorem~\ref{thm: CLT_U_stat_modified} characterizes the rate at which the U-statistic converges to its population counterpart $\theta_k := \mathbb{E}[T_k]$.

\begin{theorem}\label{thm: diff_U_V}
Let $\xi_1,\ldots,\xi_n$ be i.i.d.\ draws from an arbitrary distribution $F$ on $\Xi$, and let
$T_k:\Xi^k\to\R$ be a symmetric kernel. Assume there exists a constant $B>0$ such that
\[
\pr_F\left[|T_k(\xi_1,\ldots,\xi_k)|\le B\right]=1.
\]
Then,
\[
\bigl|V_n(T_k)-U_n(T_k)\bigr|
\le B\cdot c(n,k), \quad F\text{-a.s.}
\]
for some constant $c(n,k)$ that depends only on $(n,k)$ and is independent of both $F$ and $B$. Moreover, if $k^2/n \to 0$, then
\[
c(n,k) = O\left(\frac{k^2}{n}\right).
\]
\end{theorem}

\begin{proof}
    The number of terms in $V_{n}(T_k)$ with all distinct resamples is $\binom{n}{k}k!$, and their average is $U_{n}(T_k)$. The number of remaining terms is $n^k - \binom{n}{k}k!$. Let their average value be $R_{n,k}$ with $|R_{n,k}|\leq B$ and $|U_{n}(T_k)|\leq B$. Thus,
    \begin{align*}
        &\ n^k V_{n}(T_k) = U_{n}(T_k)\binom{n}{k}k! + \left(n^k - \binom{n}{k}k!\right)R_{n,k}\\
        \implies&\  n^k(V_{n}(T_k)-U_{n}(T_k)) = U_{n}(T_k)\left(\binom{n}{k}k! - n^k\right) + \left(n^k - \binom{n}{k}k!\right)R_{n,k}\\
        \implies&\  |V_{n}(T_k)-U_{n}(T_k)| \leq  \frac{2 B}{n^k}\left|\binom{n}{k}k! - n^k\right|
    \end{align*}
    Take:
    $$c(n,k) = \frac{2}{n^k}\left|\binom{n}{k}k! - n^k\right|$$
    Now, assuming $k^2/n \to 0$, it is enough to show:
    $$\frac{1}{n^k}\left(n^k - \binom{n}{k}k!\right) = 1-\binom{n}{k}\frac{k!}{n^k}= O(k^2/n)$$
    Observe that
    \begin{align*}
        A_{n,k} &:= \binom{n}{k}\frac{k!}{n^k}
         = \prod_{i=0}^{k-1}\left(1-\frac{i}{n}\right)\\
        L_{n,k}&:= \log A_{n,k} = \sum_{i=0}^{k-1}\log\left(1-\frac{i}{n}\right)
    \end{align*}

    Now,
    $\log(1-x) = -x + O(x^2) $. Thus,
    \begin{align*}
        L_{n,k} &= -\frac 1n\sum_{i=0}^{k-1} i + O\left(\frac {1}{n^2}\sum_{i=0}^{k-1} i^2\right) \\
        &= -\frac{k(k+1)}{2n} + O\left(\frac{k^3}{n^2}\right) \\
        &= -\frac{k(k+1)}{2n} + o\left(\frac{k^2}{n}\right) \quad [\because k/n \to 0]\\
        &= O(k^2/n) \to 0
    \end{align*}
Thus, $O(L_{n,k}^2) = o(L_{n,k})$. This gives:
$$ 1-A_{n,k} = 1-e^{L_{n,k}}  = -L_{n,k} + O(L_{n,k}^2)
    = O(|L_{n,k}|) = O(k^2/n)$$

\end{proof}

We next present the following result. Define
\[
\theta_k := \E[T_k],
\qquad
G_{k,1}(\xi) := \E\big[T_k(\xi_1,\ldots,\xi_k)\mid \xi_1 = \xi\big].
\]
More generally, for each $c\in\{1,\ldots,k\}$, define the $c$-th order conditional expectation
\[
G_{k,c}(\xi_1,\ldots,\xi_c)
:= \E\!\left[T_k(\xi_1,\ldots,\xi_k)\,\middle|\,\xi_1,\ldots,\xi_c\right].
\]
Also define the corresponding variance components
\[
\Gamma_{k,c} := \Var\!\big[ G_{k,c}(\xi_1,\ldots,\xi_c) \big],
\qquad c=1,\ldots,k.
\]
In particular, $\Gamma_{k,1}=\Var\big[G_{k,1}(\xi)\big]$ and $\Gamma_{k,k}=\Var[T_k]$.

\begin{theorem}\label{thm: CLT_U_stat_modified}
Suppose the kernel order $k = k_n$ satisfies $k^2/n \to 0$ as $n \to \infty$. If $T_k(\cdot)$ is bounded a.s, then we have:
\begin{equation}\label{eq: Important_Results_U_var_mod}
    \Var(U_{n}(T_k)) = O\left(\frac{k}{n}\cdot \Gamma_{k,k}\right)
\end{equation}
\end{theorem}

\begin{proof}
    
Following the argument in Theorem 12.3 of \cite{Vaart1998}, we can write
\begin{equation}
\Var(U_n(T_k)) = \sum_{c=1}^k \binom{n}{k}^{-1} \binom{k}{c} \binom{n-k}{k-c} \, \Gamma_{k,c} ,
\end{equation}
where
\begin{equation}
\Gamma_{k,c} = \Cov\!\left[ T_k(\xi_1,\ldots,\xi_c,\xi_{c+1},\ldots,\xi_k), \; T_k(\xi_1,\ldots,\xi_c,\xi_{k+1},\ldots,\xi_{2k-c}) \right],
\end{equation}
is the covariance between two kernels sharing exactly $c$ arguments.  Since $T_k$ is bounded, $\Gamma_{k,c}$ exists for all $(k,c)$. 

For $c=1$, dividing by $k/n$ gives:
\[
 \frac{\big(k / \binom{n}{k}\big) \binom{n-k}{k-1} }{(k^2 / n)} \cdot k\Gamma_{k,1}
= \frac{(n-k)!(n-k)!}{(n-1)!(n-2k+1)!} \cdot k\Gamma_{k,1}.
\]
Now, $k\Gamma_{k,1} \leq \Gamma _{k,k}$ (see \cite{hoeffding1992}, Theorem 5.1). And,  $$\frac{(n-k)!(n-k)!}{(n-1)!(n-2k+1)!} \rightarrow 1$$
where the limit follows from $k^2/n \to 0$ by an application of Stirling’s formula and a Taylor expansion.  

Similarly dividing rest of the terms by $k/n$,
\begin{align}
&\frac{\sum_{c=2}^k \binom{n}{k}^{-1} \binom{k}{c} \binom{n-k}{k-c} \, \Gamma_{k,c}}{k/n} \nonumber\\
&\quad\le \sum_{c=2}^k \frac{(1/c!)\, [\,k!/(k-c)!\,]^2 \, ((n-k)!/n!) \, ((n-k)!/(n-2k+c)!)}{k^2/n} \, \cdot k\Gamma_{k,c} \nonumber\\
&\quad\le \sum_{c=2}^k \frac{k^{2c}/c!}{(n-k+1)^c} \cdot k\Gamma_{k,c}
\;=\; \sum_{c=2}^k \frac{1}{c!} \, \epsilon_n^{\,c-1} \, k\Gamma_{k,c}
\end{align}
with
\[
\epsilon_n = \frac{k^2}{n-k+1} \to 0 \quad [\because k^2/n \to 0]
\]
Again, using $k\Gamma_{k,c} \le c\,\Gamma_{k,k}$ ( \cite{hoeffding1992}, Theorem 5.1), the last expression is bounded by
\begin{equation}
\Gamma_{k,k}\cdot \sum_{c=2}^k \frac{1}{(c-1)!} \, \epsilon_n^{\,c-1}
= \Gamma_{k,k} \sum_{j=1}^{k-1}\cdot \frac{\epsilon_n^{\,j}}{j!}
\le\Gamma_{k,k} \cdot \frac{\epsilon_n - \epsilon_n^{\,k}}{1 - \epsilon_n}.
\end{equation}
The result follows as $ \frac{\epsilon_n - \epsilon_n^{\,k}}{1 - \epsilon_n} \to 0$.
\end{proof}

\subsubsection{The ANOVA Decomposition}
To analyze the behavior of $U_{n,k}$, we rely on the classical ANOVA decomposition of symmetric statistics \cite{effron_stein_1981}:

\begin{lemma}\label{lem: anova_decomp}
Let $T(\xi_1,\ldots,\xi_k)$ be a symmetric statistic with 
$\Var\!\left[T(\xi_1,\ldots,\xi_k)\right] < \infty$.  
Then there exist measurable functions $T_1,\ldots,T_k$ such that
\begin{equation}\label{eq: anova_expansion}
    T(\xi_1,\ldots,\xi_k)
    \;=\;
    \E[T]
    \;+\;
    \sum_{i=1}^k T_1(\xi_i)
    \;+\;
    \sum_{1\le i<j\le k} T_2(\xi_i,\xi_j)
    \;+\;\cdots+\;
    T_k(\xi_1,\ldots,\xi_k),
\end{equation}
where each term has mean zero and all terms are mutually uncorrelated.
\end{lemma}

A key tool in understanding the sampling distribution of $U_{n,k}$ is its Hájek projection, which extracts the leading first-order contribution from the ANOVA expansion.  
For our kernel $H_k$, define
\[
    g_k(\xi)
    := \E\!\left[\, H_k(\xi,\xi_2,\ldots,\xi_k) \,\right] - W_k,
\]
which coincides with the first-order ANOVA component $T_1$ in Lemma~\ref{lem: anova_decomp}.  
The Hájek projection of the U-statistic is then
\begin{equation}\label{eq: hajek_projection}
    \U_{n,k}
    \;:=\;
    W_k
    \;+\;
    \frac{k}{n}\sum_{i=1}^n g_k(\xi_i),
\end{equation}
providing a linear approximation whose variability dominates the asymptotic behavior of $U_{n,k}$.  
The remainder $U_{n,k} - \U_{n,k}$ collects all higher-order ($\ge2$) ANOVA terms. To ensure that this linear approximation is sufficiently accurate, we have the following lemma:

\begin{lemma}[Control of higher-order terms]\label{lem:athey_bound}
Suppose the base learner $H_k(\xi_1,\ldots,\xi_k)$ is symmetric and uniformly bounded, i.e., Assumption \ref{asm: h_bounded} holds. Then the associated complete bagged estimator $U_{n,k}$ and its Hájek projection $\mathring U_{n,k}$ satisfy
\[
\E\!\left[(U_{n,k}-\mathring U_{n,k})^2\right]
\le
\frac{k^2}{n^2}\Var(H_k)
=
O\left(\frac{k^2}{n^2}\right)
\]
\end{lemma}

This lemma shows that the higher-order ANOVA terms are automatically negligible whenever the tree-based base learner is uniformly bounded. Consequently, the Hájek projection $\mathring U_{n,k}$ approximates $U_{n,k}$ with the required precision.

\begin{proof}[Proof of Lemma \ref{lem:athey_bound}]
Following the notation of Lemma \ref{lem: anova_decomp} and the U-statistic notation introduced in Section \ref{sec: Notation}, for any symmetric statistic $T(\xi_1,\ldots,\xi_k)$ with 
$\Var\!\left[T(\xi_1,\ldots,\xi_k)\right] < \infty$, The remainder $U_{n}(T)-\mathring U_{n}(T)$ consists exactly of the ANOVA terms of order at least two. By orthogonality of the ANOVA decomposition (\cite{Vaart1998}),
\[
\E\!\left[U_{n}(T)-\mathring U_{n}(T)\right]^2
=
\Var(U_{n}(T)-\mathring U_{n}(T)).
\]
Using the same arguments as in the proof of Lemma 3.3 of \cite{athey_RF}, since $T_j$ are mutually uncorrelated, we have:
\begin{align*}
    \E\!\left[\left(U_{n}(T)-\mathring U_{n}(T)\right)^2\right]
&=
\sum_{j=2}^{k}
\left(\frac{k_j}{n_j}\right)^2
\binom{n}{j} \V(T_j)\\
&= \sum_{j=2}^{k}
\frac{k_j}{n_j}
\binom{k}{j} \V(T_j)\\
&\leq \frac{k_2}{n_2}
\sum_{j=2}^{k}
\binom{k}{j} \V(T_j)\\
&\leq \frac{k_2}{n_2}\Var(T)
\end{align*}

where:
\[
k_j = k\cdot (k-1)\cdots (k-j+1),
\]
Thus,
\[
\E\!\left[U_{n}(T)-\mathring U_{n}(T)\right]^2
\le
\frac{k(k-1)}{n(n-1)}\Var(T) \leq \frac{k^2}{n^2} \Var(T)
\]
The result follows by replacing $T$ by $H_k$ and using the uniform boundedness Assumption \ref{asm: h_bounded}
\end{proof}

\subsubsection{Variance Bounds}

\begin{lemma}\label{lemma: g_bounds}
    Under Assumption \ref{asm: sup bound on H diff}, for every $k$ and every
$1\le c\le k$,
\[
\left|
g_{k,c}(\xi_1',\dots,\xi_c')
-
\mathbb{E}\big[g_{k,c}(\xi_1,\dots,\xi_c)\big]
\right|
\le
K_1 \cdot \frac{c}{k} \qquad\P\text{-a.s.}
\]
Additionally,
\[
\pr^*\left(\left|
g_{k,c}(\xi_1',\dots,\xi_c')
-
\E^*\big[g_{k,c}(\xi_1,\dots,\xi_c)\big]
\right|\le
K_1 \cdot \frac{c}{k}\right) = 1
 \qquad\P\text{-a.s.}
\]
It follows that for any $p>0,\  1\leq c\leq k,\ \exists M_{p,c}<\infty$ such that:
\begin{enumerate}
    \item $k^p\E\left[\left|
g_{k,c}(\xi_1',\dots,\xi_c')
-
\mathbb{E}\big[g_{k,c}(\xi_1,\dots,\xi_c)\big]
\right|^p\right]<M_{p,c}$
\item  $k^p\E^*\left[\left|
g_{k,c}(\xi_1',\dots,\xi_c')
-
\Es\big[g_{k,c}(\xi_1,\dots,\xi_c)\big]
\right|^p\right]<M_{p,c}\quad \P$-a.s.
\end{enumerate}
\end{lemma}
\begin{proof}
    \begin{align*}
        \left|
g_{k,c}(\xi_1',\dots,\xi_c')
-
\mathbb{E}\big[g_{k,c}(\xi_1,\dots,\xi_c)\big]
\right| & = \big|\E[H_k(\xi'_1,...,\xi'_c,\xi_{c+1},...,\xi_k) - H_k(\xi_1,...,\xi_k) \mid \xi_1',...,\xi_c']\big|\\
&\leq  \E\left[\big|H_k(\xi'_1,...,\xi'_c,\xi_{c+1},...,\xi_k) - H_k(\xi_1,...,\xi_k) \big|\mid \xi_1',...,\xi_c'\right]\\
&\leq K_1\cdot \frac{c}{k}\qquad \P\text{-a.s}.
    \end{align*}

Similarly,
\begin{align*}
   & \pr^*\left( \left|
g_{k,c}(\xi_1',\dots,\xi_c')
-
\E^*\big[g_{k,c}(\xi_1,\dots,\xi_c)\big]
\right|\leq K_1\cdot \frac ck\right)\\
\geq\  &     \pr^*\left(\E^*\left[\big|H_k(\xi'_1,...,\xi'_c,\xi_{c+1},...,\xi_k) - H_k(\xi_1,...,\xi_k) \big|\mid \xi_1',...,\xi_c'\right]\leq K_1\cdot \frac ck\right) = 1 \qquad \P\text{-a.s.}
\end{align*}
\end{proof}

\begin{lemma}\label{lem: efron-stein}[Efron--Stein inequality, Theorem~3.1 in \cite{Boucheron2013}]

Let $X_1,\ldots,X_k \sim F$ be $k$ independent random elements, and let
$f(X_1,\ldots,X_k)$ be a measurable function such that
\[
\mathbb{E}_F\bigl[f(X_1,\ldots,X_k)^2\bigr] < \infty .
\]
Then,
\[
\Var_F\!\bigl(f(X_1,\ldots,X_k)\bigr)
\;\le\;
\frac12 \sum_{i=1}^k
\mathbb{E}_F\!\left[
\bigl(
f(X_1,\ldots,X_i,\ldots,X_k)
-
f(X_1,\ldots,X_i',\ldots,X_k)
\bigr)^2
\right],
\]
where $X_i'$ is an independent copy of $X_i$ for each $i=1,\ldots,k$.
\end{lemma}

As a consequence, we have a bound on the variance of the base learners:
\begin{lemma}\label{lemma: var_bound_Tk}
    Suppose $\xi_1,..,\xi_k \sim F$. Consider a measurable function $T_k$ satisfying the following assumptions:
    \begin{enumerate}
        \item $T_k$ is bounded a.s.
        \item $ |T_k(\xi_1,\xi_2,...,\xi_k) - T_k(\xi_1',\xi_2,...,\xi_k)| \leq c/k\quad  a.s$ for some constant c.
    \end{enumerate}
    Then, $\V_F(T_k) \leq  c^2/k$.
\end{lemma}

\subsection{Proofs of Section \ref{sec:main results}}

\subsubsection{Proofs of Section \ref{sec: main_results_complete_U}}
\begin{proof}[Proof of Theorem \ref{thm: clt_bagging_modified}]
\[\sqrt n (U_{n,k}-W_k) = \sqrt n (U_{n,k}-\U_{n,k})  + \sqrt n (\U_{n,k}-W_k) \]

    By Assumption \ref{asm: rate_of_k} and Lemma \ref{lem:athey_bound}, we get:
    \[\sqrt n k^{\epsilon}(U_{n,k} - \U_{n,k})\xrightarrow[]{p} 0\] 

     Thus,
        $$\sqrt n(U_{n,k} - W_k) = \sqrt n \frac{k}{n}\sum_{i=1}^n g_k(\xi_i) + o_p(1/k^{\epsilon})$$
        Define
\[
Z_{n,k} :=
\begin{cases}
    \dfrac{\sum_{i=1}^n g_k(\xi_i)}{\sqrt{n \, \Var(g_k(\xi))}} & \text{if } \Var(g_k(\xi)) > 0, \\[1.2em]
    \text{an arbitrary standard normal variable} & \text{if } \Var(g_k(\xi)) = 0.
\end{cases}
\]
which gives: 
$$\sqrt n(U_{n,k} - W_k) = k\sqrt{V(g_k(\xi))}Z_{n,k} +o(1/k^{\epsilon})$$
For the convergence of $Z_{n,k}$, along any sequence where $k^{2+\epsilon} \V(g_k(\xi))>1$, we verify the Lyapounov's condition:
\begin{align*}
\frac{n \, \E\!\left[ \lvert g_k(\xi) \rvert^{2+\epsilon} \right]}
     {\big( n \, \Var(g_k(\xi)) \big)^{1+\epsilon/2}} & \leq \frac{ M_{1,1}/k^{2+\epsilon}}
     {n^{\epsilon/2
     }\big(\Var(g_k(\xi)) \big)^{1+\epsilon/2}} \quad [\text{where $M_{1,1}$ is a constant as defined in Lemma \ref{lemma: g_bounds}}]\\
     &= \frac{M_{1,1}}
     {n^{\epsilon/2
     }\cdot k^{2+\epsilon}\big(k^{2+\epsilon}\Var(g_k(\xi)) \big)^{1+\epsilon/2} {(k^{2+\epsilon}})^{-1-\epsilon/2}} \\
     &\leq M_{1,1} \left(\frac{k^{2+\epsilon}}{n}\right)^{{\epsilon/2}} \to 0
\end{align*}

\end{proof}

 Before proceeding to the main result of Theorem \ref{thm: growing_k_main_theorem}, we first state some preliminary results:

\begin{lemma}\label{lemma: var-conv-helper}
 Consider the subsequence of $k$ for which $\V(g_k(\xi))k^{2+\epsilon}>1$. Along any such sequence, under Assumptions \ref{asm: h_bounded} and \ref{asm: rate_of_k}, we have:
\begin{enumerate}
    \item $k^{2+\epsilon}(\zeta^*_{k,1} - \zeta_{k,1}) \xrightarrow[]{p} 0$
       \item $\frac{\V^*(U_{n,k}^*)}{\V^*(\Uh)}\xrightarrow[]{p}1$
        \item $\frac{\V(U_{n,k})}{\V(\U_{n,k})}\xrightarrow[]{}1$ 
\end{enumerate}
\end{lemma}
\begin{proof}
\begin{enumerate}
    \item $k^{2+\epsilon}(\zeta^*_{k,1} - \zeta_{k,1}) \xrightarrow[]{p} 0$
  
Consider $\xi_1^*,...,\xi_{2k}^*$, i.i.d from $P_n$, and $\xi_1,...,\xi_{2k}$, i.i.d from $P$. Then,
    \begin{align*}
        \zeta^*_{k,1} &= cov^*(H_k(\xi_1^*,\xi^*_{2:k}), H_k(\xi_1^*,\xi^*_{k+1:2k-1}))\\
        &= \E^*\left[H_k(\xi_1^*,\xi^*_{2:k})\cdot H_k(\xi_1^*,\xi^*_{k+1:2k-1})\right] - \left(W_k^*\right)^2
    \end{align*}
    Similarly,
        \begin{align*}
        \zeta_{k,1} &= cov(H_k(\xi_1,\xi_{2:k}), H_k(\xi_1,\xi_{k+1:2k-1}))\\
        &= \E\left[H_k(\xi_1,\xi_{2:k})\cdot H_k(\xi_1,\xi_{k+1:2k-1})\right] - \left(W_k\right)^2
    \end{align*}
Now, we have:
\begin{align*}
    k^{2+\epsilon}\left|(W_k^*)^2-(W_k)^2\right| &= k^{2+\epsilon}(W_k^*-W_k)(W_k^*+W_k)\\
    & \leq 2K_0 k^{2+\epsilon}|(V_{n,k}-W_k)|\qquad \P\text{-a.s}.\qquad [\because \text{ Assumption } \ref{asm: h_bounded}]
\end{align*}
Now, using Theorems \ref{thm: diff_U_V} and \ref{thm: clt_bagging_modified}, since we are considering a sequence along which $k^{2+\epsilon}\V(g_k(\xi))>1$, we have:
\begin{align*}
    \sqrt n |V_{n,k}-W_k| &\leq \sqrt n \left|V_{n,k}-U_{n,k}\right| + \sqrt n \left|U_{n,k}-W_k\right|\\
    &\leq \sqrt n O(k^2/n)+\left|k\sqrt{\V(g_k(\xi))}Z_{n,k}+o_p(1/k^{\epsilon})\right|\\
    & = o(1) +O_p(1)\qquad [\because \text{By Lemma \ref{lemma: g_bounds}, $k^2\V(g_k(\xi))$ is bounded}]\\
    &= O_p(1)
\end{align*}
Thus, by Assumption \ref{asm: rate_of_k},
\begin{align*}
     k^{2+\epsilon}\left|(W_k^*)^2-(W_k)^2\right| &\leq \frac{2K_0k^{2+\epsilon}}{\sqrt n}\cdot O_p(1) \xrightarrow[]{p} 0
\end{align*}

We next show that:
\begin{align*}
    & k^{2+\epsilon}\left[\E^*\left[H_k(\xi_1^*,\xi^*_{2},...,\xi^*_{k})\cdot H_k(\xi_1^*,\xi^*_{k+1},...,\xi^*_{2k-1})\right] - \E\left[H_k(\xi_1,\xi_{2},...,\xi_{k})\cdot H_k(\xi_1,\xi_{k+1},...,\xi_{2k-1})\right]\right]\\
    = &\  o_p(1)
\end{align*}

Define the following quantities:
\begin{align}
 t(\xi_1,...,\xi_{2k}) &:= H_k(\xi_1,...,\xi_k)H_k(\xi_1,\xi_{k+1},...,\xi_{2k-1})
\end{align}
Let $T^{cov}_k$ be the symmetric version of the above. Observe that $$V_n(T^{cov}_k): = \E^*\left[H_k(\xi_1^*,\xi^*_{2},...,\xi^*_{k})\cdot H_k(\xi_1^*,\xi^*_{k+1},...,\xi^*_{2k-1})\right]$$ is the $V$ statistic of the data $(\xi_1,...,\xi_n)$ with the kernel $T$. Let $U^{cov}_{n,k}$ be the corresponding $U$ statistic with kernel $T^{cov}_k$, and
$$\theta^{cov}_{n,k}:=\E\left[H_k(\xi_1,\xi_{2},...,\xi_{k})\cdot H_k(\xi_1,\xi_{k+1},...,\xi_{2k-1})\right]$$
Thus, it is enough to show:
\begin{equation}\label{eqn: U_V cov convergence}
    k^{2+\epsilon}|V_n(T^{cov}_k) - \theta^{cov}_{n,k}| \leq  k^{2+\epsilon}|U_n(T^{cov}_k) - \theta^{cov}_{n,k}| + k^{2+\epsilon}|V_n(T^{cov}_k) - U_n(T^{cov}_k)| = o_p(1)
\end{equation}
From Assumption \ref{asm: h_bounded} and Theorem \ref{thm: diff_U_V}, we have:
\begin{equation}\label{eqn: var_conv_eq1}
k^{2+\epsilon}|V_n(T^{cov}_k) - U_n(T^{cov}_k)| = o_p(1)    
\end{equation}
By Assumption \ref{asm: h_bounded} and Theorem \ref{thm: CLT_U_stat_modified}, we have:
\begin{align*}
    k^{4+2\epsilon}\E\left[U_n(T^{cov}_k) - \theta^{cov}_{n,k}\right]^2 & =  O\left(\frac {k^{5+2\epsilon}}n \cdot \V(T^{cov}_{k})\right)
\end{align*}
Now, we claim that:
\begin{claim}\label{claim: Gamma_kk bound}
    $\V(T^{cov}_{k}) = O(1/k)$
\end{claim}
Thus, $$O\left(\frac {k^{5+2\epsilon}}n \cdot \V(T^{cov}_{k})\right) = O\left(\frac {k^{4+2\epsilon}}n\right) = o(1)$$
which ensures \begin{equation}\label{eqn: var_conv_eq2}
    k^{2+\epsilon}|U_n(T^{cov}_k) - \theta^{cov}_{n,k}| = o_p(1)
\end{equation}
Finally, from \ref{eqn: var_conv_eq1} and \ref{eqn: var_conv_eq2} we get \ref{eqn: U_V cov convergence} which completes the proof. It remains to show Claim \ref{claim: Gamma_kk bound}.
\begin{proof}[Proof of Claim \ref{claim: Gamma_kk bound}]
Define:
\begin{equation}
    \tilde t(\xi_1,...,\xi_{2k}) := H_k(\xi_1,...,\xi_k)H_k(\xi_{k+1},\xi_{k+2}...,\xi_{2k}) 
\end{equation}
   Similarly, let $\tilde T^{cov}_k$ be the symmetric version of the above.

   Now, by Assumptions \ref{asm: h_bounded} and \ref{asm: sup bound on H diff}, for some constant $C >0$, we have:
       \begin{align*}
        &|t - \tilde t| \leq |H_k(\xi_1,...,\xi_k)|\cdot |H_k(\xi_1,\xi_{k+1}...,\xi_{2k-1}) - H_k(\xi_{2k},\xi_{k+1}...,\xi_{2k-1})|  \leq C/k\\
        \implies\ & |T^{cov}_k - \tilde T^{cov}_k| \leq  C/k\\
        \implies\ & \V(T^{cov}_k - \tilde T^{cov}_k) = o(1/k)
    \end{align*}
Let $\sigma^2_k = \V(H_k(\xi_1,...,\xi_k))$. Then,
\begin{align*}
    \V(\tilde T^{cov}_k) \leq V(\tilde t) &= \V[H_k(\xi_1,...,\xi_{k})H_k(\xi_{k+1},...,\xi_{2k})]\\
    & = \E[H_k(\xi_1,...,\xi_{k})^2]^2 - \E[H_k(\xi_1,...,\xi_{k})]^4\\
    & = (\sigma_k^2+W_k^2)^2 - (W_k^2)^2\\
    & = \sigma_k^2(\sigma_k^2+2W_k^2)  = O(\sigma_k^2) = O(1/k)
\end{align*}
where the last equality holds from Lemma \ref{lemma: var_bound_Tk} using Assumptions \ref{asm: h_bounded} and \ref{asm: sup bound on H diff}.

    Finally,
    \begin{align*}
        \V(T^{cov}_k) &\leq \V(\tilde T^{cov}_k) + \V(\tilde T^{cov}_k - T^{cov}_k) +2\sqrt{\V(\tilde T^{cov}_k) \cdot \V(\tilde T^{cov}_k - T^{cov}_k)}\\
        &= O(1/k) +o(1/k)
    \end{align*}
\end{proof}
This completes the proof of Lemma \ref{lemma: var-conv-helper}.1.

\item $\frac{\V^*(U_{n,k}^*)}{\V^*(\Uh)}\xrightarrow[]{p}1$

By Lemma \ref{lemma: var-conv-helper}.1, $k^{2+\epsilon}(\zeta^*_{k,1}-\zeta_{k,1}) \xrightarrow[]{p} 0$ with $k^{2+\epsilon}\zeta_{k,1} = k^{2+\epsilon}\V(g_k(\xi))>1$. Thus, $$\pr(k^{2+\epsilon}\zeta^*_{k,1}>0.5)\to 1$$
\begin{align*}
    \pr\left(\left|\frac{\V^*(U_{n,k}^*)}{\V^*(\Uh)}-1\right|>\delta\right) & \leq   \pr\left(\left|\frac{\V^*(U_{n,k}^*)}{\V^*(\Uh)}-1\right|>\delta\ \mid k^{2+\epsilon}\zeta^*_{k,1}>0.5\ \right) + \pr\left( k^{2+\epsilon}\zeta^*_{k,1}\leq 0.5\ \right)
\end{align*}
where $\pr\left( k^{2+\epsilon}\zeta^*_{k,1}\leq 0.5\ \right) \to 0$. Thus, it is enough to show:
\begin{equation}\label{eq: conditional_zeta^* geq 1}
    \pr\left(\left|\frac{\V^*(U_{n,k}^*)}{\V^*(\Uh)}-1\right|>\delta\ \mid k^{2+\epsilon}\zeta^*_{k,1}>0.5\ \right) \to 0
\end{equation}

Again, from Theorem 12.3 of \cite{Vaart1998}, we can write
\begin{equation}\label{eq: variance star U_n,k}
\V^*(U_{n,k}^*) = \sum_{c=1}^k \binom{n}{k}^{-1} \binom{k}{c} \binom{n-k}{k-c} \, \zeta^*_{k,c} 
\end{equation}
and 
\begin{equation}
\V^*(\Uh) = \frac{k^2}{n}\zeta^*_{k,1}
\end{equation}

Now, conditional on $ k^{2+\epsilon}\zeta^*_{k,1}> 0.5$, for $c = 1$ in \ref{eq: variance star U_n,k},
\[
\frac{\big(k / \binom{n}{k}\big) \binom{n-k}{k-1} \zeta^*_{k,1}}{(k^2 / n)\zeta^*_{k,1}}
= \frac{(n-k)!(n-k)!}{(n-1)!(n-2k+1)!} \to 1 
\]
To continue, 
$$\frac{\sum_{c=2}^k \binom{n}{k}^{-1} \binom{k}{c} \binom{n-k}{k-c} \, \zeta^*_{k,c}}{(k^2/n)\,\zeta^*_{k,1}} \leq \frac{k\zeta^*_{k,k}}{n\,\zeta^*_{k,1}} \cdot \frac{1 - \epsilon_n^{\,k-1}}{1 - \epsilon_n}\cdot \frac{n}{n-k+1}$$
where
\[
\epsilon_n = \frac{k^2}{n-k+1} \to 0 \quad \text{ by Assumption \ref{asm: rate_of_k}}
\]
Again, with $ k^{2+\epsilon}\zeta^*_{k,1}> 0.5$, using $k^{3+\epsilon}/n \to 0$ from Assumption \ref{asm: rate_of_k} we have:
\begin{align*}
    &\frac{k\zeta^*_{k,k}}{n\,\zeta^*_{k,1}} \cdot \frac{1 - \epsilon_n^{\,k-1}}{1 - \epsilon_n}\cdot \frac{n}{n-k+1}\\
    =\ & \frac{\zeta^*_{k,k}}{k^{2+\epsilon}\,\zeta^*_{k,1}} \cdot \frac{1 - \epsilon_n^{\,k-1}}{1 - \epsilon_n}\cdot \frac{n}{n-k+1}\frac{k^{3+\epsilon}}{n} = o(\zeta^*_{k,k}) \xrightarrow[]{a.s} 0\numberthis \label{eq: star_last_Unk}
\end{align*}
Where the last limit holds since $\zeta^*_{k,k} = \V^*(H_k)\leq K_0^2\ \P$-a.s.  by Assumption \ref{asm: h_bounded}.

\end{enumerate}  

\end{proof}

We next move on to prove Theorem \ref{thm: growing_k_main_theorem}:
\begin{proof}[Proof of Theorem \ref{thm: growing_k_main_theorem}] 
For any subsequence along which $k^{2+\epsilon}\V(g)k(\xi))>1$, we will first show that:
    \begin{equation} \label{eqn: U_star to V conv}
        \left|\pr^*\left(\frac{U_{n,k}^* - V_{n,k}}{\sigma^*_{n,k}}\leq x\right) - \Phi(x)\right| \xrightarrow[]{p} 0
    \end{equation}
     Conditional on the data, the Hajek projection of $U_{n,k}^*$ is given by $\Uh$, where
    \begin{equation}
        \Uh - V_{n,k} = \sum_{i = 1}^n \E^*[U_{n,k}^* - V_{n,k} \mid \xi_i^*] = \frac 1n \sum_{i=1}^n k[g_k^*(\xi_i^*) -V_{n,k}] 
    \end{equation}
    Now, conditional on the data, $k[g_k^*(\xi_i^*) -V_{n,k}]$ are i.i.d, with:
    \begin{align*}
    \E^*[k(g_k^*(\xi_i^*) -V_{n,k})] &= 0\\
        \E^*[k^2[g_k^*(\xi_i^*) -V_{n,k}]^2] &= k^2\zeta^*_{k,1} \\
        \E^*[k^3|g_k^*(\xi_i^*) -V_{n,k}|^3]&=k^3{\rho^*_{n,k}}
    \end{align*}
        Define:
    \begin{equation}
        {\sigma^*}^2_{n,k} := \frac{k^2\zeta^*_{k,1}}{n} = \V^*(\Uh)
    \end{equation}
        By Lemma \ref{lemma: var-conv-helper}, $k^{2+\epsilon} (\zeta^*_{k,1}- \zeta_{k,1}) \xrightarrow[]{p} 0$. Now, we are considering the subsequence along which: 
\begin{equation}\label{eqn: var_threshold}
    k^{2+\epsilon}\zeta_{k,1} = k^{2+\epsilon}\V(g_k(\xi))>1
\end{equation} 
Thus, along that subsequence we also have:
\begin{equation}\label{eqn: zeta_ratio_to_1}
    \frac{k^{2+\epsilon}\zeta^*_{k,1}}{k^{2+\epsilon}\zeta_{k,1}}  \xrightarrow[]{p}  1 \implies \frac{k^{2}\zeta^*_{k,1}}{k^{2}\zeta_{k,1}} \xrightarrow[]{p} 1 \implies \frac{n{\sigma^*}^2_{n,k}}{k^2\V(g_k(\xi))}  \xrightarrow[]{p} 1
\end{equation}
which ensures (\ref{eqn: growing_k_main_theorem: variance limit}). \text{From Lemma }(\ref{lemma: g_bounds}), we have $k^2\zeta^*_{k,1} <M_{2,1}$ and $k^3\rho^*_{n,k}<M_{3,1}$ $\P\text{-a.s}$.  Thus, by Barry-Esseen theorem:
\begin{align*}
&\left| \pr^*\!\left( 
   \frac{\Uh - V_{n,k}}{\sigma^*_{n,k}} \;\leq\; x 
   \right) - \Phi(x) \right|\\
=\ &\left| \pr^*\!\left( 
   \frac{\sqrt{n}(\Uh - V_{n,k})}{\sqrt{k^2 \zeta_{k,1}^*}} \;\leq\; x 
   \right) - \Phi(x) \right|\\
\leq\ &
   \frac{C k^3 \rho^{*}_{n,k}}
        {\sqrt{n}\,\big(k^2 \zeta^*_{k,1}\big)^{3/2}} \leq \frac{C M_{3,1}}
        {\sqrt{n}\,\big(k^2 \zeta^*_{k,1}\big)^{3/2}}\xrightarrow[]{p} 0\numberthis\label{eqn: growing_k_berry-esseen}
\end{align*}

The final convergence (\ref{eqn: growing_k_berry-esseen}) holds  because, by using \ref{eqn: zeta_ratio_to_1}, \ref{eqn: var_threshold} and Assumption \ref{asm: rate_of_k}, we have:
\begin{equation}
    \sqrt n \left(k^2\zeta^*_{k,1}\right)^{3/2} =  \left(k^{2+\epsilon}\zeta_{k,1}\right)^{3/2} \cdot \frac{\left(k^{2}\zeta^*_{k,1}\right)^{3/2}}{\left(k^{2}\zeta_{k,1}\right)^{3/2}}\cdot\frac{\sqrt n}{k^{3\epsilon/2}} \to \infty
\end{equation}
We next show that:
\begin{equation}\label{eq: growing_k_main U minus hajek U converjence}
    \pr^*\!\left( 
   \left|\frac{U_{n,k}^* - \Uh }{\sigma^*_{n,k}}\right| > \delta 
   \right) \xrightarrow[]{p} 0
\end{equation}
It is enough to show:

\begin{equation}
    \E^*\!\left( 
   \frac{U_{n,k}^* - \Uh }{\sigma^*_{n,k}}\right)^2 \xrightarrow[]{p} 0
\end{equation}
Here, by construction ${\sigma^*}^2_{n,k} = \V^*(\Uh)$. Let $\tilde{\sigma^*}^2_{n,k} = \V^*(U^*_{n,k})$. Since $\tilde\sigma^*_{n,k}$ and $\sigma^*_{n,k}$ are constants given the data, we proceed as follows:
\begin{align*}
     \E^*\left(\frac{U^*_{n,k}-\Uh}{\sigma^*_{n,k}}\right)^2 & =  \E^*\left(\frac{U^*_{n,k}-V_{n,k}}{\tilde\sigma^*_{n,k}}\cdot 
     \frac{\tilde\sigma^*_{n,k}}{\sigma^*_{n,k}} - \frac{\Uh-V_{n,k}}{\sigma^*_{n,k}}\right)^2\\
     &= \left(    \frac{\tilde\sigma^*_{n,k}}{\sigma^*_{n,k}} \right)^2\E^*\left(\frac{U^*_{n,k}-V_{n,k}}{\tilde\sigma^*_{n,k}}\right)^2 + \E^*\left(\frac{\Uh-V_{n,k}}{\sigma^*_{n,k}}\right)^2\\
     &\quad -2  \left(    \frac{\tilde\sigma^*_{n,k}}{\sigma^*_{n,k}} \right)\cdot\E^*\left(\frac{U^*_{n,k}-V_{n,k}}{\tilde\sigma^*_{n,k}}\cdot \frac{\Uh-V_{n,k}}{\sigma^*_{n,k}}\right)\\
     & = \left(    \frac{\tilde\sigma^*_{n,k}}{\sigma^*_{n,k}} \right)^2 +1 - 2  \left(    \frac{\tilde\sigma^*_{n,k}}{\sigma^*_{n,k}} \right)\cdot\E^*\left(\frac{U^*_{n,k}-V_{n,k}}{\tilde\sigma^*_{n,k}}\cdot \frac{\Uh-V_{n,k}}{\sigma^*_{n,k}}\right)\\
     &= \left(    \frac{\tilde\sigma^*_{n,k}}{\sigma^*_{n,k}} \right)^2 +1 -  2  \left(    \frac{\tilde\sigma^*_{n,k}}{\sigma^*_{n,k}} \right)\cdot \left(\frac{{\sigma^*_{n,k}}^2}{\tilde\sigma^*_{n,k}\cdot \sigma^*_{n,k}}\right)\numberthis \label{eq: hajek_bound}\\
&= \left(    \frac{\tilde\sigma^*_{n,k}}{\sigma^*_{n,k}} \right)^2 +1 -  2  \xrightarrow[]{p} 0\quad \text{ (By Lemma \ref{lemma: var-conv-helper}.2)}
\end{align*}

Here, eqn (\ref{eq: hajek_bound}) holds because $\Uh$ is the Hajek Projection and $$\E^*\left((U^*_{n,k}-\Uh)\cdot (\Uh-V_{n,k})\right) = 0 \implies \E^*\left((U^*_{n,k}-V_{n,k})\cdot (\Uh-V_{n,k})\right) = \E^*\left(\Uh-V_{n,k}\right)^2$$

Thus, (\ref{eqn: U_star to V conv}) follows by applying Theorem A\ref{thm: conditional_Slutsky} (i) on \ref{eq: growing_k_main U minus hajek U converjence} and \ref{eqn: growing_k_berry-esseen}. Next, we have:

\begin{equation}\label{eqn: theorem main growing k U,V gap}
    \pr^*\left(\frac{|U_{n,k} -V_{n,k} |}{\sigma^*_{n,k}}>\delta\right) \xrightarrow[]{p} 0
\end{equation}

where (\ref{eqn: theorem main growing k U,V gap}) holds because:
\begin{align*}
     \frac{1}{{\sigma^*}^2_{n,k}}\cdot \E^*\left(|U_{n,k} -V_{n,k} |^2\right) &= \frac{O(k^4/n^2)}{{\sigma^*}^2_{n,k}} \qquad [\because \text{Theorem \ref{thm: diff_U_V} and Asm \ref{asm: h_bounded}}]\\
     &= O(k^{4+\epsilon}/n)\cdot \frac{k^{2}\V(g_k(\xi))}{n\cdot {\sigma^*}^2_{n,k}}\cdot \frac{1}{k^{2+\epsilon}\V(g_k(\xi))} \xrightarrow[]{p} 0
\end{align*}
by Assumption \ref{asm: rate_of_k}, eqn \ref{eqn: growing_k_main_theorem: variance limit} and the hypothesis of Assumption \ref{thm: growing_k_main_theorem}. Finally, the conclusion of Theorem \ref{thm: growing_k_main_theorem}, holds by applying Slutsky (Theorem A\ref{thm: conditional_Slutsky}) on \ref{eqn: U_star to V conv} and \ref{eqn: theorem main growing k U,V gap}.
\end{proof}

\begin{proof}[Proof of Corollary \ref{corr: growing_k_main_CI}] Let 
      $$ I_{n,k,\infty,B} \subseteq \tilde I_{n,k,\infty,B} = [U_{n,k} \pm (t_{B,{1-\alpha/2}}\cdot S_{n,k,\infty,B} + \frac{1}{k^{\epsilon/2}\sqrt n})]$$
    Here, $\frac{1}{k^{\epsilon/2}\sqrt n} = o(1/\sqrt n)$.
    Consider the subsequence long which $k^{2+\epsilon}\V(g_k(\xi))>1$. Along that subsequence, from Theorem \ref{thm: clt_bagging_modified}, we get:
    \begin{equation}\label{eq: growing_k_CLT}
\frac{\sqrt n(U_{n,k} - W_k)}{k\sqrt{\V(g_k(\xi))}} = Z_{n,k} + \frac{o(1)}{k^{\epsilon/2}\sqrt{k^{2+\epsilon}\V(g_k(\xi))}} \xrightarrow[]{d} N(0,1)        
    \end{equation}

    Thus using Theorem \ref{thm: main_cheap_boot} on (\ref{eq: growing_k_CLT}) above and \ref{eqn: Full_U bootstrap convergence} (from Theorem \ref{thm: growing_k_main_theorem}), we get that:
    $$\lim_{n\to \infty}\  \pr(W_k \in  I_{n,k,\infty,B}) =  1-\alpha \implies \liminf_{n\to \infty}\  \pr(W_k \in  \tilde I_{n,k,\infty,B}) \geq  1-\alpha$$
    
Now consider the complement subsequence along which $k^{2}\V(g_k(\xi))\leq 1/{k^{\epsilon}}$. Along this subsequence, we have: $k^{\epsilon/2}\sqrt n(U_{n,k} - W_k)  =o_p(1) $. Thus, 
$$\pr(W_k \in \tilde I_{n,k,\infty,B}) \geq \pr(W_k \in U_{n,k} \pm \frac{1}{\sqrt nk^{\epsilon/2}}) = \pr(\sqrt nk^{\epsilon/2}|W_k - U_{n,k}| \leq 1)\to 1$$
\end{proof}

\subsubsection{Proofs of Section \ref{sec:main_result_U_statistic}}\label{sec:proof main}
Let $N := \binom{n}{k}$. Conditional on the observed data $(\xi_1,\ldots,\xi_n)$, consider the finite population
\[
\left\{ H_k(\xi_{B}) : B \subset \{1,\ldots,n\},\, |B|=k \right\},
\]
which consists of $N$ elements. Let $F_n$ denote the uniform distribution on this finite set. Let $v_1,\ldots,v_{D_n}$ be independent draws from $F_n$. The corresponding incomplete U-statistic is defined as
\begin{equation}\label{eqn: U_nkd_defn}
    U_{n,k,D_n}
\;:=\;
\frac{1}{D_n}\sum_{i=1}^{D_n} v_i .
\end{equation}

Where $U_{n,k} = \E_{F_n}[v_i]$. Note that the distribution $F_n$ is completely determined by the observed data $(\xi_1,\ldots,\xi_n)$.

% Requires:
% \usepackage{tikz}
% \usetikzlibrary{positioning,arrows.meta,fit}

% Requires:
% \usepackage{tikz}
% \usetikzlibrary{positioning,arrows.meta}

\begin{tikzpicture}[
    font=\scriptsize,
    >=Latex,
    node distance=1.0cm and 1.2cm,
    box/.style={draw, rounded corners, align=left, inner sep=5pt, minimum width=3.2cm, minimum height=2.2cm},
    big/.style={draw, rounded corners, align=left, inner sep=6pt, minimum width=4.3cm, minimum height=2.6cm},
    arr/.style={->, thick, shorten >=2pt, shorten <=2pt}
]

%---------------------------
% Left: True distribution
%---------------------------
\node[box] (P) {
\textbf{True distribution}
$\P$\\[5pt]
Data: $\xi_1,\ldots,\xi_n \overset{\text{i.i.d.}}{\sim} \P$
};

%---------------------------
% Middle: Empirical distribution
%---------------------------
\node[big, right=1.6cm of P] (Pn) {
\textbf{Empirical distribution}\\
$P_n := \frac1n\sum_{i=1}^n \delta_{\xi_i}$\\[3pt]
$\E^*[\cdot] := \E[\cdot \mid P_n]$\\
$\pr^*[\cdot] := \pr(\cdot \mid P_n)$\\
$U_{n,k}, V_{n,k}$ are $P_n$ measurable
};

\draw[arr] (P.east) -- (Pn.west);

%---------------------------
% Bottom: F_n from data
%---------------------------
\node[big, below=1.3cm of Pn, minimum height=2.4cm] (Fn) {
\textbf{Finite-Bagging law}\\
$N=\binom{n}{k}$ subsamples from $\{\xi_i\}$\\
$F_n =$ uniform over $\{T_k(\xi_{i_1},\ldots,\xi_{i_k})\}$\\[3pt]
Draw: $v_i \sim F_n$\\
(\,$F_n$ fully determined by $P_n$\,)
};

\draw[arr] (Pn.south) -- (Fn.north);

%---------------------------
% Right-top: Bootstrap resample
%---------------------------
\node[box, right=1.6cm of Pn] (boot) {
\textbf{Bootstrap resample}\\
$\xi_1^*,\ldots,\xi_n^* \overset{\text{i.i.d.}}{\sim} P_n$\\[2pt]
$P_n^* := \frac1n\sum_{i=1}^n \delta_{\xi_i^*}$\\
$U^*_{n,k}, V^*_{n,k}$ are $P^*_n$ measurable
};

\draw[arr] (Pn.east) -- node[above, align=center]{\scriptsize bootstrap} (boot.west);

%---------------------------
% Right-bottom: F_n^* from bootstrap resample
%---------------------------
\node[box, below=1.3cm of boot] (Fnstar) {
\textbf{Bootstrap Finite-Bagging law}\\
$N=\binom{n}{k}$ subsamples from $\{\xi_i^*\}$\\
$F_n^* =$ uniform over $\{T_k(\xi_{i_1}^*,\ldots,\xi_{i_k}^*)\}$\\[3pt]
Draw: $v_i^* \sim F_n^*$\\
(\,$F_n^*$ fully determined by $P_n^*$\,)
};

\draw[arr] (boot.south) -- (Fnstar.north);

\end{tikzpicture}

Similarly, define $F_n^*$ for a resample drawn from the original sample, and $v_i^* \sim F_n^*$. We next have the following lemma:
\begin{lemma}\label{lemma: var_bound F_n}
    Under Assumptions \ref{asm: h_bounded} and \ref{asm: sup bound on H diff}, for some constant $c$, independent of the data, we have: $$\V_{F_n}(v_i)  \leq c/k \qquad \P\text{-a.s}$$
\end{lemma}
\begin{proof}
    \begin{align*}
        \V_{F_n}(v_i) &= \E_{F_n}[v_i^2] - \E_{F_n}[v_i]^2\\
        &= \frac1N\sum_{i=1}^N v_i^2 - U_{n,k}^2\\
        &= U_n(H_k^2) - U_n(H_k)^2\\
        & \leq \left| U_n(H_k^2) - V_n(H_k^2)\right| + |U_n(H_k)^2 - V_n(H_k)^2| + |V_n(H_k^2) - V_n(H_k)^2|
    \end{align*}

    Now, by Assumptions \ref{asm: h_bounded} and \ref{asm: sup bound on H diff}, using Lemma \ref{lemma: var_bound_Tk}, for some constant $c_1$, independent of the data, we have:
    \begin{equation}
        |V_n(H_k^2) - V_n(H_k)^2| = \V_{P_n}(H_k) \leq c_1/k 
    \end{equation}
    Now, Assumption \ref{asm: h_bounded} implies both $H_k$, $H_k^2$ and $| U_n(H_k) + V_n(H_k)|$ are all bounded uniformly, $\P$-\text{a.s.} Thus, the result follows by Assumption \ref{asm: rate_of_k} and Theorem \ref{thm: diff_U_V}, as:
    \begin{equation}
        \left| U_n(H_k^2) - V_n(H_k^2)\right| \leq \frac{c_2}{k}\qquad \text{ and }\qquad  \left| U_n(H_k)^2 - V_n(H_k)^2\right| \leq \frac{c_3}{k}
    \end{equation}
\end{proof}

\begin{lemma}\label{lemma: subsampled D_n}
    Suppose Assumptions \ref{asm: h_bounded} - \ref{asm: D_n growth} hold. Consider $\epsilon$ as defined in Assumption \ref{asm: rate_of_k}. Then we have:
    \begin{enumerate}
        \item $\sqrt n k^{\epsilon}|U_{n,k,D_n} - U_{n,k}| = o_p(1)$.
        \item For any subsequence along which $k^{2+\epsilon}\V(g_k(\xi))>1$, we have:
            \begin{equation} 
                \Es\left[\left(\frac{U_{n,k,D_n} - U_{n,k}}{\sigma^*_{n,k}}\right)^2\ \right] \xrightarrow[]{p} 0
            \end{equation}
           \item For any subsequence along which $k^{2+\epsilon}\V(g_k(\xi))>1$, we have:
            \begin{equation} 
                \Es\left[\left(\frac{U^*_{n,k,D_n} - U^*_{n,k}}{\sigma^*_{n,k}}\right)^2\right] \xrightarrow[]{p} 0
            \end{equation}
    \end{enumerate}
 where $\sigma^*_{n,k}$ is as defined in Theorem \ref{thm: growing_k_main_theorem}.
\end{lemma}

\begin{proof} We have:
\begin{equation}
    \Es\left[(U_{n,k,D_n}-U_{n,k})^2\right] = \Es\left[\left(\frac{1}{D_n}\sum_{i=1}^{D_n}(v_i-U_{n,k})\right)^2\right]=\frac{1}{D_n} \V_{F_n}(v_i)
           \lesssim  \frac{1}{kD_n}\qquad \P\text{-a.s.}\\
\end{equation}
    
   \begin{enumerate}
       \item Now,
       \begin{align*}
           \E\left[nk^{2\epsilon}(U_{n,k,D_n}-U_{n,k})^2\right] &= nk^{2\epsilon}\E\left[\Es\left[(U_{n,k,D_n}-U_{n,k})^2\right]\right]\\
           &\lesssim \frac{n}{D_n}\cdot \frac {k^{2\epsilon}}k\to 0
       \end{align*}
       \item Similarly,
              \begin{align*}
           \Es\left[\left(\frac{U_{n,k,D_n} - U_{n,k}}{\sigma^*_{n,k}}\right)^2\ \right] &= \frac{1}{{\sigma^*_{n,k}}^2}\Es\left[(U_{n,k,D_n}-U_{n,k})^2\right]\\
           &\lesssim \frac{1}{k^{1-\epsilon}}\cdot\frac{n}{k^{2+\epsilon}\V(g_k(\xi))D_n}\cdot \frac {k^2\V(g_k(\xi))}{n{\sigma^*_{n,k}}^2}\xrightarrow[]{p} 0
       \end{align*}
       \item Again,
       \begin{align*}
            \Es\left[(U^*_{n,k,D_n}-U^*_{n,k})^2\right] &= \Es\left[\E^*\left[\left(\frac{1}{D_n}\sum_{i=1}^{D_n}(v^*_i-U^*_{n,k})\right)^2 \bigm| F_n^*\right]\right]\\
            &= \frac{1}{D_n} \E^*\left[\V_{F^*_n}(v^*_i)\right]
       \end{align*}
       Now, similar to the proof of Lemma \ref{lemma: var_bound F_n},
       \begin{align*}
            \E^*\left[\V_{F^*_n}(v^*_i)\right]&\leq \E^*\left[\left| U^*_n(H_k^2) - V^*_n(H_k^2)\right| + |U^*_n(H_k)^2 - V^*_n(H_k)^2| + |V^*_n(H_k^2) - V^*_n(H_k)^2|\right]
       \end{align*}
       Since $P_n$ is always a discrete distribution on a finite support, for any $P_n$ satisfying the conditions of Theorem \ref{thm: diff_U_V} and Lemma \ref{lemma: var_bound F_n}, any resampled $P_n^*$ generated from $P_n$ would also satisfy the assumptions. Thus, suppose $\xi_i^*\sim P_n^*$, then by Assumption $\ref{asm: h_bounded}$ and \ref{asm: sup bound on H diff}, the following statements hold $\P$\text{-a.s.}
       \[
       \pr^*\Big[|T_k(\xi_1^*,...,\xi_k^*)|\leq B\quad P_n^*\text{-a.s}\Big] = 1 
       \]
       where $T_k$ is $H_k$ or $H_k^2$; and, 
       \[
\ps\left[\sup_{\xi_i, \xi_i'\ \in\  \Xi}
\Big|
H_k(\xi_1'^*,\dots,\xi_c'^*,\xi_{c+1}^*,\dots,\xi_k^*)
-
H_k(\xi_1^*,\dots,\xi_c^*,\xi_{c+1}^*,\dots,\xi_k^*)
\Big|
\le
K_1 \cdot\frac{c}{k}\quad  P_n^*\text{-a.s}\right]
 = 1
\]
Thus, similar to the proof of Lemma \ref{lemma: var_bound F_n},  $\E^*\left[\V_{F^*_n}(v^*_i)\right] \lesssim  1/k\ \P$-a.s. The remainder of the proof then follows similar to part~(2) of Lemma~\ref{lemma: subsampled D_n}.
   \end{enumerate}
\end{proof}

\begin{proof}[Proof of Theorem \ref{thm: growing_k_main_CI_D_n}]
    Consider any subsequence along which $k^{2+\epsilon}V(g_k(\xi))>1$. Along that subsequence, using Theorem \ref{thm: clt_bagging_modified} and Lemma \ref{lemma: subsampled D_n} (1), we have:
    
    \begin{align*}
        \frac{\sqrt n (U_{n,k,D_n} - W_k)}{k\sqrt{V(g_k(\xi))}} &=  \frac{\sqrt n(U_{n,k
        } - W_k)}{k\sqrt{V(g_k(\xi))}} +  \frac{\sqrt n (U_{n,k,D_n} - U_{n,k})}{k\sqrt{V(g_k(\xi))}}\\
        &= Z_{n,k}+\frac{o_p(1)}{k^{\epsilon/2}\sqrt{k^{2+\epsilon}V(g_k(\xi))}} \xrightarrow[]{d} N(0,1)\numberthis \label{eqn: corr_2_main_CLT}
    \end{align*}
Next, we show that:
\begin{equation}\label{eqn: corr_2_conditional_CLT}
    \pr^*\left(\frac{U^*_{n,k,D_n} - U_{n,k,D_n}}{\sigma^*_{n,k}}\leq x\right)  \xrightarrow[]{p} \Phi(x)
\end{equation}

Now,
    \begin{align*}
        &\pr^*\left(\frac{U^*_{n,k,D_n} - U_{n,k,D_n}}{\sigma^*_{n,k}}\leq x \right)  \\
        =\ &  \pr^*\left(\frac{U^*_{n,k,D_n} - U^*_{n,k}}{\sigma^*_{n,k}} + \frac{U^*_{n,k} - U_{n,k}}{\sigma^*_{n,k}} + \frac{U_{n,k} - U_{n,k,D_n}}{\sigma^*_{n,k}}\leq x\right)
    \end{align*}

 Now, by Lemma \ref{lemma: subsampled D_n} (2), we have that for any $\delta>0$,
 \begin{equation}\label{eqn: cor_2_eq1}
     \pr^*\left[\left|\frac{U^*_{n,k,D_n} - U_{n,k}^*}{\sigma^*_{n,k}} + \frac{U_{n,k} - U_{n,k,D_n}}{\sigma^*_{n,k}}\right|>\delta\right] \xrightarrow[]{p} 0
 \end{equation}
And, from Theorem \ref{thm: growing_k_main_theorem} we have:
\begin{equation}\label{eqn: cor_2_eq2}
    \left|\pr^*\left( \frac{U^*_{n,k} - U_{n,k}}{\sigma^*_{n,k}}  \leq x\right) - \Phi(x)\right| \xrightarrow[]{p} 0
\end{equation}
Thus, using Theorem \ref{thm: conditional_Slutsky} on (\ref{eqn: cor_2_eq1} and \ref{eqn: cor_2_eq2}) gives \ref{eqn: corr_2_conditional_CLT}. Now, applying Theorem~\ref{thm: main_cheap_boot} to \eqref{eqn: corr_2_main_CLT} and \eqref{eqn: corr_2_conditional_CLT}, the remainder of the argument proceeds exactly as in the proof of Corollary~\ref{corr: growing_k_main_CI}.
    
\end{proof}

\subsubsection{Proofs of Section \ref{sec: main_result_V_statistic}}
We proceed similar to the proofs of Section \ref{sec:main_result_U_statistic}. Let $M := n^k$. Conditional on the observed data $(\xi_1,\ldots,\xi_n)$, consider the finite population
\[
\left\{
H_k(\xi_{i_1},\ldots,\xi_{i_k})
:
(i_1,\ldots,i_k)\in\{1,\ldots,n\}^k
\right\},
\]
which consists of $M$ elements (allowing repetitions). Let $\widetilde F_n$ denote the uniform distribution on this finite set. Let $\tilde v_1,\ldots,\tilde v_{D_n}$ be independent draws from $\widetilde F_n$. The corresponding incomplete V-statistic is defined as
\begin{equation}\label{eqn: V_nkd_defn}
    V_{n,k,D_n}
    :=
    \frac{1}{D_n}\sum_{i=1}^{D_n} \tilde v_i.
\end{equation}

Observe that
\[
V_{n,k} = \E_{\widetilde F_n}[\tilde v_i].
\]
Note that the distribution $\widetilde F_n$ is completely determined by the observed data $(\xi_1,\ldots,\xi_n)$.

Similarly, define $\widetilde F_n^*$ based on a bootstrap resample from the original data, and let $\tilde v_i^* \sim \widetilde F_n^*$. We next state the following lemma:

\begin{lemma}\label{lemma: var_bound_Fn_V}
Under Assumptions~\ref{asm: h_bounded} and~\ref{asm: sup bound on H diff}, there exists a constant $c$, independent of the data, such that
\[
\Var_{\widetilde F_n}(\tilde v_i)
\le
\frac{c}{k}
\qquad
\P\text{-a.s.}
\]
\end{lemma}

\begin{proof}
We compute
\begin{align*}
\Var_{\widetilde F_n}(\tilde v_i)
&=
\E_{\widetilde F_n}[\tilde v_i^2]
-
\E_{\widetilde F_n}[\tilde v_i]^2 \\
&=
\frac1M \sum_{i=1}^M \tilde v_i^2
-
V_{n,k}^2 \\
&=
V_n(H_k^2)
-
V_n(H_k)^2.
\end{align*}

By Assumptions~\ref{asm: h_bounded} and~\ref{asm: sup bound on H diff}, using Lemma~\ref{lemma: var_bound_Tk}, there exists $c_1$ independent of the data such that
\[
V_n(H_k^2) - V_n(H_k)^2
=
\Var_{P_n}(H_k)
\le
\frac{c_1}{k}.
\]

This completes the proof.
\end{proof}

We next state and prove an equivalent version of Lemma \ref{lemma: subsampled D_n} for the V-statistic:

\begin{lemma}\label{lemma: subsampled_Dn_V}
Suppose Assumptions \ref{asm: h_bounded}--\ref{asm: D_n growth} hold. Consider $\epsilon$ as defined in Assumption~\ref{asm: rate_of_k}. Then we have:
\begin{enumerate}
    \item $\sqrt n\, k^{\epsilon}\,|V_{n,k,D_n} - V_{n,k}| = o_p(1)$.
    \item For any subsequence along which $k^{2+\epsilon}\Var(g_k(\xi))>1$, we have
    \begin{equation}\label{eq:V_subsample_part2}
        \E^*\!\left[\left(\frac{V_{n,k,D_n} - V_{n,k}}{\sigma^*_{n,k}}\right)^2\right] \xrightarrow[]{p} 0.
    \end{equation}
    \item For any subsequence along which $k^{2+\epsilon}\Var(g_k(\xi))>1$, we have
    \begin{equation}\label{eq:V_subsample_part3}
        \E^*\!\left[\left(\frac{V^*_{n,k,D_n} - V^*_{n,k}}{\sigma^*_{n,k}}\right)^2\right] \xrightarrow[]{p} 0.
    \end{equation}
\end{enumerate}
where $\sigma^*_{n,k}$ is as defined in Theorem~\ref{thm: growing_k_main_theorem}.
\end{lemma}

\begin{proof}
Let $\widetilde F_n$ be the uniform distribution on
$\{H_k(\xi_{i_1},\ldots,\xi_{i_k}) : (i_1,\ldots,i_k)\in\{1,\ldots,n\}^k\}$,
and let $\tilde v_1,\ldots,\tilde v_{D_n}\stackrel{\text{i.i.d.}}{\sim}\widetilde F_n$ so that
$V_{n,k,D_n}=\frac1{D_n}\sum_{i=1}^{D_n}\tilde v_i$ and $V_{n,k}=\E_{\widetilde F_n}[\tilde v_i]$.
Then
\begin{equation}\label{eq:V_subsample_mse}
\E^*\!\left[(V_{n,k,D_n}-V_{n,k})^2\right]
=
\E^*\!\left[\left(\frac{1}{D_n}\sum_{i=1}^{D_n}(\tilde v_i-V_{n,k})\right)^2\right]
=
\frac{1}{D_n}\Var_{\widetilde F_n}(\tilde v_i).
\end{equation}
By Lemma~\ref{lemma: var_bound_Fn_V}, $\Var_{\widetilde F_n}(\tilde v_i)\lesssim 1/k$ $\P$-a.s., hence
\begin{equation}\label{eq:V_subsample_basic}
\E^*\!\left[(V_{n,k,D_n}-V_{n,k})^2\right] \lesssim \frac{1}{kD_n}
\qquad \P\text{-a.s.}
\end{equation}

\begin{enumerate}
\item Using \eqref{eq:V_subsample_basic},
\[
\E\!\left[nk^{2\epsilon}(V_{n,k,D_n}-V_{n,k})^2\right]
=
nk^{2\epsilon}\E\!\left[\E^*\!\left[(V_{n,k,D_n}-V_{n,k})^2\right]\right]
\lesssim
\frac{n}{D_n}\cdot\frac{k^{2\epsilon}}{k}
\to 0,
\]
by Assumption~\ref{asm: D_n growth}. This implies $\sqrt n\,k^\epsilon|V_{n,k,D_n}-V_{n,k}|=o_p(1)$.

\item On any subsequence with $k^{2+\epsilon}\Var(g_k(\xi))>1$, we have
\[
\E^*\!\left[\left(\frac{V_{n,k,D_n}-V_{n,k}}{\sigma^*_{n,k}}\right)^2\right]
=
\frac{1}{(\sigma^*_{n,k})^2}\E^*\!\left[(V_{n,k,D_n}-V_{n,k})^2\right]
\lesssim
\frac{1}{(\sigma^*_{n,k})^2}\cdot \frac{1}{kD_n}.
\]
Rewriting,
\[
\frac{1}{(\sigma^*_{n,k})^2}\cdot \frac{1}{kD_n}
=
\left(\frac{n(\sigma^*_{n,k})^2}{k^2\Var(g_k(\xi))}\right)^{-1}
\cdot
\frac{n}{k^{2+\epsilon}\Var(g_k(\xi))D_n}
\cdot
\frac{1}{k^{1-\epsilon}}.
\]
By Theorem~\ref{thm: growing_k_main_theorem}, $\frac{n(\sigma^*_{n,k})^2}{k^2\Var(g_k(\xi))}\xrightarrow{p}1$,
and by Assumption~\ref{asm: D_n growth} together with $k^{2+\epsilon}\Var(g_k(\xi))>1$,
the right-hand side converges to $0$ in probability, proving \eqref{eq:V_subsample_part2}.

\item Let $\widetilde F_n^*$ be the uniform distribution on
$\{H_k(\xi_{i_1}^*,\ldots,\xi_{i_k}^*) : (i_1,\ldots,i_k)\in\{1,\ldots,n\}^k\}$,
and let $\tilde v_1^*,\ldots,\tilde v_{D_n}^*\stackrel{\text{i.i.d.}}{\sim}\widetilde F_n^*$ so that
$V_{n,k,D_n}^*=\frac1{D_n}\sum_{i=1}^{D_n}\tilde v_i^*$ and $V_{n,k}^*=\E_{\widetilde F_n^*}[\tilde v_i^*]$.
Conditioning on $\widetilde F_n^*$,
\[
\E^*\!\left[(V_{n,k,D_n}^*-V_{n,k}^*)^2\mid \widetilde F_n^*\right]
=
\frac{1}{D_n}\Var_{\widetilde F_n^*}(\tilde v_i^*).
\]
Moreover,
\[
\Var_{\widetilde F_n^*}(\tilde v_i^*)
=
V_n^*(H_k^2)-\bigl(V_n^*(H_k)\bigr)^2
=
\Var_{P_n^*}(H_k).
\]
Using Lemma~\ref{lemma: var_bound_Tk} (applied under $P_n^*$),
$\Var_{P_n^*}(H_k)\lesssim 1/k$ $\P$-a.s., hence
\[
\E^*\!\left[(V_{n,k,D_n}^*-V_{n,k}^*)^2\right]
=
\E^*\!\left[\E^*\!\left[(V_{n,k,D_n}^*-V_{n,k}^*)^2\mid \widetilde F_n^*\right]\right]
\lesssim
\frac{1}{kD_n},
\qquad \P\text{-a.s.}
\]
The remainder follows identically to part (2), yielding \eqref{eq:V_subsample_part3}.
\end{enumerate}
\end{proof}

We finally proceed to prove Theorem \ref{thm: V_stat_CI}:

\begin{proof}[Proof of Theorem~\ref{thm: V_stat_CI}]
The proof follows the same structure as in Corollaries~\ref{corr: growing_k_main_CI}
and Theorem ~\ref{thm: growing_k_main_CI_D_n}. We first verify the following statements, for any subsequence along which $k^{2+\epsilon}\V(g_k(\xi))>1$.

\begin{enumerate}
    \item \textbf{Asymptotic normality.}
    The incomplete V-statistic satisfies
    \begin{equation}\label{eqn: thm4_CLT}
        \frac{\sqrt{n}\,\bigl(V_{n,k,D_n}-W_k\bigr)}
        {k\sqrt{\Var(g_k(\xi))}}
        \xrightarrow{d} N(0,1).
    \end{equation}
    \begin{proof}[Proof of Equation (\ref{eqn: thm4_CLT})] Using \ref{eqn: corr_2_main_CLT}, it is enough to show:
    \begin{equation}
        \frac{\sqrt{n}\,\bigl(V_{n,k,D_n}-U_{n,k,D_n}\bigr)}
        {k\sqrt{\Var(g_k(\xi))}} \xrightarrow[]{p} 0
    \end{equation}
      Again,
     \begin{equation*}
\frac{\sqrt{n}\,\bigl(V_{n,k,D_n}-U_{n,k,D_n}\bigr)}
{k\sqrt{\Var(g_k(\xi))}}
=
\underbrace{
\frac{\sqrt{n}\,\bigl(V_{n,k,D_n}-V_{n,k}\bigr)}
{k\sqrt{\Var(g_k(\xi))}}
}_{:= A}
+
\underbrace{
\frac{\sqrt{n}\,\bigl(V_{n,k}-U_{n,k}\bigr)}
{k\sqrt{\Var(g_k(\xi))}}
}_{:=B}
+
\underbrace{
\frac{\sqrt{n}\,\bigl(U_{n,k}-U_{n,k,D_n}\bigr)}
{k\sqrt{\Var(g_k(\xi))}}
}_{:=C}.
\end{equation*}
Using Lemma \ref{lemma: subsampled_Dn_V} (1), $A = o_p(1)$. From the proof of Theorem \ref{thm: growing_k_main_CI_D_n}, $C = o_p(1)$. Finally, using Theorem \ref{thm: diff_U_V}, $$B = O(k^{2+\epsilon/2}/n)\cdot\frac{\sqrt n}{\sqrt{k^{2+\epsilon }\V(g_k(\xi))}} = o(1)$$
which completes the proof.
    \end{proof}

    \item \textbf{Conditional bootstrap consistency.}
    For every $x\in\mathbb R$,
    \begin{equation}\label{eqn: thm_4_conditional_CLT}
        \pr^*\!\left(
        \frac{V^*_{n,k,D_n}-V_{n,k,D_n}}
        {\sigma^*_{n,k}}
        \le x
        \right)
        \xrightarrow{p}
        \Phi(x),
    \end{equation}
   \begin{proof}[Proof of \eqref{eqn: thm_4_conditional_CLT}]
Decompose
\begin{align*}
\frac{V^*_{n,k,D_n}-V_{n,k,D_n}}{\sigma^*_{n,k}}
&=
\underbrace{\frac{V^*_{n,k,D_n}-V^*_{n,k}}{\sigma^*_{n,k}}}_{A_n^*}
+
\underbrace{\frac{V^*_{n,k}-V_{n,k}}{\sigma^*_{n,k}}}_{B_n^*}
+
\underbrace{\frac{V_{n,k}-V_{n,k,D_n}}{\sigma^*_{n,k}}}_{C_n}.
\end{align*}

By Lemma~\ref{lemma: subsampled_Dn_V}(3),
\[
\E^*\!\left[(A_n^*)^2\right]\xrightarrow{p}0.
\]
Similarly, by Lemma~\ref{lemma: subsampled_Dn_V}(2),
\[
\E^*\!\left[(C_n)^2\right]\xrightarrow{p}0
\]
Thus, using conditional Slutsky's Theorem \ref{thm: conditional_Slutsky}, it is enough to  show:
\begin{equation}\label{eq:completeV_bootstrap_CLT}
\pr^*\!\left(
\frac{V^*_{n,k}-V_{n,k}}{\sigma^*_{n,k}}
\le x
\right)
\xrightarrow{p}
\Phi(x).
\end{equation}
Again, we write:
\begin{align*}
\frac{V^*_{n,k}-V_{n,k}}{\sigma^*_{n,k}}
&=
\frac{V^*_{n,k}-U^*_{n,k}}{\sigma^*_{n,k}}
+
\frac{U^*_{n,k}-U_{n,k}}{\sigma^*_{n,k}}
+
\frac{U_{n,k}-V_{n,k}}{\sigma^*_{n,k}}
\end{align*}

where, due to Theorem \ref{thm: growing_k_main_theorem}, we have
\begin{equation}
\pr^*\!\left(
\frac{U^*_{n,k}-U_{n,k}}{\sigma^*_{n,k}}
\le x
\right)
\xrightarrow{p}
\Phi(x).
\end{equation}

Now, by eqn \ref{eqn: h_bounded_P} in Assumption \ref{asm: h_bounded} and Theorem \ref{thm: diff_U_V}, we have: 
\begin{align*}
    \frac{U_{n,k}-V_{n,k}}{\sigma^*_{n,k}} &= \sqrt n k^{\epsilon/2}(U_{n,k}-V_{n,k})\cdot\sqrt{\frac{k^2\V(g_k(\xi))}{ n{\sigma^*_{n,k}}^2}}\cdot \frac{1}{\sqrt{k^{2+\epsilon}\V(g_k(\xi))}}\xrightarrow[]{p}0
\end{align*}
Similarly, by eqn \ref{eqn: h_bounded_Pn} in Assumption \ref{asm: h_bounded} and Theorem \ref{thm: diff_U_V}, we have:
\begin{align*}
   \P^*\left(\frac{U^*_{n,k}-V^*_{n,k}}{\sigma^*_{n,k}}>\delta\right)\xrightarrow[]{p}0
\end{align*}
\end{proof}
Thus, by Theorem \ref{thm: conditional_Slutsky}, we get eqn \ref{eq:completeV_bootstrap_CLT}, which in turn proves \ref{eqn: thm_4_conditional_CLT}
\end{enumerate}

The remainder of the proof is similar to that of Corollary \ref{corr: growing_k_main_CI}.
\end{proof}

\subsection{Proofs of Section \ref{sec:Applications}}

\begin{proof}[Proof of Proposition~\ref{prop: SO_bagging_validity}]
We verify the assumptions one by one.

\begin{itemize}
    \item \textbf{Assumption~\ref{asm: h_bounded}.}
    By Assumption~\ref{asm: h_bounded_SO}, there exists $M<\infty$ such that
    $|h(x,\xi)| \le M$ for all $x\in\mathcal X$ and $\P$-a.s.\ $\xi$.
    Hence, for any $(\xi_1,\ldots,\xi_{k})$,
    \[
    \left|
    H_{k}(\xi_1,\ldots,\xi_{k})
    \right|
    =
    \left|
    \min_{x\in\mathcal X}
    \frac{1}{k}\sum_{i=1}^{k} h(x,\xi_i)
    \right|
    \le M
    \qquad \P\text{-a.s.}
    \]
    Thus Assumption~\ref{asm: h_bounded} holds with $K_0=M$.

    Under the bootstrap law $\P^*$ induced by $\P_n$, the same bound holds since
    $\xi_i^*$ take values in $\{\xi_1,\ldots,\xi_n\}$, implying
    \[
    \P^*\!\left(
    \left|H_{k}(\xi_1^*,\ldots,\xi_{k}^*)\right| \le M
    \right)=1
    \quad \P\text{-a.s.}
    \]

        \item \textbf{Assumption \ref{asm: sup bound on H diff} (Bounded sensitivity).}

    Fix $1 \le c \le k$ and consider replacing the first $c$ arguments.For any $\varepsilon>0$, let $x_\varepsilon$ be an $\varepsilon$-optimal
    solution for $H_k(\xi_1,\ldots,\xi_k)$.
        \begin{align*}
&
H_k(\xi_1',\ldots,\xi_c',\xi_{c+1},\ldots,\xi_k)
-
H_k(\xi_1,\ldots,\xi_c,\xi_{c+1},\ldots,\xi_k)
\\
=\ &
\min_{x\in\mathcal X}
\left\{
\frac1k \sum_{i=1}^c h(x,\xi_i')
+
\frac1k \sum_{i=c+1}^k h(x,\xi_i)
\right\}
-
\min_{x\in\mathcal X}
\left\{
\frac1k \sum_{i=1}^c h(x,\xi_i)
+
\frac1k \sum_{i=c+1}^k h(x,\xi_i)
\right\}
\\
\leq\ &
\min_{x\in\mathcal X}
\left\{
\frac1k \sum_{i=1}^c h(x,\xi_i')
+
\frac1k \sum_{i=c+1}^k h(x,\xi_i)
\right\}
-
\left(
\frac1k \sum_{i=1}^c h(x_\varepsilon,\xi_i)
+
\frac1k \sum_{i=c+1}^k h(x_\varepsilon,\xi_i)
\right)
+ \varepsilon
\\
\leq\ &
\left(
\frac1k \sum_{i=1}^c h(x_\varepsilon,\xi_i')
+
\frac1k \sum_{i=c+1}^k h(x_\varepsilon,\xi_i)
\right)
-
\left(
\frac1k \sum_{i=1}^c h(x_\varepsilon,\xi_i)
+
\frac1k \sum_{i=c+1}^k h(x_\varepsilon,\xi_i)
\right)
+ \varepsilon
\\
=\ &
\frac1k \sum_{i=1}^c
\bigl(
h(x_\varepsilon,\xi_i')
-
h(x_\varepsilon,\xi_i)
\bigr)
+ \varepsilon
\\
\leq\ &
\frac1k \sum_{i=1}^c
\sup_{x\in\mathcal X}
\bigl|
h(x,\xi_i')
-
h(x,\xi_i)
\bigr|
+ \varepsilon
\end{align*}
By symmetry, and since $\varepsilon$ is arbitrary, we have:
\begin{equation}\label{eq:HK_sensitivity_bound}
\bigm|H_k(\xi_1',\ldots,\xi_c',\xi_{c+1},\ldots,\xi_k)
-
H_k(\xi_1,\ldots,\xi_c,\xi_{c+1},\ldots,\xi_k)\bigm|\  \leq\  \frac1k \sum_{i=1}^c
\sup_{x\in\mathcal X}
\bigl|
h(x,\xi_i')
-
h(x,\xi_i)
\bigr|    
\end{equation}

Under Assumption~\ref{asm: h_bounded_SO}, $|h(x,\xi)|\le M$ for all $x\in\mathcal X$ and $\P$-a.s.\ $\xi$, hence
\[
\sup_{x\in\mathcal X}\bigl|h(x,\xi_i')-h(x,\xi_i)\bigr|
\le
\sup_{x\in\mathcal X}\bigl(|h(x,\xi_i')|+|h(x,\xi_i)|\bigr)
\le 2M,
\qquad i=1,\ldots,c.
\]
Plugging this into \eqref{eq:HK_sensitivity_bound} yields
\[
\left|
H_k(\xi_1',\ldots,\xi_c',\xi_{c+1},\ldots,\xi_k)
-
H_k(\xi_1,\ldots,\xi_k)
\right|
\le
\frac{2Mc}{k},
\qquad \P\text{-a.s.}
\]
Therefore Assumption \ref{asm: sup bound on H diff} holds with $K_1:=2M$, i.e.,
\[
\P\!\left(
\left|
H_k(\xi_1',\ldots,\xi_c',\xi_{c+1},\ldots,\xi_k)
-
H_k(\xi_1,\ldots,\xi_k)
\right|
\le K_1\frac{c}{k}
\right)=1.
\]

The bootstrap statement follows identically, since the same bound
$|h(x,\xi)|\le M$ applies to $\xi_i^*,\xi_i^{*'}\sim\P_n$. Hence
\[
\P^*\!\left(
\left|
H_k(\xi_1^{*'},\ldots,\xi_c^{*'},\xi_{c+1}^*,\ldots,\xi_k^*)
-
H_k(\xi_1^*,\ldots,\xi_k^*)
\right|
\le K_1\frac{c}{k}
\right)=1,
\qquad \P\text{-a.s.}
\]
This completes the verification of Assumption \ref{asm: sup bound on H diff}.

\end{itemize}
\end{proof}

\begin{proof}[Proof of Proposition \ref{prop:rf_assump23}]
Fix \(k\) and \(x\). Write
\[
L_k(x)=:L,\qquad N_k(x)=:N=|L|.
\]

\paragraph{Verification of Assumption~\ref{asm: h_bounded}.}
By Assumption \ref{asm:rf_leaf_size}(1), \(|Y_i|\le M\) almost surely for all \(i\). Hence
\[
|H_k(x;Z_1,\ldots,Z_k)|
=
\left|
\frac{1}{N}\sum_{i\in L} Y_i
\right|
\le
\frac{1}{N}\sum_{i\in L}|Y_i|
\le M
\qquad \text{a.s.}
\]
Thus Assumption~\ref{asm: h_bounded} holds with \(K_0=M\).

The bootstrap statement follows identically. Indeed, if \(Z_1^*,\ldots,Z_k^*\) are drawn under the bootstrap law induced by \(P_n\), then each bootstrap response still satisfies \(|Y_i^*|\le M\), and therefore
\[
\left|H_k(x;Z_1^*,\ldots,Z_k^*)\right|
\le M
\qquad P^*\text{-a.s.}
\]
So Assumption~\ref{asm: h_bounded} is verified.

\paragraph{Verification of Assumption~\ref{asm: sup bound on H diff}.}
Now replace the first \(c\) training observations, obtaining
\[
Z_1',\ldots,Z_c',Z_{c+1},\ldots,Z_k,
\]
and let the resulting leaf index set and leaf size be
\[
L_k'(x)=:L',\qquad N_k'(x)=:N'=|L'|.
\]
Define the corresponding predictor
\[
H_k'(x)
:=
H_k(x;Z_1',\ldots,Z_c',Z_{c+1},\ldots,Z_k)
=
\frac{1}{N'}\sum_{i\in L'} Y_i'.
\]
Here \(Y_i'=Y_i\) for \(i>c\). Also, by Assumption \ref{asm:rf_leaf_size}(3),
\[
N\ge \lambda k,\qquad N'\ge \lambda k
\qquad \text{a.s.}
\]
and by Assumption \ref{asm:rf_leaf_size}(4),
\[
|L\triangle L'|\le c
\qquad \text{a.s.}
\]

Let \(S:=L\cap L'\). Then
\[
H_k(x)-H_k'(x)
=
\frac{1}{N}\sum_{i\in L}Y_i
-
\frac{1}{N'}\sum_{i\in L'}Y_i'
\]
can be written as
\[
H_k(x)-H_k'(x)
=
\frac{1}{N}\sum_{i\in S}(Y_i-Y_i')
+\frac{1}{N}\sum_{i\in L\setminus L'}Y_i
-\frac{1}{N}\sum_{i\in L'\setminus L}Y_i'
+\left(\frac{1}{N}-\frac{1}{N'}\right)\sum_{i\in L'}Y_i'.
\]
We bound the four terms separately.

For the first term, note that \(Y_i'=Y_i\) for all \(i>c\), so only indices among
\(\{1,\ldots,c\}\) can contribute:
\[
\left|
\frac{1}{N}\sum_{i\in S}(Y_i-Y_i')
\right|
=
\left|
\frac{1}{N}\sum_{i\in S\cap\{1,\ldots,c\}}(Y_i-Y_i')
\right|
\le
\frac{1}{N}\cdot 2M \cdot c
\le
\frac{2M}{\lambda}\frac{c}{k}.
\]

For the second and third terms together, using \(|Y_i|\le M\) and \(|Y_i'|\le M\),
\[
\left|
\frac{1}{N}\sum_{i\in L\setminus L'}Y_i
-\frac{1}{N}\sum_{i\in L'\setminus L}Y_i'
\right|
\le
\frac{M}{N}\bigl(|L\setminus L'|+|L'\setminus L|\bigr)
=
\frac{M}{N}|L\triangle L'|
\le
\frac{M}{\lambda}\frac{c}{k}.
\]

For the last term,
\[
\left|
\left(\frac{1}{N}-\frac{1}{N'}\right)\sum_{i\in L'}Y_i'
\right|
=
\frac{|N'-N|}{NN'}\left|\sum_{i\in L'}Y_i'\right|
\le
\frac{|N'-N|}{NN'}\cdot M N'
=
\frac{M|N'-N|}{N}.
\]
Since \(|N'-N|\le |L\triangle L'|\le c\), we obtain
\[
\left|
\left(\frac{1}{N}-\frac{1}{N'}\right)\sum_{i\in L'}Y_i'
\right|
\le
\frac{M}{N}c
\le
\frac{M}{\lambda}\frac{c}{k}.
\]

Combining the bounds gives
\[
|H_k(x)-H_k'(x)|
\le
\frac{2M}{\lambda}\frac{c}{k}
+
\frac{M}{\lambda}\frac{c}{k}
+
\frac{M}{\lambda}\frac{c}{k}
=
\frac{4M}{\lambda}\frac{c}{k}.
\]
Therefore Assumption~\ref{asm: sup bound on H diff} holds with
\[
K_1=\frac{4M}{\lambda}.
\]

The bootstrap version is proved in exactly the same way, since under the bootstrap law the responses remain bounded by \(M\), the leaf size is still bounded below by \(\lambda k\), and the symmetric-difference condition continues to hold. Hence
\[
P^*\!\left(
\left|
H_k(x;Z_1^{*\prime},\ldots,Z_c^{*\prime},Z_{c+1}^*,\ldots,Z_k^*)
-
H_k(x;Z_1^*,\ldots,Z_k^*)
\right|
\le
\frac{4M}{\lambda}\frac{c}{k}
\right)
=1
\qquad \text{a.s.}
\]
This verifies Assumption~\ref{asm: sup bound on H diff} under the bootstrap law as well.
\end{proof}
\end{document}